\documentclass[11pt]{article}
\usepackage{amsmath,amssymb,amsthm,amsopn,amsfonts,pdfpages,algorithmic,algorithm,dsfont}
\usepackage[top=1in,bottom=1in, left=1in, right=1in]{geometry}
\usepackage{tabularx}
\usepackage{bbm}
\usepackage{mathtools}
\usepackage{enumitem}
\usepackage[english]{babel}
\usepackage{xcolor}
\usepackage{relsize}
\usepackage{graphics}
\usepackage{subcaption}
\usepackage{url}
\usepackage{stackengine}
\usepackage{authblk}
\usepackage{verbatim}
\usepackage{bm}

\DeclareMathOperator{\EX}{\mathbb{E}}
\DeclareMathOperator{\PX}{\mathbb{P}}
\newcommand{\Ex}{\mathbb{E}}
\newcommand{\p}{\mathbb{P}}
\newtheorem{theorem}{Theorem}
\newtheorem{example}{Example}

\newtheorem{definition}{Definition}[section]
\newtheorem{remark}{Remark}

\newtheorem{lemma}{Lemma}
\newcommand{\R}{\mathbb{R}}

\newcommand{\bde}{\begin{definition}}
\newcommand{\ede}{\end{definition}}

\newcommand{\bX}{\mathbf{X}}

\newcommand{\bvr}{\boldsymbol{\varrho}}

\newcommand{\Xhat}{\widehat X}
\newcommand{\Yhat}{\widehat Y}
\newcommand{\bhx}{\widehat{\bf X}}

\newcommand{\bZ}{\mathbf{Z}}
\newcommand{\E}{\mathbb{E}}

\newcommand{\X}{\Tilde{X}_i}

\newcommand{\TP}{\widetilde{P}}
\newcommand{\fM}{\mathfrak{M}}
\newcommand{\sm}{S_{\fM}}
\newcommand{\um}{U_{\fM}}
\newcommand{\spp}{S_{\TP}}
\newcommand{\up}{U_{\TP}}
\newcommand{\TQ}{\widetilde Q}

\begin{document}
%
\title{The Importance of Being Correlated:\\
	Implications of Dependence in Joint Spectral Inference across Multiple Networks}
\author{Konstantinos Pantazis$^1$, Avanti Athreya$^2$, Jes\'us Arroyo$^3$, William N. Frost$^4$, Evan S. Hill$^4$, Vince Lyzinski$^1$\\
	\small{$^1$ University of Maryland, College Park, Department of Mathematics\\
		$^2$ Johns Hopkins University, Department of Applied Mathematics and Statistics\\
		$^3$ Texas A \& M University, College Station, TX, Department of Statistics\\
		$^4$ Rosalind Franklin University of Medicine and Science, Chicago Medical School, Chicago, IL, Department of Cell Biology and Anatomy and Center for Brain Function and Repair}}
	\maketitle
\begin{abstract}
Spectral inference on multiple networks is a rapidly-developing subfield of graph statistics. Recent work has demonstrated that joint, or simultaneous, spectral embedding of multiple independent networks can deliver more accurate estimation than individual spectral decompositions of those same networks. Such inference procedures typically rely heavily on independence assumptions across the multiple network realizations, and even in this case, little attention has been paid to the induced network correlation that can be a consequence of such joint embeddings. 
In this paper, we present a {\em generalized omnibus} embedding methodology and we provide a detailed analysis of this embedding across both independent and correlated networks, the latter of which significantly extends the reach of such procedures, and we describe how this omnibus embedding can itself induce correlation. 
This leads us to distinguish between {\em inherent} correlation---that is, the correlation that arises naturally in multisample network data---and {\em induced} correlation, which is an artifice of the joint embedding methodology. We show that the generalized omnibus embedding procedure is flexible and robust, and we prove both consistency and a central limit theorem for the embedded points. 
We examine how induced and inherent correlation can impact inference for network time series data, and we provide network analogues of classical questions such as the effective sample size for more generally correlated data. 
Further, we show how an appropriately calibrated generalized omnibus embedding can detect changes in real biological networks that previous embedding procedures could not discern, confirming that the effect of inherent and induced correlation can be subtle and transformative. 
By allowing for and deconstructing both forms of correlation, our methodology widens the scope of spectral techniques for network inference, with import in theory and practice.
\end{abstract}

\section{Introduction}

Networks and graphs, which consist of objects of interest and a vast array of possible relationships between them, arise very naturally in fields as diverse as political science (party affiliations among voters); bioinformatics (gene interactions);
physics (dimer systems); and sociology (social network analysis), to name but a few.
As such, they are a useful data structure for modeling complex interactions between different experimental entities. 
Network data, however, is qualitatively distinct from more traditional Euclidean data, and statistical inference on networks is a comparatively new discipline, one that has seen explosive growth over the last two decades.
While there is a significant literature devoted to the rigorous statistical study of single networks, multiple network inference---the analogue of the classical problem of multiple-sample Euclidean inference---is still relatively nascent.

Much recent progress in network inference has relied on extracting Euclidean representations of networks, and popular methods include spectral embeddings of network adjacency \cite{athreya_survey} or Laplacian \cite{rohe2011spectral} matrices, 
representation learning  \cite{grover2016node2vec,ribeiro2017struc2vec}, 
or Bayesian hierarchical methods \cite{durante2017nonparametric}.
Moreover, many network models \cite{Hoff2002} allow for important properties of network entities to be hidden, or {\em latent}, and posit that relationships between entities depend on these latent variables. 
Such models, known as {\em latent position networks}, have wide intuitive appeal. For instance, relationships among participants in a social network are a function of the participants' personal interests, which are typically not directly observed. 
In these cases, spectral embeddings can provide useful estimates of latent variables, effectively transforming, via eigendecompositions, a non-Euclidean inference problem into a Euclidean one.

For single latent position networks, spectrally-derived estimates of important graph parameters are well-understood, and under mild assumptions, these estimates satisfy classical notions of consistency \cite{STFP-2011,rohe2011spectral}, asymptotic normality \cite{athreya2013limit,tang_lse}, and efficiency \cite{tang_lse,tang2017asymptotically,xie2019efficient,LSM}.
More recently, spectral methods have also proven useful in multi-sample network inference, including (non)parametric estimation \cite{durante2017nonparametric,tang2018connectome}, two-sample hypothesis testing \cite{tang14:_semipar,tang14:_nonpar,asta,li2018two}, and graph matching \cite{lyzinski2015spectral,zhang2018consistent,zhang2018unseeded}.
Typically, these methods rely upon separately embedding multiple networks into a lower-dimensional Euclidean space and then aligning the embeddings via Procrustes analysis \cite{gower_procrustes} or point set registration methods \cite{myronenko2010point}. 
An important issue in multi-sample inference, however, is the use of {\em multiple} networks both for improved estimation of underlying model parameters and for more streamlined testing across several populations of networks. To this end, a number of recent papers are dedicated to the development of novel techniques for simultaneously embedding several networks into a common Euclidean space, employing spectral graph techniques \cite{levin_omni_2017,nielsen2018multiple,wang2019joint,arroyo2019inference}, tensor factorizations \cite{zhang2018tensor,zhang2019cross,jing2020community}, 
multilayer network decompositions \cite{kivela2014multilayer,paul2016consistent,paul2020spectral}, and nonparametric Bayesian algorithms \cite{durante2017nonparametric,durante2014bayesian}.

While multisample joint embedding methods allow for accurate graph inference and are often superior to individual separate embeddings \cite{levin_omni_2017, arroyo2019inference}, there are a number of potential pitfalls in joint embeddings.  In particular, network statisticians must confront issues of noisy vertex alignments across graphs \cite{lyzinski16:_infoGM}; 
large, high-rank matrices that arise in a joint embedding \cite{draves2020bias}; the relationship between individual network sparsity and the signal in a joint embedding; and the {\em induced} correlation across estimates that arise from the joint embedding, the last of which is inevitable in any simultaneous embedding procedure. What is more, virtually all existing procedures for multisample network inference rely, like their classical analogues, on an assumption of independence across network realizations.  In this sense, existing methodology is ill-equipped to handle, at least in a principled manner, the {\em inherent} network correlation---for example, the natural and unavoidable correlation across edges in a network time series---to say nothing of the additional correlation induced by any dimension-reduction procedure.

This paper is devoted to broadening spectral analysis to account for both types of correlation, and to understanding how the correlation induced by joint spectral procedures can mask or amplify important signal. We focus on a generalization of the omnibus multiple graph embedding procedure (OMNI) of \cite{levin_omni_2017}, in which multiple networks are simultaneously embedded into a single lower-dimensional subspace, with a distinct representation for each vertex across the networks.
The work of \cite{levin_omni_2017} considers this problem in the case where the network realizations themselves are independent, though even when independent network samples are jointly  embedded, correlation across the embedded point clouds is automatically induced by the OMNI procedure (this is the price we pay to circumvent the pairwise Procrustes/registration analysis necessary in separate embedding settings), to say nothing of the impact of OMNI in preserving or masking the a priori present inherent correlation across networks.
It is natural, then, to seek to adapt the OMNI embedding technique in order to preserve in the embedded (independent or correlated) graphs the same correlation that would be present if edge-wise inherently correlated networks are embedded separately and then aligned.
This would allow for the jointly-embedded networks to be a more appropriate proxy in embedding space for sequences of graphs with complex dependency structures. 




To understand these phenomena more rigorously, we anchor our analysis in a specific class of latent position random graphs, the {\em random dot product graph} (RDPG; see \cite{young2007random}). 
Random dot product graphs have proven to be a theoretically tractable family of latent position networks \cite{athreya_survey} suitable for modeling a host of complex real-data networks \cite{tang_lse,priebe2019two,patsolic2020vertex}. In Section \ref{ssec:RDPG}, we formulate several models for inherent correlation across a series of random dot product graphs, and, thereafter, examine the impact of a joint spectral embedding of such a collection of multiple networks.
Given $m$ adjacency matrices of multiple independent, $n$-vertex, aligned RDPGs,
the OMNI embedding of \cite{levin_omni_2017} and its more general counterpart---the genOMNI embedding we define here---provide $m$ distinct representations for the latent attributes of each of the $n$ vertices in the collection of graphs. 
This permits both consistent estimation of underlying RDPG latent positions (in which the omnibus embedding is empirically shown to be competitive with embedding the sample mean of the respective adjacencies) and inference {\em across} the latent positions, including testing, classification, and change-point detection.
The generalized omnibus embedding jointly embeds the collection of graphs, though by construction the $m$ distinct $n \times d$ blocks in the genOMNI embedding (there are $m$ such blocks) are necessarily correlated. 
This is not a unique feature of the generalized omnibus methodology; all joint embedding procedures typically induce correlation across network pairs. 
What is unique, at least to our knowledge, about the genOMNI embedding is that the dual impact of induced and inherent correlation in the embedding is theoretically tractable.

In light of this, the major contributions of this paper are as follows. The first is an entirely novel treatment of method-induced correlation in the output of joint embedding procedures, and the second is the development, through genOMNI, of a flexible joint embedding that can not only reproduce complex correlation in the embedded space, but can also accommodate inherently correlated network data while retaining important theoretical performance guarantees of consistency and asymptotic normality.
By comparing the omnibus embedding of independent graphs to the separate embeddings of correlated (and subsequently Procrustes-aligned) graphs, we can explicitly capture the level of correlation the joint OMNI embedding induces, in the limit, across independent networks.
This, in turn, motivates the creation of the generalized omnibus embedding (Definition \ref{def:genOMNI}), which produces more complex correlation structure in the embedded space, enabling higher-fidelity application of the omnibus methodology in real data.
This replication of more complicated correlation structure renders the generalized omnibus embedding suitable for inference on network time series, because it can reproduce, via realizations of independent networks, the correlation that is an important component of a time series. 

The core result underlying the above is a central limit theorem (Theorems \ref{theorem:ind_lat_classical_omni} and \ref{thm:genOMNI}) for the row-wise residuals of the estimated latent positions in both inherently correlated or independent RDPGs in a generalized omnibus embedding framework.
In addition, 
we are able to precisely characterize the dual effects of inherent and induced correlation on the limiting covariance structure across the embedded networks. In Sections \ref{sec:damp} and \ref{sec:experiments-simulations2}, we show how the weights of genOMNI can be adapted for certain specific inherent correlation structures.
As an illustration of the power of the more nuanced embedding correlation enabled by the genOMNI setting, we present in Section \ref{sec:aplysia} an analysis of a motor program time series of networks in the brain of the marine mollusk {\em Aplysia californica}. 
The classical omnibus embedding on this time series homogenizes the inherent correlation across the time-series, 
effectively obscuring important network changes corresponding directly to transitions in animal behavior (from stimulus to gallop and crawl). Our genOMNI embedding, however, is flexible enough to permit different weightings of networks over time, and this more general joint inference procedure captures exactly the signal the earlier omnibus embedding misses. Figure \ref{fig:omniaplysia} in Section \ref{sec:aplysia} and Figure \ref{fig:domni_aplysia} in Section  \ref{sec:damp} demonstrate this contrast in inferential accuracy between the two.

Lastly, in Section \ref{sec:effective-sample-size}, we further show, with theory, simulated and real data examples, how inherent and induced correlation across networks impact the effective sample size for subsequent inference tasks in the joint embedded space. This provides a network analogue of the classical statistical challenge of quantifying, via a comparison of sample sizes, the extent to which dependence in data can impact inference. 
In sum, our generalized omnibus embedding and accompanying correlation analysis form a tractable, scalable inference methodology that can be applied to independent and correlated data, carries straightforward theoretical guarantees, has demonstrable empirical utility, and correctly identifies important and subtle network changes that its predecessors miss.

\vspace{2mm}

\noindent\textbf{Notation:} For a positive integer $n$, we let $[n]=\{1,2,\cdots,n\}$, let $0_n$ be the zero $n\times n$ matrix, and let $\vec{1}_n\in\R^{n}$ (resp. $J_n\in\R^{n\times n}$) be the vector (resp. matrix) with all entries identically equal to one.
The set of $n\times n$ real orthogonal matrices is denoted by $\mathcal{O}_n$.
We represent a simple (no self-loops or multiple edges), un-weighted and un-directed graph as the ordered pair $G= (V,E)$, 
where $V=[n]$ represents the set of nodes and $E\subset \binom{n}{2}$ the set of edges of the graph; we 
denote the set of all $n$-vertex labeled graphs via $\mathcal{G}_n$. 
For the graph $G=(V,E)$, we will denote its adjacency matrix via $A\in\{0,1\}^{n\times n}$; i.e., $A_{ij}$ is equal to 1 if there exists an edge between nodes $i$ and $j$ in $G$, and 0 otherwise. 
Where there is no danger of confusion, we will often refer to a graph $G$ and its adjacency matrix $A$ interchangeably. The Kronecker product is denoted by $\otimes$ and the direct sum by $\oplus$. 
Finally, the symbols $\|\cdot\|_F,\,\|\cdot\|$, and $\|\cdot\|_{2\rightarrow \infty}$  correspond to the Frobenius, spectral and two-to-infinity norms respectively.
\section{Background}
\label{sec:bg}
In this section, we will introduce the modeling and spectral embedding frameworks that we build our theory and methods upon.

\subsection{Random Dot Product Graphs}
\label{sec:rdpg}

The theoretical developments to follow are situated in the context of the \emph{random dot product graph} (as mentioned above, abbreviated RDPG) model of \cite{young2007random}. 
Random dot product graphs are a special case of the more general {\em latent position random graphs} (abbreviated LPGs) of \cite{hoff_raftery_handcock}. Every vertex in a latent position random graph has associated to it a (typically unobserved) {\em latent position}, itself an object belonging to some (often Euclidean) space $\mathcal{X}$.  Probabilities of an edge between two vertices $i$ and $j$, $p_{ij}$, are then a function $\kappa(\cdot,\cdot): \mathcal{X} \times \mathcal{X} \rightarrow [0,1]$ (known as the {\em link function}) of their associated latent positions $(x_i, x_j)$. Thus $p_{ij}=\kappa(x_i, x_j)$, and edges between vertices arise independently of one another. Given these probabilities, the entries $A_{ij}$ of the adjacency matrix $A$ are conditionally independent Bernoulli random variables with success probabilities $p_{ij}$. We consolidate these probabilities into a matrix $P=(p_{ij})$, and we write $A \sim P$ to denote this relationship.
 
 In a $d$-dimensional random dot product graph, the latent space is an appropriately-constrained subspace of $\mathbb{R}^d$, and the link function is simply the dot product of the two latent $d$-dimensional vectors. Random dot product graphs are often divided into two types: those in which the latent positions are fixed, and those in which the latent positions are themselves random. Specifically, we consider the case in which the latent position $X_i \in \mathbb{R}^d$ for vertex $i$ is drawn from some distribution $F$ on $\mathbb{R}^d$, and we further assume that the latent positions for each vertex are drawn independently and identically from this distribution $F$.
 Random dot product graphs have proven to be a tractable and useful model for low-rank latent position networks, and variants of the RDPG model have recently emerged that extend the framework to allow for modeling more complex network topologies \cite{rubin-delanchy_tang_priebe_grdpg,tang_lse}.
 
\begin{definition}[$d$-dimensional RDPG]
\label{def:RDPG}
Let $F$ be a distribution on a set $\mathcal{X}\in \R^d $ satisfying $\langle x,x' \rangle \in [0,1]$ for all $x,x'\in \mathcal{X}$.
Let $X_1,X_2,\cdots,X_n\sim F$ be i.i.d. random variables distributed via $F$,
and let $P=\bX\bX^T$, where  
$\bX=[X_1^T| X_2^T| \cdots| X_n^T]^T\in\R^{n\times d}$.
Let $A$ be a symmetric, hollow adjacency matrix with above diagonal entries distributed via
\begin{align}
\label{eq:rdpg}
P(A|\bX)=\prod_{i<j}(X_i^TX_j)^{A_{ij}}(1-X_i^TX_j)^{1-A_{ij}};
\end{align}
i.e., conditioned on $\bX$ the above diagonal entries are independent Bernoulli random variables with success probabilities provided by the corresponding above diagonal entries in $P$.
The pair $(A,\bX)$ is then said to be an instantiation of a $d$-dimensional Random Dot Product Graph with distribution $F$, denoted $(A,\bX)\sim\mathrm{RDPG}(F,n)$. 
\end{definition}

\noindent Note that there is a rotational  non-identifiability inherent to the RDPG model.
Indeed, if $\mathbf{Y}=\bX W$ for $W\in\mathcal{O}_d$, then the distribution over graphs induced by Eq.\@ 
(\ref{eq:rdpg}) by $\bX$ and $\mathbf{Y}$ are identical; i.e., $\p(A|\bX)=\p(A|\mathbf{Y})$ for all $A$.
As inference in the RDPG setting often proceeds by first estimating the latent positions $\bX$, which can only be done up to a rotation factor, this model is not generally suitable for inference tasks that are not rotationally invariant.

\begin{remark}
We note that a generalization of the $\mathrm{RDPG}$ model has recently been developed, namely the \emph{Generalized Random Dot Product Graph} of \cite{rubin-delanchy_tang_priebe_grdpg}, which allows for modeling latent position graphs where $P$ is not necessarily positive definite (this, for example, allows for disassortative connectivity behavior in stochastic block model networks modeled via the $\mathrm{RDPG}$ framework which is not possible under Definition \ref{def:RDPG}).
In more general latent space models (see, for example, \cite{Hoff2002, lei2020network,NIPS2007_766ebcd5}), the probability of connections in the network are computed via more general similarity functions (i.e., kernels) compared to the dot product in the $\mathrm{RDPG}$. 
While we suspect that our results translate immediately to the generalized $\mathrm{RDPG}$ setting (and perhaps less immediately to the general latent space setting; see \cite{tang2012universally}), we do not pursue this further here.
\end{remark}


\subsection{Modeling Multiple Correlated RDPGs}
\label{ssec:RDPG}

Next, we present a natural extension of the RDPG model to the multiple network setting, namely the \emph{Joint Random Dot Product Graph} from \cite{levin_omni_2017}.
This model allows us to simultaneously characterize multiple graphs $(A^{(k)})_{i=1}^m$ with a common set of latent positions $\bX$ so that the collection of $A^{(k)}$ are conditionally independent given $\bX$.
\begin{definition}[Joint Random Dot Product Graph]
\label{def:JRDPG}
Let $F$ be a distribution on a set $\mathcal{X}\in \R^d $ satisfying $\langle x,x' \rangle \in [0,1]$ for all $x,x'\in \mathcal{X}$.
Let $X_1,X_2,\cdots,X_n\stackrel{i.i.d.}{\sim} F$,
and let $P=\bX\bX^T$, where  
$\bX=[X_1^T| X_2^T| \cdots| X_n^T]^T\in\R^{n\times d}$.
We say that the random graphs $(A^{(1)},A^{(2)},\cdots,A^{(m)})$ are an instantiation of a \emph{Joint
Random Dot Product Graph} model (abbreviated $\mathrm{JRDPG}$), written 
$$
(A^{(1)},A^{(2)},\cdots,A^{(m)},\bX)\sim \mathrm{JRDPG}(F,n,m)
$$
if marginally each $(A^{(k)},\bX)\sim\mathrm{RDPG}(F,n)$ and conditioned on $\bX$, the $A^{(k)}$'s are independent with distribution given by Eq.\@  \eqref{eq:rdpg}.
\end{definition}

\noindent
While this model allows for modeling multiple networks simultaneously, the conditional independence is ill suited for a number of inference tasks (e.g., time-series analysis in graphs as in \cite{wang2013locality,tang2013attribute,bhattacharyya2020consistent}) that necessitate more nuanced dependency structure across graphs.
While correlated RDPG models exist for pairs of graphs (see, for example, \cite{patsolic2020vertex}), we seek a framework that allows for (pairwise) correlation across the entire collection of $A^{(i)}$.
To do so, we assume, as in JRDPG, that each network is marginally distributed as an RDPG, and then we assign an edge-wise correlation among the network pairs.
We will then provide a few constructive methods via which such networks can be sampled.
 
\begin{definition}[Pairwise multiple edge-correlated RDPG]
\label{defn:multirhordpg}
With notation as in Definition \ref{def:JRDPG}, 
we say that the random graphs $(A^{(1)},A^{(2)},\cdots,A^{(m)})$ are an instantiation of a \emph{R-correlated Joint
Random Dot Product Graph} model, written $(A^{(1)},A^{(2)},\cdots,A^{(m)},\bX)\sim\mathrm{JRDPG}(F,n,m,R)$, if 
\begin{itemize}
    \item[i.] Marginally, $(A^{(k)},\bX)\sim\mathrm{RDPG}(F,n$) for every $k\in[m]$;
    \item[ii.] The matrix $R\in[-1,1]^{m\times m}$ is symmetric and has diagonal entries identically equal to 1.  We will write the $(k_1,k_2)$-element of $R$ via $\rho_{k_1,k_2}$.
    \item[iii.] Conditioned on $\bX$ the collection
    $$\{A^{(k)}_{i,j}\}_{k\in[m],i<j}$$ is mutually independent except that for each $\{i,j\}\in\binom{V}{2}$, we have for each $k_1,k_2\in[m]$,
    $$\mathrm{correlation}(A^{(k_1)}_{i,j},A^{(k_2)}_{i,j})=\rho_{k_1,k_2}.$$
\end{itemize}
Note that if all off-diagonal elements of $R$ are identically equal to $\rho$, then we will often write $(A^{(1)},A^{(2)},\cdots,A^{(m)},\bX)\sim\mathrm{JRDPG}(F,n,m,\rho)$.
\end{definition}
\noindent Below we will assume that $R\geq 0_m$ entry-wise, although negatively correlated graphs can be considered in Definition \ref{defn:multirhordpg}.
The pairwise multiple edge-correlated RDPG is a better candidate for modeling time series of networks and multilayer networks than the conditionally independent JRDPG model, as it allows for generating conditionally (within graph) edge-independent networks and induces correlation across networks pairwise. 
A natural extension would be to allow for correlation across edges within each network, and we are actively working on this extension; see \cite{babkin2020large} for an example of how this structure could be introduced in a model related to the stochastic blockmodel.

We next illustrate two constructions that result in RDPG$(F,n,m,R)$ networks.

\subsubsection{Forward Propagation (Sequential) Model}
\label{sec:FWD}

The \emph{Forward Propagation (Sequential) model} (abbreviated $\mathrm{JRDPG}_{\text{for}}$) fits time-varying networks with forward propagation of correlation suitable for network time-series inference tasks such as (spectral) clustering \cite{pensky2019} and anomaly detection of vertices (or networks) \cite{chen_AnDet2020} in a given time period. 
It derives the JRDPG$(F,n,m,R)$ for $R$ equal to the symmetric matrix
\[R_{for}=\begin{bmatrix}
    1&\varrho_{1,2}&\varrho_{1,2}\varrho_{2,3}&\cdots&\prod_{k=1}^{m-1}\varrho_{k,k+1}\\\varrho_{1,2}&1&\varrho_{2,3}&\cdots&\prod_{k=2}^{m-1}\varrho_{k,k+1}\\\varrho_{1,2}\varrho_{2,3}&\varrho_{2,3}&1&\cdots&\prod_{k=3}^{m-1}\varrho_{k,k+1}\\\vdots&\vdots&\vdots&\ddots&\vdots\\\prod_{k=1}^{m-1}\varrho_{k,k+1}&\prod_{k=2}^{m-1}\varrho_{k,k+1}&\prod_{k=3}^{m-1}\varrho_{k,k+1}&\cdots&1
    \end{bmatrix}\in\R^{m\times m}.\]
The formal definition is presented as follows.
\begin{definition}\emph{(Forward Propagation (Sequential) model)}
\label{def:seq2_RDPG}
With the notation as in Definition \ref{defn:multirhordpg},
 we say that the random graphs $A^{(1)},A^{(2)},\cdots,A^{(m)}$ are an instantiation of a sequential (correlated) $\mathrm{JRDPG}$ and we write $(A^{(1)},A^{(2)},\cdots,A^{(m)},\bX)\sim\mathrm{JRDPG}_{\text{for}}$ $(F,n,m,\boldsymbol{\varrho})$ if 
\begin{itemize}
\item[i.] $\boldsymbol{\varrho}\in[0,1]^{m-1}$ is a vector whose $k$-th element is denoted $\varrho_{k,k+1}$.   
\item[ii.] $(A^{(1)},A^{(2)},\cdots,A^{(m)},\bX)\sim \mathrm{JRDPG}(F,n,m,R)$ where 
for $1\leq k_1<k_2\leq m$,
$$R_{k_1,k_2}
:=\rho_{k_1,k_2}
=\prod_{k=k_1}^{k_2-1} \varrho_{k,k+1}.$$
\end{itemize}
\end{definition}

\noindent 
Practically, we can sample from $(A^{(1)},A^{(2)},\cdots,A^{(m)},\bX)\sim \mathrm{JRDPG}_{\text{for}}(F,n,m,\bvr)$ by first sampling from $(A^{(1)},\bX)\sim $RDPG$(F,n)$, and then conditional on $\bX$ and $A^{(\ell)}$ for $\ell\geq 1$, independently sampling the edges of $A^{(\ell+1)}$ according to the following scheme
\[
 \forall\{i,j\}\in \binom{V}{2},\hspace{1cm}   A^{(\ell+1)}_{ij} \sim  \left.
        \begin{cases}
            \text{Bern}(P_{ij}+\varrho_{\ell,\ell+1}(1-P_{ij}))  &\text{ if }A^{(\ell)}_{ij}=1,\\
            \text{Bern}(P_{ij}(1-\varrho_{\ell,\ell+1}))  &\text{ if }A^{(\ell)}_{ij}=0.
        \end{cases}
    \right.
\]
A straightforward induction on $\ell$ guarantees that marginally, $(A^{(\ell)},\bX)\sim$ RDPG$(F,n)$. The form of the correlation follows from the recursion that for $k_1<k_2$ (suppressing the conditioning on $\bX$ below),
\begin{align*}
\text{correlation}(A^{(k_1)}_{i,j},A^{(k_2)}_{i,j})=\varrho_{k_1,k_2}\cdot \text{correlation}(A^{(k_1)}_{i,j},A^{(k_2-1)}_{i,j}).
\end{align*}
It is also possible to model non-stationary time series of graphs (i.e., allowing for distinct (latent) distributions of $A^{(k_1)}$ and $A^{(k_2)}$ still with pairwise correlation; see \cite{het} for details). 

%
%
\subsubsection{Single Generator Model}
\label{sec:GEN}

Similarly, the \emph{Single Generator model} derives from JRDPG$(F,n,m,R)$ when the correlation matrix $R$ is equal to
    \[R_{gen}=\begin{bmatrix}
   1&\varrho_1\varrho_2&\cdots&\varrho_1\varrho_{m}\\
   \varrho_2\varrho_1&1&\cdots&\varrho_2\varrho_{m}\\
   \vdots&\vdots&\ddots&\vdots\\
   \varrho_{m}\varrho_1&\varrho_{m}\varrho_2&\cdots&1
    \end{bmatrix}\in\R^{m\times m}.\]
    Define the (generator) vector $\nu=[\varrho_1, \cdots ,\varrho_{m}]^T\in\R^{m}$, so that the correlation matrix $R_{gen}$ can be also written as $R_{gen}=\nu\nu^T+\text{diag}(I_{m}-\nu\nu^T)$. 
    It is reasonable to add the term $\text{diag}(I_{m}-\nu\nu^T)$; this ensures that the correlation of each graph with itself remains the same for all graphs and is equal to $1$.
\begin{definition}\label{def:generator_RDPG}\emph{(Single Generator Model)}
With notation as in Definition \ref{defn:multirhordpg}, 
 we say that the random graphs $A^{(0)},A^{(1)},\cdots,A^{(m)}$ are an instantiation of a multiple $\mathrm{RDPG}$ with generator matrix $A^{(0)}$ and we write  $(A^{(0)},A^{(1)},A^{(2)},\cdots,A^{(m)},\bX)\sim \mathrm{JRDPG}_{\text{gen}}(F,n,m,\nu)$ if 
\begin{itemize}
\item[i.] Marginally, $(A^{(0)},\bX)\sim\mathrm{RDPG}(F,n$);
\item[ii.] $\nu\in[0,1]^{m}$ is a nonnegative vector with entries in $[0,1]$; we denote the $k$-th entry of $\nu$ via $\nu_k=\varrho_{k}$.  
\item[iii.] $(A^{(1)},A^{(2)},\cdots,A^{(m)},\bX)\sim \mathrm{JRDPG}(F,n,m,R)$ where 
$R=\nu\nu^T+\mathrm{diag}(I_{m}-\nu\nu^T)$ so that 
for $1\leq k_1<k_2\leq m$,
$$R_{k_1,k_2}
:=\rho_{k_1,k_2}
=\varrho_{k_1}\varrho_{k_2}.$$
\end{itemize}
\end{definition}

The single generator model, JRDPG$_{gen}$, mimics the correlated Erd\H os-R\'enyi graph pairs that are a common model in the graph matching literature (see, for example, \cite{Deanon_Pedarshani,cullina2016improved}) where the pair of graphs are noisy realizations of a background network (here $A^{(0)}$).
The single generator model with $m$ networks provides a suitable framework for studying problems of aligning multiple networks.
As in the JRDPG$_{for}$ case, we can sample from $(A^{(0)},\cdots,A^{(m)},\bX)\sim \mathrm{JRDPG}_{\text{gen}}(F,n,m,\nu)$) by first sampling from $(A^{(0)},\bX)\sim $RDPG$(F,n)$), and then conditional on $A^{(0)}$ and $\bX$, independently sampling the edges of each $A^{(\ell)}$, $\ell\in[m]$, according to the following scheme
\[ \forall\{i,j\}\in\binom{V}{2},\hspace{1cm}
  A^{(\ell)}_{ij}\sim \left.
  \begin{cases}
        \text{Bern}(P_{ij}+\varrho_{\ell}(1-P_{ij}))  &\text{ if }A^{(0)}_{ij}=1,\\
        \text{Bern}(P_{ij}(1-\varrho_{\ell}))  &\text{ if }A^{(0)}_{ij}=0.
  \end{cases}
  \right.
\]
The above scheme implies that the correlation between two networks $A^{(\ell_1)},A^{(\ell_2)}$ is given by the product $\varrho_{\ell_1}\varrho_{\ell_2}$, i.e., $\text{correlation}(A^{(\ell_1)}_{i,j},A^{(\ell_2)}_{i,j})=\varrho_{\ell_1}\varrho_{\ell_2}$.

\subsection{Spectral graph embeddings}
\label{sec:omni}
One of the key inference tasks in latent position random graphs (LPGs) is to estimate the unobserved latent positions for each of the vertices based on a single observation of the adjacency matrix of a sufficiently large graph. Since the matrix of connection probabilities for an RDPG is expressible as an outer product of the matrix of true latent positions, and since the adjacency matrix $A$ can be regarded as a ``small" perturbation of $P$, the inference of properties of $P$ from an observation of $A$ is a problem well-suited to spectral graph methods, such as singular value decompositions of adjacency or Laplacian matrices. Indeed, these spectral decompositions have been the basis for a suite of approaches to graph estimation, community detection, and hypothesis testing for random dot product graphs. For a comprehensive summary of these techniques, see \cite{athreya_survey}.
  Note that the popular stochastic blockmodel (SBM) with positive semidefinite block connection probabilities can be regarded as a random dot product graph. In an SBM, there are a finite number of possible latent positions for each vertex---one for each block---and the latent position exactly determines the block assignment for that vertex. 
%
 
 For a random dot product graph in which the latent position $X_i \in \mathbb{R}^d$ for each vertex $i$, $1 \leq i \leq n$, are drawn i.i.d from some distribution $F$ on $\mathbb{R}^d$, a common graph inference task is to infer properties of $F$ from an observation of the graph alone. For example, in a stochastic block model, in which the distribution $F$ is discretely supported, we may wish to estimate the point masses in the support of $F$. In the graph inference setting, however, there are two sources of randomness that culminate in the generation of the actual graph: first, the randomness in the latent positions, and second, {\em given these latent positions}, the conditional randomness in the existence of edges between vertices. 
 
 
A rank-$d$ RDPG has a connection probability matrix $P$ that is necessarily low rank (rank $d$, regardless of the number of vertices in the graph); hence random dot product graphs can be productively analyzed using low-dimensional embeddings. 
Under mild assumptions, the adjacency matrix $A$ of a random dot product graph approximates the matrix $P=\E(A)$. To be more precise, $P=\E(A|\bX)$ 
in the sense that the spectral norm of $A-P$ can be controlled; see  \cite{oliveira2009concentration} and \cite{lu13:_spect}.
 It is reasonable to ask how close the spectrum and associated invariant subspaces of $A$ are to those of $P$. Weyl's Theorem \cite{horn85:_matrix_analy} describes how the eigenvalues of $A$ differ from those of $P$. 
 Sharp bounds on differences between the associated invariant subspaces are fewer, with the Davis-Kahan Theorem \cite{davis70, Bhatia1997,DK_usefulvariant} perhaps the best known. Because of the invariance of the inner product to orthogonal transformations, however, the RDPG exhibits a clear nonidentifiability: latent positions can be estimated only up to an orthogonal transformation.

Since $P$ is a symmetric, positive definite, rank $d$-matrix of the form $P=\bX\bX^T$,  the latent position matrix $\bX$ can be written as $\bX=U_PS_P^{1/2} W$ for some orthogonal matrix $W$, where $S_P$ is the diagonal matrix of the $d$ nonzero eigenvalues of $P$, sorted by magnitude, and $U_P$ the associated eigenvectors.
The Davis-Kahan Theorem translates spectral norm bounds on $A-P$ to projection operator bounds between $U_AU_A^T$ and $U_PU_P^T$, and, in turn, into Frobenius norm bounds between $U_A$ and a rotation of $U_P$ \cite{rohe2011spectral}.

These bounds can be sufficiently sharpened to ensure that the rows of a partial spectral decomposition of $A$, known as the {\em adjacency spectral embedding} (ASE)  are accurate estimates of the latent positions $X_i$ for each vertex.  With this in mind, we define the adjacency spectral embedding (ASE) as follows.
\begin{definition}[Adjacency Spectral Embedding (ASE)] Let $d\geq 1$ be a positive integer.  
The $d$-dimensional \emph{adjacency spectral embedding} of a graph $A\in\mathcal{G}_n$ into $\mathcal{R}^d$, denoted by $\mathrm{ASE}(A, d)$, is defined to be $\widehat \bX=U_{A}S_{A}^{1/2}$, where 
$$|A| = (A^{\top} A)^{1/2}=\left[U_A\,|\,\widetilde U_A\right]\left[S_A
\oplus \widetilde S_A\right]\left[U_A\,|\,\widetilde U_A\right]^T,$$ 
is the spectral decomposition of $|A|$, $S_A\in\mathbb{R}^{d\times d}$ is the diagonal matrix with the $d$ largest eigenvalues of $|A|$, and $U_A\in\mathbb{R}^{n\times d}$ the corresponding matrix of the $d$-largest eigenvectors.
\end{definition}
\vspace{2mm}
Now, if we define $\bhx=U_AS_A^{1/2}$, where $S_A$ is the diagonal matrix of the top $d$ eigenvalues of $(A^TA)^{1/2}$, sorted by magnitude, and the columns of $U_A$ are the associated unit eigenvectors, results in \cite{STFP-2011} and \cite{lyzinski13:_perfec} establish that, under assumptions on the spectrum of $P$, the rows $\widehat{\bX}$ are consistent estimates of the latent positions $\{X_i\}$ (up to orthogonal rotation).
Further, in \cite{athreya2013limit}, it is shown that under the RDPG, the (suitably-scaled) ASE of the adjacency matrix converges in distribution to a Gaussian mixture. 

The utility of the ASE in single graph inference points us to a natural test statistic for determining whether two random dot product graphs have the same latent positions. Namely, we can perform Procrustes alignment of two graphs' embeddings \cite{tang14:_semipar}.
Specifically,
let $A^{(1)}$ and $A^{(2)}$ be the adjacency matrices of two random dot product graphs on the same vertex set, with vertices aligned so that vertex $i$ in $A^{(1)}$ can be sensibly identified with vertex $i$ in $A^{(2)}$ for all $i \in [n]$. Letting ${\bf\Xhat}$ and ${\bf\Yhat}$ be the respective adjacency spectral embeddings of these two graphs, if the two graphs have the same generating $P$ matrices, then it is reasonable to surmise that the Procrustes distance
\begin{equation} \label{eq:semipar:proc}  \min_{W \in \mathcal{O}^{d \times d}}\|{\bf \Xhat -\Yhat} W \|_F,\end{equation}
will be small.
In \cite{tang14:_semipar}, the authors show that a scaled version of the Procrustes distance in~\eqref{eq:semipar:proc} provides a valid and consistent
test for the equality of latent positions for pairs of graphs. Unfortunately, the fact that a Procrustes minimization must be performed both complicates the test statistic and compromises its power.

The Procrustes alignment is necessary, though, because these two embeddings may well inhabit different $d$-dimensional subspaces of $\mathbb{R}^n$.  An alternative approach is to consider jointly embedding a collection of random graphs into the {\em same} subspace, which is the topic of \cite{levin_omni_2017} and which we extend here. 
Consider $(A^{(1)},A^{(2)},\cdots,A^{(m)},\bX)\sim \text{JRDPG}(F,n,m,R)$.
Once the latent position matrix is generated, the $m$ correlated random graphs with respective adjacency matrices $A_1, \cdots, A_m$ all have the same connection probability matrix $P$. 
That is,
$\mathbb{E}(A^{(k)}|\bX)=P$ for all $k$.
This is a direct graph-analogue of correlated Euclidean data from the same generating distribution, and given that we have {\em multiple} adjacency matrices from the same distribution, it is plausible that a latent position inference procedure using all the $A^{(k)}$ matrices is superior to an inference procedure that depends on a single $A^{(k)}$.  
In addition, since individual graph embeddings cannot be compared without Procrustes alignments, a joint embedding procedure that eliminates post-hoc pairwise alignments can be particularly useful.  

In \cite{levin_omni_2017}, the authors consider the setting where the $A^{(k)}$ are independent (i.e., the conditionally independent JRDPG model), and they build a spectral embedding of an $mn \times mn$ matrix $M$ from the $A^{(k)}$ matrices by placing each $A^{(k)}$ on the main block-diagonal and, on the $(k,\ell)$-th off-diagonal block, the average $\frac{A^{(k)}+A^{(\ell)}}{2}$. 
That is, for, say the $m=2$ case, the matrix $M$ is 
\begin{equation}\label{eq:omnibus_M}
M=\begin{bmatrix}
A^{(1)} 			& \frac{A^{(1)}+A^{(2)}}{2}\\
\frac{A^{(1)} + A^{(2)}}{2}	& A^{(2)}
\end{bmatrix},
\end{equation}
Observe that the expected value of $M$ is the matrix
\begin{equation}
\mathbb{E}[M]
=\begin{bmatrix}
P 	& P\\
P	& P
\end{bmatrix},
\end{equation}
which is still a rank $d$ matrix. As a consequence, a $d$-dimensional embedding of $M$ can produce an $m$-fold collection of {\em correlated} estimates for the rows of the latent position matrix $\bX$.
This {\em joint} or {\em omnibus} spectral embedding, denoted $OMNI$, is defined as follows. 
\begin{definition}[Omnibus Spectral Embedding]
\label{def:omni}
Let $A^{(1)},A^{(2)},\cdots,A^{(m)}$ be a collection of $m$ graphs each in $\mathcal{G}_n$.  
Define the omnibus matrix of $A^{(1)},A^{(2)},\cdots,A^{(m)}$ to be the $mn\times mn$ matrix $M$ defined via
\[M =\begin{bmatrix}A^{(1)}&\frac{A^{(1)}+A^{(2)}}{2}&\frac{A^{(1)}+A^{(3)}}{2}&\cdots&\frac{A^{(1)}+A^{(m)}}{2}\\\frac{A^{(2)}+A^{(1)}}{2}&A^{(2)}&\frac{A^{(2)}+A^{(3)}}{2}&\cdots&\frac{A^{(2)}+A^{(m)}}{2}\\\frac{A^{(3)}+A^{(1)}}{2}&\frac{A^{(3)}+A^{(2)}}{2}&A^{(3)}&\cdots&\frac{A^{(3)}+A^{(m)}}{2}\\\vdots&\vdots&&\ddots&\vdots\\\frac{A^{(m)}+A^{(1)}}{2}&\frac{A^{(m)}+A^{(2)}}{2}&\frac{A^{(m)}+A^{(3)}}{2}&\cdots&A^{(m)}\end{bmatrix}.\]
The $d$-dimensional \emph{Omnibus Spectral Embedding} of $A^{(1)},A^{(2)},\cdots,A^{(m)}$ is then given by
$$
\mathrm{OMNI}(A^{(1)},A^{(2)},\cdots,A^{(m)},d)= \mathrm{ASE}(M, d)=U_M S_M^{1/2},
$$
where $\mathrm{ASE}(M,d)$ is the $d$-dimensional adjacency spectral embedding of $M$; that is, $S_M\in\R^{d\times d}$ is the diagonal matrix of the top $d$ eigenvalues of $|M|$, and $U_M\in\R^{mn\times d}$ the corresponding eigenvectors.
\end{definition}
Note that the omnibus embedding is an $mn \times d$-dimensional matrix.  Each $m$-fold block of rows supplies an $n \times d$ matrix that can serve as a latent position estimate for the corresponding graph. That is, the $s$-th $n \times d$ block of the omnibus embedding, denoted $\hat{\bX}^{(s)}=[U_M S_M^{1/2}]^{s}$, is an estimate for $\bX^{(s)}$, the matrix of latent positions for the $s$-th graph.

In \cite{levin_omni_2017}, it is shown that---in parallel to the same result for the ASE---when $(A^{(1)},A^{(2)},\cdots,A^{(m)},\bX)\sim \text{JRDPG}(F,n,m)$, the rows of the omnibus embedding provide simultaneous consistent estimation for the latent positions $X_j$, where $1 \leq j \leq n$.  There are $m$ such rows for each vertex $j$.  In addition, \cite{levin_omni_2017} demonstrates that for fixed $m$, as $n \rightarrow \infty$, the distribution of any fixed $l$ sub-collection of the rows of the omnibus matrix, suitably scaled, converges in distribution to a mixture of Gaussians. It is important to note that in \cite{levin_omni_2017}, the jointly embedded adjacency matrices are all {\em independent}; as we shall see, in this work, which allows for explicit correlation in the multiple-graph model, we extend consistency and normality to the dependent case.

Even when the adjacency matrices are themselves independent, the simultaneous nature of the omnibus embedding forces correlation across the estimated latent positions, and in return it obviates the need for Procrustes alignments between $n \times d$ matrices $\hat{\bX}^{(s)}$ and $\hat{\bX}^{(t)}$, where $ 1 \leq s \leq t \leq m$. One consequence of this is that with the omnibus embedding, an empirically useful test statistic for assessing latent position equality is simply the Frobenius norm $\|\hat{\bX}^{(s)}-\hat{\bX}^{(t)}\|_F$ (as opposed to $\min_{W\in \mathcal{O}_d}\|\hat{\bX}^{(s)}-\hat{\bX}^{(t)}W\|_F$ in the separately embedded graph setting of \cite{tang14:_semipar}).
Quantifying the correlation induced by this joint embedding, and determining how it relates to the choice of block-matrix on the off-diagonal of $M$, is an important question.  The classical omnibus matrix uses a simple pairwise average, chosen to balance the requirements of estimation accuracy (when all graphs have the same latent positions) with the need to retain discriminatory power in hypothesis testing (when some of the graphs have different latent positions).  By analyzing precisely how the off-diagonal blocks can impact this correlation, we can describe how the omnibus embedding can be used to perform inference on collections of graphs that are not necessarily independent, or can replicate some desired correlation structure in spectral estimates. 

\section{Inherent and induced correlation in classical OMNI}
\label{sec:corr}
  \begin{figure}[t!]
\centering
\includegraphics[width=0.8\textwidth]{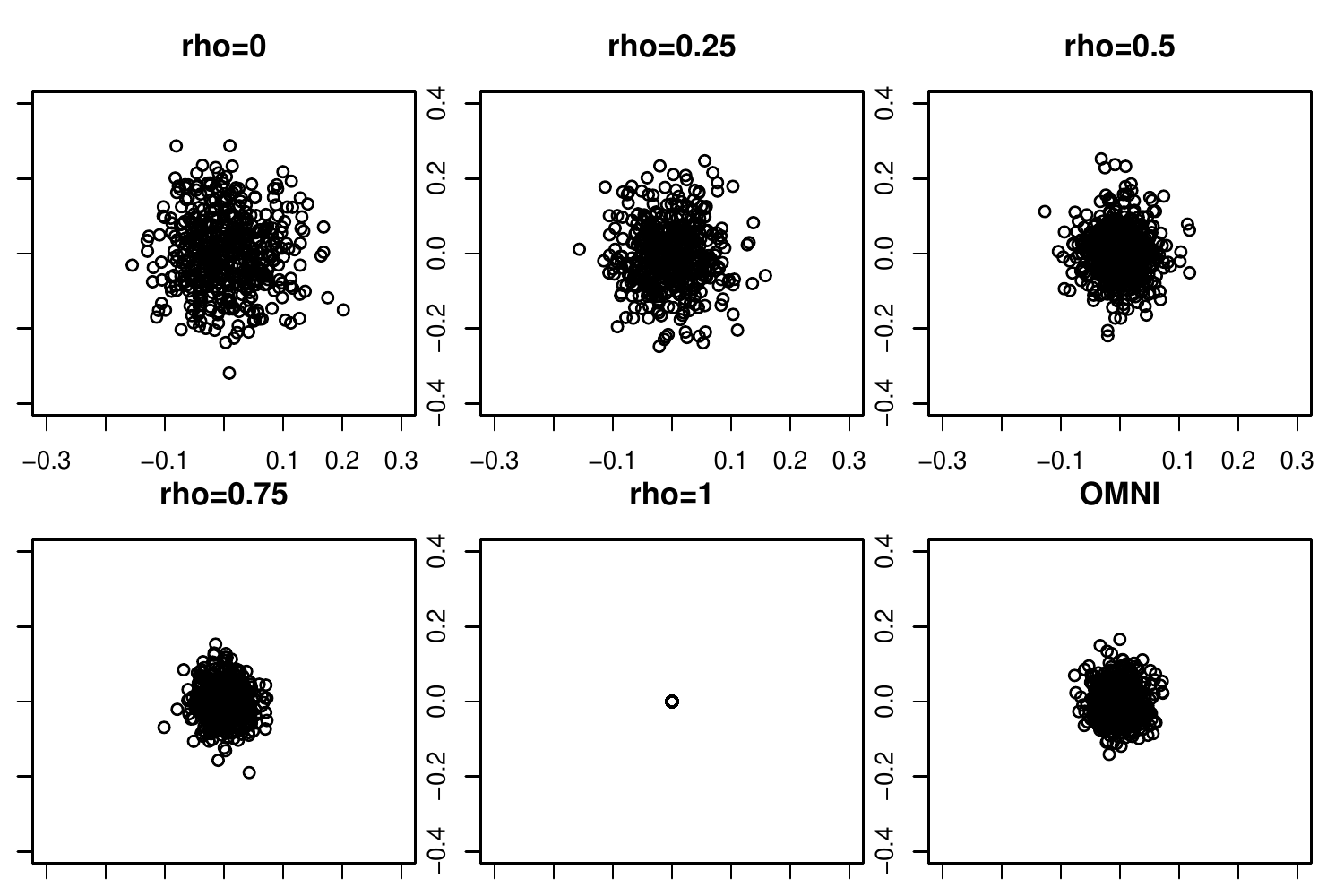}
\caption{The effect of induced versus inherent correlation in the embedded space.  
In panels 1--5 we plot, for various levels of $\rho$, the difference in the aligned estimates of $X_1$ from separately spectrally embedding
$(B^{(1)},B^{(2)},\bX)\sim \text{JRDPG}(F,n,2,\rho)$ (where $F$ has the form in Eq. (\ref{eq:Fex})).
In the sixth panel,
we consider  $(B^{(1)},B^{(2)},\bX)\sim \text{JRDPG}(F,300,2)$, and plot the 
difference in the pair of estimates of $X_1=\xi_1$ derived from the omnibus embedding.
In all panels, the experiment was repeated $nMC=500$ times.
}
  \label{fig:indlat}
\end{figure}

Before delving into our theoretical results on the competing roles of induced versus inherent correlation in the omnibus framework, we first consider the effect of classical OMNI in the original setting of embedding conditionally independent networks considered in \cite{levin_omni_2017,draves2020bias}.
This highlights both the correlation induced by the OMNI method and the dual contributions of inherent and induced correlation in subsequent correlated graph results.
  
To this end, we start with the task of embedding a pair of $n$-vertex correlated $d$-dimensional random dot product graphs,
$(B^{(1)},B^{(2)},\bX)\sim \text{JRDPG}(F,n,2,\rho).$
Prior to the development of methods to jointly embed the networks, a common approach was to separately embed the two graphs into a common Euclidean space, and then align the networks via orthogonal Procrustes analysis \cite{tang14:_semipar}.
One motivating question for the present work is how to capture the effect of the correlation on the embedded pair. Consider a simple, motivating example with $n=300$, $\rho\in\{0,0.25,0.5,0.75,1\}$ and $F$, a mixture of point mass distributions, defined via:
\begin{equation}
\label{eq:Fex}
F=\frac{1}{2}\delta_{\xi_1}+\frac{1}{2}\delta_{\xi_2},
\end{equation}
where $\xi_1$, $\xi_2\in\mathbb{R}^2$ satisfy 
$$
\begin{bmatrix}
\xi_1\\
\xi_2
\end{bmatrix}
\begin{bmatrix}
\xi_1\\
\xi_2
\end{bmatrix}^T=
\begin{bmatrix}
0.7&0.3\\
0.3&0.5
\end{bmatrix},
$$
so that the RDPGs drawn from $F$ are examples of correlated \emph{stochastic blockmodel} random graphs \cite{Holland1983}.
In order to better understand the role of the correlation in the embedded space, we separately spectrally embed each network, $\bhx_{B^{(1)}}=\mathrm{ASE}(B^{(1)},2)$ and 
$\bhx_{B^{(2)}}=\mathrm{ASE}(B^{(2)},2)$, 
and then align the networks via $\bhx_{B^{(1)}}W^{(1)}$ and 
$\bhx_{B^{(2)}}W^{(2)}$ where for each $k=1,2$,
$$W^{(k)}=\text{argmin}_{W\in \mathcal{O}_2}\|\bhx_{B^{(k)}}W-\bX\|_F.$$
In Figure \ref{fig:indlat}, we plot 
$(\bhx_{B^{(1)}}W^{(1)}-\bhx_{B^{(2)}}W^{(2)})_1$ (i.e., the first row of $\bhx_{B^{(1)}}W^{(1)}-\bhx_{B^{(2)}}W^{(2)}$), the  distance between the (aligned) estimates of $X_1=\xi_1$ derived from the embeddings over a range of values of $\rho$; note that in each panel the experiment is repeated $nMC=500$ times.
In the first five panels of the figure, we see the effect of increasing $\rho$ on the difference, namely that the covariance of the difference is monotonically decreasing as $\rho$ increases. 

In the sixth panel of Figure \ref{fig:indlat} (again performing $nMC=500$ Monte Carlo replicates), we consider the conditionally independent case (i.e., $\rho=0$), so that $(B^{(1)},B^{(2)},\bX)\sim \text{JRDPG}(F,300,2)$.
We consider the omnibus spectral embedding of $(B^{(1)},B^{(2)})$, denoted
$\bhx_M=\mathrm{OMNI}(B^{(1)},B^{(2)},2)$, and plot the 
difference in the estimates of $X_1=\xi_1$ derived from the omnibus embedding, namely
$(\bhx_M^{(1)})_1-(\bhx_M^{(2)})_{1}.$
From the figure, we see that the correlation induced between the independent graphs in the embedded space is (roughly) equivalent to the inherent correlation between $\rho=0.75$ correlated networks that have been separately embedded and aligned.  
In the next section, we will formalize this notion of induced versus inherent correlation, and we will see that, as the Figure suggests, OMNI does indeed induce correlation of level $\rho=0.75$ across independent graphs in the embedded space. 

\subsection{Central limit theorems and correlation in the embedded space}
\label{sec:CLTcorr}

Viewing $(B^{(1)},B^{(2)},\bX)$ as
an element of the sequence 
$$\left((B^{(1)}_n,B^{(2)}_n,\bX_n)\right)_{n=1}^\infty,$$
where for each $n\geq 1$, 
$(B^{(1)}_n,B^{(2)}_n,\bX_n)\sim \text{JRDPG}(F,n,2,\rho),$
the Central Limit Theorem established in \cite{athreya2013limit} provides a framework for understanding the effect, in the embedded space,  of the edge-wise correlation across networks.
Letting 
$\widehat \bX_{B^{(1)}_n}=\mathrm{ASE}(B^{(1)}_n,d)$ and 
$\widehat \bX_{B^{(2)}_n}=\mathrm{ASE}(B^{(2)}_n,d)$, Theorem 3.3 of \cite{athreya2013limit} (in the form presented in Theorem 9 of \cite{athreya_survey}) implies that for each $k=1,2$, if $\Delta := \mathbb{E}[X_1 X_1^{T}]$ is rank $d$, then
there exist sequences of orthogonal $d$-by-$d$ matrices
	$( \Tilde W_n^{(k)})_{n=1}^\infty$ such that for all $z \in \R^d$ and for any fixed index $i$,
	$$ \lim_{n \rightarrow \infty}
	\p\left[ n^{1/2} \left( \widehat \bX_{B^{(k)}_n} \Tilde W_n^{(k)} - \bX_n \right)_i
	\le z \right]
	= \int_{\text{supp } F} \Phi\left(z, \Sigma(x) \right) dF(x), $$
	where
	\begin{align}
	\label{def:sigma}
	\Sigma(x) 
	&:= \Delta^{-1} \EX\left[ (x^{T} X_1 - ( x^{T} X_1)^2 ) X_1 X_1^{T} \right] \Delta^{-1};
	\end{align}
$\Phi(\cdot,\Sigma)$ denotes the cdf of a (multivariate)
	Gaussian with mean zero and covariance matrix $\Sigma$,
	and 
	$( \widehat \bX_n \Tilde W_n - \bX_n )_i$ denotes the i-th row of $\widehat \bX_n \Tilde W_n - \bX_n.$ 
Combining the above central limit theorems for $\widehat \bX_{B^{(1)}_n}$ and $\widehat \bX_{B^{(2)}_n}$, we have the following theorem (proven in Section \ref{sec:latent}).
\begin{theorem}
\label{thm:rhoCLT}
Let $\rho\in(0,1)$ be fixed.
Let $F$ be a distribution on a set $\mathcal{X}\subset \R^{d}$, where $\langle x, x'\rangle\in[0,1]$ for all $x,x'\in\mathcal{X}$, and assume that $\Delta := \mathbb{E}[X_1 X_1^{T}]$ is rank $d$.
Let $(B^{(1)}_n,B^{(2)}_n,\bX_n)\sim \mathrm{JRDPG}(F,n,2,\rho),$ be a sequence of adjacency matrices and associated latent positions, where for each $n\geq 1$ the rows of $\bX_n$ are i.i.d. distributed according to $F$.
Letting 
$\bhx_{B^{(1)}_n}=\mathrm{ASE}(B^{(1)}_n,d)$ and 
$\widehat \bX_{B^{(2)}_n}=\mathrm{ASE}(B^{(2)}_n,d)$, 
there exist sequences of orthogonal $d$-by-$d$ matrices
	$( W^{(1)}_n )_{n=1}^\infty$, $( W^{(2)}_n )_{n=1}^\infty$ such that for all $z \in \R^d$ and for any fixed index $i$, 
	\begin{align}
	    \label{eq:corCLT}
	    \lim_{n \rightarrow \infty}
	\p\left[ n^{1/2} \left( \widehat \bX_{B^{(1)}_n} W^{(1)}_n - \widehat \bX_{B^{(2)}_n}W^{(2)}_n \right)_i
	\le z \right]
	= \int_{\text{supp } F} \Phi\left(z, \widetilde\Sigma(x,\rho) \right) dF(x), 
	\end{align}
	where
	$\widetilde\Sigma(x,\rho)=2(1-\rho)\Sigma(x).$
\end{theorem}

From Theorem \ref{thm:rhoCLT}, we see that the effect of the correlation $\rho$ in the embedding space is to introduce a dampening factor of $(1-\rho)$ into the asymptotic limiting covariance (see Figure \ref{fig:indlat}).
This is entirely reasonable; indeed, consider $Y_1,Y_2$ to be $\rho$ correlated Norm($\mu,\sigma^2)$ random variables, in which case $Y_1-Y_2\sim$Norm($0,2(1-\rho)\sigma^2$). 
In the embedded space, no extraneous correlation is introduced (in the limit) by separately embedding the networks and aligning the embeddings via $(W^{(k)}_n)_{n=1}^\infty$ for $k=1,2$.
Joint embedding procedures like the Omnibus method forgo these Procrustes rotations, but at a price: they induce correlation across even independent networks. 
To understand this, we consider the omnibus central limit theorem of \cite{levin_omni_2017} and derive the following result. Its proof can be obtained from \cite{levin_omni_2017}, but is also an immediate consequence of our more general main result, Theorem \ref{thm:pfdiffgenomni}, which is stated formally in the following section.
\begin{theorem}[Induced correlation in classical OMNI]
\label{thm:omniCLT}
Let $F$ be a distribution on a set $\mathcal{X}\subset \R^{d}$, where $\langle x, x'\rangle\in[0,1]$ for all $x,x'\in\mathcal{X}$, and assume that $\Delta := \mathbb{E}[X_1 X_1^{T}]$ is rank $d$.
Let $(A_n^{(1)},A_n^{(2)},\cdots,A_n^{(m)},\bX_n)\sim\mathrm{JRDPG}(F,n,m)$ be a sequence of independent $\mathrm{RDPG}$ random graphs, and for each $n\geq 1$, let $M_n$ denote the omnibus matrix as in Definition \ref{def:omni}. Also, let $\bhx_{M_n}=\mathrm{ASE}(M_n,d)$ and denote the $s$-th $n \times d$ block of the omnibus embedding $\bhx_{M_n}$ as $\hat{\bX}^{(s)}_{M_n}$.
Consider fixed indices $i\in[n]$ and $s_1,s_2\in[m]$. 
Then there exists a sequence of orthogonal matrices $(\Tilde W_n)_{n=1}^\infty$ such that for all $z \in \R^d$, we have that
\begin{align}
\label{eq:omniclt}
    \lim_{n \rightarrow \infty}
	\p\left[n^{1/2}\left[\left(\bhx_{M_n}^{(s_1)}-\bhx_{M_n}^{(s_2)}\right)\Tilde W_n\right]_{i}\leq z\right]=
    \int_{\text{supp } F}\Phi(z,\widetilde\Sigma(x,3/4))dF(x).
\end{align}
\end{theorem}
\noindent 
Unpacking Theorem \ref{thm:omniCLT}, $(\bhx_{M_n}^{(s_1)})_{i}$ and $(\bhx_{M_n}^{(s_2)})_{i}$ represent the estimates of $X_i$ derived from $A^{(s_1)}$ and $A^{(s_2)}$ by the omnibus embedding paradigm.
In light of Theorem \ref{thm:rhoCLT}, the induced correlation in OMNI can be understood in the context of the limiting covariance of the difference of a pair of estimates for the same underlying latent position.
Comparing Eqs. (\ref{eq:corCLT}) and (\ref{eq:omniclt}), we see that OMNI effectively induces a correlation of level $\rho=3/4$ uniformly across embedded network pairs.

While it is easy to surmise that this flat correlation would have a detrimental signal dampening effect on inference tasks in settings where there is nuanced inherent correlation (e.g., in time-series analysis), this theory is in the setting of conditionally independent networks.  
It is natural to ask whether in the presence of inherent correlation across the network pairs the dampening effect of the induced correlation in OMNI diminished.  
As we will see in our next result (a special case of Theorem \ref{thm:pfdiffgenomni} in Section \ref{sec:indcorr_gen}), this is not the case.
\begin{theorem}[Induced and inherent correlation in classical OMNI]
\label{theorem:ind_lat_classical_omni}
Let $F$ be a distribution on a set $\mathcal{X}\subset \R^{d}$, where $\langle x, x'\rangle\in[0,1]$ for all $x,x'\in\mathcal{X}$, and assume that $\Delta := \mathbb{E}[X_1 X_1^{T}]$ is rank $d$.
Let $$(A_n^{(1)},A_n^{(2)},\cdots,A_n^{(m)},\bX_n)\sim \mathrm{JRDPG}(F,n,m,R)$$ be a sequence of correlated $\mathrm{RDPG}$ random graphs, and for each $n\geq 1$,
let $M_n$ denote the omnibus matrix as in Definition \ref{def:omni}. Also, let $\bhx_{M_n}=\mathrm{ASE}(M_n,d)$ and denote the $s$-th $n \times d$ block of the omnibus embedding $\bhx_{M_n}$ as $\hat{\bX}^{(s)}_{M_n}$. Consider fixed indices $i\in[n]$ and $s_1,s_2\in[m]$.
Then there exists a sequence of orthogonal matrices $(\Tilde W_n)_{n=1}^\infty$ such that for all $z \in \R^d$, we have that
	\begin{align}
	    \label{eq:corCLT_omni}
	    \lim_{n \rightarrow \infty}
\p\left[n^{1/2}\left[\left(\bhx_{M_n}^{(s_1)}-\bhx_{M_n}^{(s_2)}\right)\Tilde W_n\right]_{i}\leq z\right]
	= \int_{\text{supp } F} \Phi\left(z, \widetilde\Sigma(x, \rho(s_1,s_2)) \right) dF(x), 
	\end{align}
	where
	 $\rho(s_1,s_2)=\frac{3}{4}+\frac{1}{4}\rho_{s_1,s_2}$.
	\end{theorem}
\noindent In the $\mathrm{JRDPG}(F,n,m,R)$ setting, the form of the limiting correlation highlights the separate contributions from the \emph{method} (i.e., the induced correlation 3/4) and the \emph{model} (i.e., the inherent correlation $\rho_{s_1,s_2}$).
This correlation is independent of $m$, and the outsized effect of the induced correlation versus the downscaled inherent correlation further suggests that OMNI is not ideal for embedding temporal sequences of networks that exhibit complex dependency patterns (i.e., a change point or anomaly).
Indeed, the induced correlation from OMNI whitens out the existing complex correlation structures amongst the $A^{(k)}$, and, as we will show in our next example, this whitening effect of OMNI can serve to mask significant data features and structures in complex data environs.


\subsection{OMNI's whitening correlation: Aplysia spike train analysis}
\label{sec:aplysia}
We consider 
a time series of networks derived from the Aplysia californica escape motor program of \cite{Aplysia}.
The motor program consists of a 20 min recording of the action potentials generated by 82 neurons in the dorsal pedal ganglion of an isolated brain preparation from the marine mollusk Aplysia californica.  
One minute into the recording a brief electrical stimulus was applied to pedal ganglion nerve 9 to elicit the animal’s rhythmic escape locomotion motor program.  
This consists of an initial rapid bursting, lasting several cycles, that drives the animal’s gallop behavior, followed by a slower rhythm persisting until the end of the recording that drives the animal’s crawling behavior.  
To extract a network time series from the recording, we binned the motor program into 24 bins, each approximately $\approx 50$ second long; the binned motor program (with the second bin, containing the stimulus, highlighted) is pictured in Figure \ref{fig:binned}.
\begin{figure}[t!]
\centering
\includegraphics[width=1\textwidth]{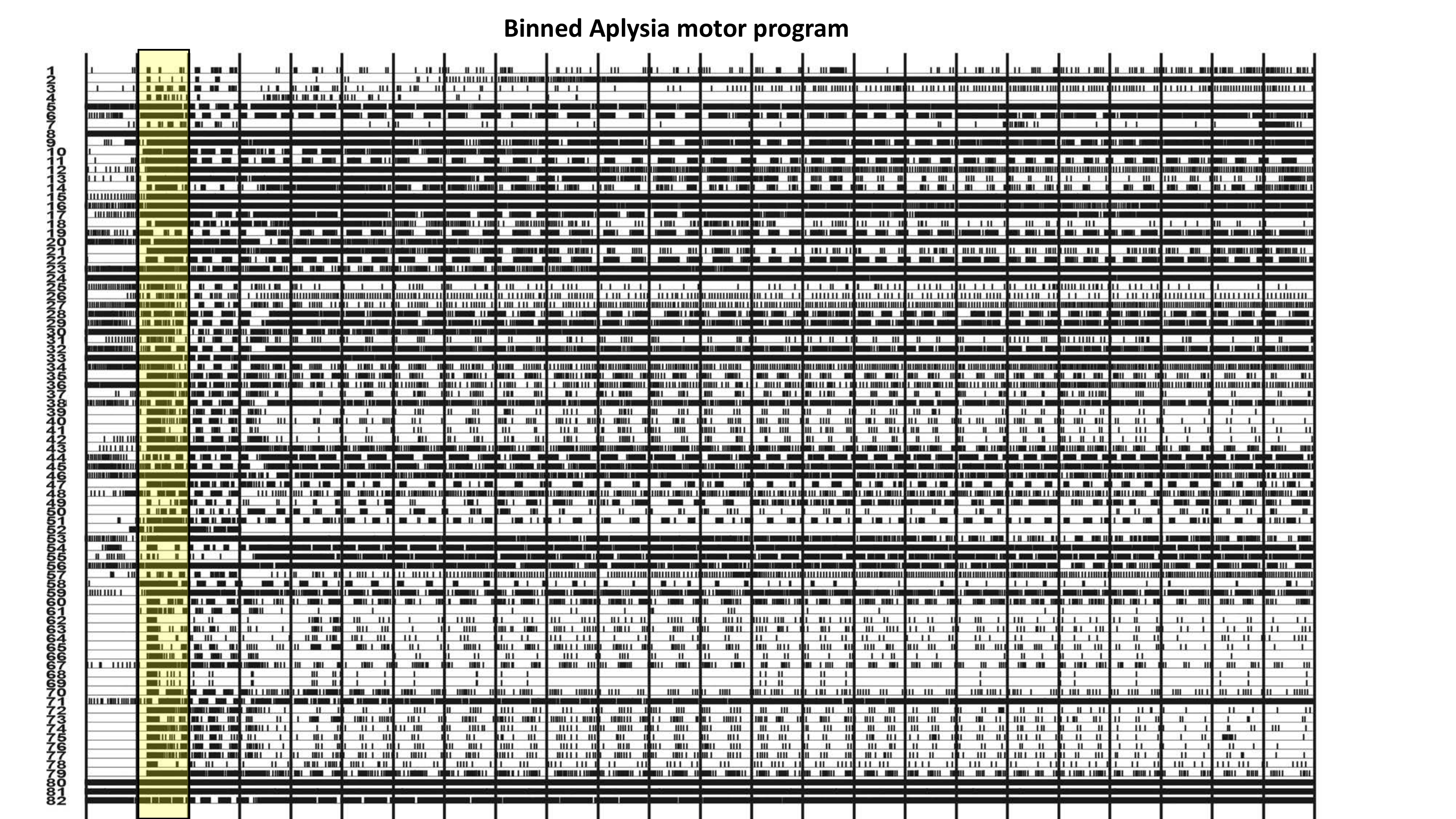}
\caption{The 20 minute Aplysia escape motor program from \cite{Aplysia}, binned into 24 windows, each approximately $50$ seconds in length.  
The stimulus happens one minute into the motor program, in the highlighted second bin.}
  \label{fig:binned}
\end{figure}
Using the \texttt{meaRtools} package in \texttt{R} \cite{gelfman2018meartools}, we apply the STTC (spike time tiling correlation) method of \cite{cutts} to convert each 50 second window into a weighted correlation matrix amongst the 82 neurons.
Each of the 24 bins then yields one weighted graph on 82 vertices representing this correlation structure amongst the 82 neurons.
The stimulus occurs then in the second of the 24 bins, followed by the gallop and crawling motor programs.  

\subsubsection{Classical omnibus embeddings and correlation masking}
\label{sec:masking}

Our motivation in the following analyses is to determine whether, in light of the correlation structure in Theorem \ref{theorem:ind_lat_classical_omni}, the classical OMNI method can detect the distinct phases of behavior in the motor program.
In order to explore the impact of the flat correlation induced in the classical omnibus embedding, we use the classical OMNI spectral embedding to embed the $\{A^{(k)}\}_{k=1}^{24}$ into a common $\mathbb{R}^d$ (where $d=4$ as chosen by locating the elbow in the scree plot as suggested by \cite{zhu06:_autom,chatterjee2015}).
Further isolating the impact of the stimulus and the evolution of the galloping and crawling phases of the motor program, we plot the average distance between the estimated latent positions for each vertex (as the bar heights) in the embedding between embedded graph 1 (res., embedded graph 2) and embedded graph $k$ for each $k\in[m]$ (where the location of the bars in Figure \ref{fig:omniaplysia} correspond to $k=1,2,\cdots,24$).
To wit, we plot (where $\bhx_{M_n}:=[(\bhx_{M_n}^{(1)})^T\,|\, (\bhx_{M_n}^{(2)})^T\,|\, \cdots\,|\, (\bhx_{M_n}^{(m)})^T]^T$)

\begin{align*}
    \frac{1}{82}\sum_{i=1}^{82}
    \left\|\left(\bhx_{M_n}^{(1)}\right)_i-\left(\bhx_{M_n}^{(k)}\right)_i\right\|_2 & \text{ in the left panel in Figure \ref{fig:omniaplysia}};\\
    \frac{1}{82}\sum_{i=1}^{82}
    \left\|\left(\bhx_{M_n}^{(2)}\right)_i-\left(\bhx_{M_n}^{(k)}\right)_i\right\|_2 & \text{ in the center panel in Figure \ref{fig:omniaplysia}}.
\end{align*}
\begin{figure}[t!]
\centering
\includegraphics[width=1\textwidth]{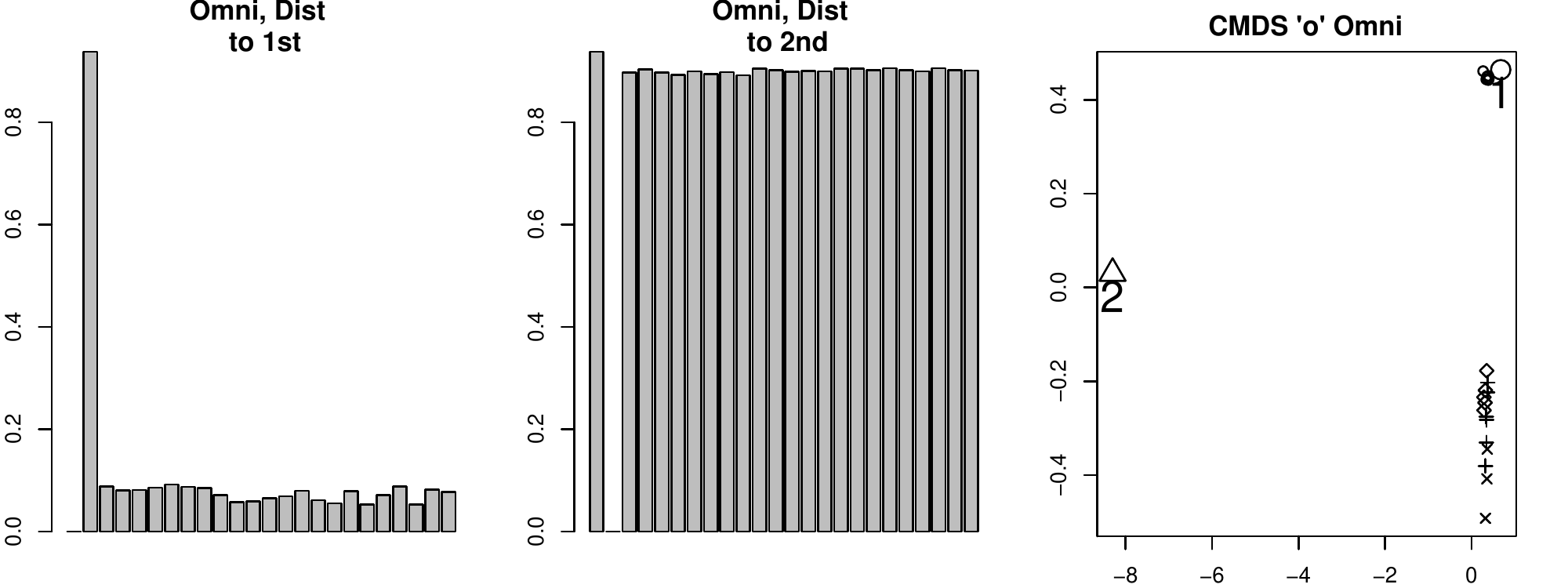}
\caption{
In the left (resp., center) panel, we plot the average vertex distance (as the bar heights) in the embedding between graph 1 (resp., graph 2) and graph $k$ for each $k\in[m]$ in OMNI (where the bars are labeled $k=1,2,\cdots,24$).
In the right panel, we compute the $24\times 24$ distance matrices ${\bf D}$, and embed it into $\mathbb{R}^2$.
The resulting $24$ data points are clustered using \texttt{Mclust} (clusters denoted by shape), and are plotted, with graphs 1 and 2 further labeled with their corresponding number.
}
  \label{fig:omniaplysia}
\end{figure}

From Figure \ref{fig:binned} and our knowledge of the Aplysia motor program's evolution, we see that if the omnibus embedding is able to detect the biologically distinct phases of this network time series, then 
\begin{itemize}
    \item[i.] There will be a significant difference between the embeddings of graphs 1 and 2 (as demonstrated in Figure \ref{fig:omniaplysia}).
    \item[ii.] The effect of the stimulus is less apparent as the Aplysia transitions from galloping to crawling; this would be manifested as the distance between graphs 2 and $k>2$ increasing as the Aplysia transitions from galloping to crawling. 
    \item[iii.] The distance from graph 1 to graphs $k>2$ should be large, as the Aplysia never returns to its spontaneous firing state in the motor program.
\end{itemize}
While we see that the omnibus methodology demonstrates the capacity to detect the anomaly (namely, the stimulus) in the second graph, the flat correlation structure induced in the embedding space has the effect of masking the transition from galloping to crawling and creates an artificial similarity between graphs 1 and some of the graphs $k>2$ in the embedded space (as shown in the right panel of Figure \ref{fig:omniaplysia}).

Exploring this further, we compute the $24\times 24$ distance matrix ${\bf D}$, where
\begin{equation}
\label{eq:d}
    {\bf D}=[D_{k,\ell}]\text{ where }D_{k.\ell}=\|\bhx_{M_n}^{(k)}-\bhx_{M_n}^{(\ell)}\|_F.
    \end{equation}
The matrix ${\bf D}$ is then each embedded into $\mathbb{R}^2$ (where $d=2$ is, again, chosen by the elbow in the scree plot
\begin{figure}[t]
    \centering
    \includegraphics[width=0.7\textwidth]{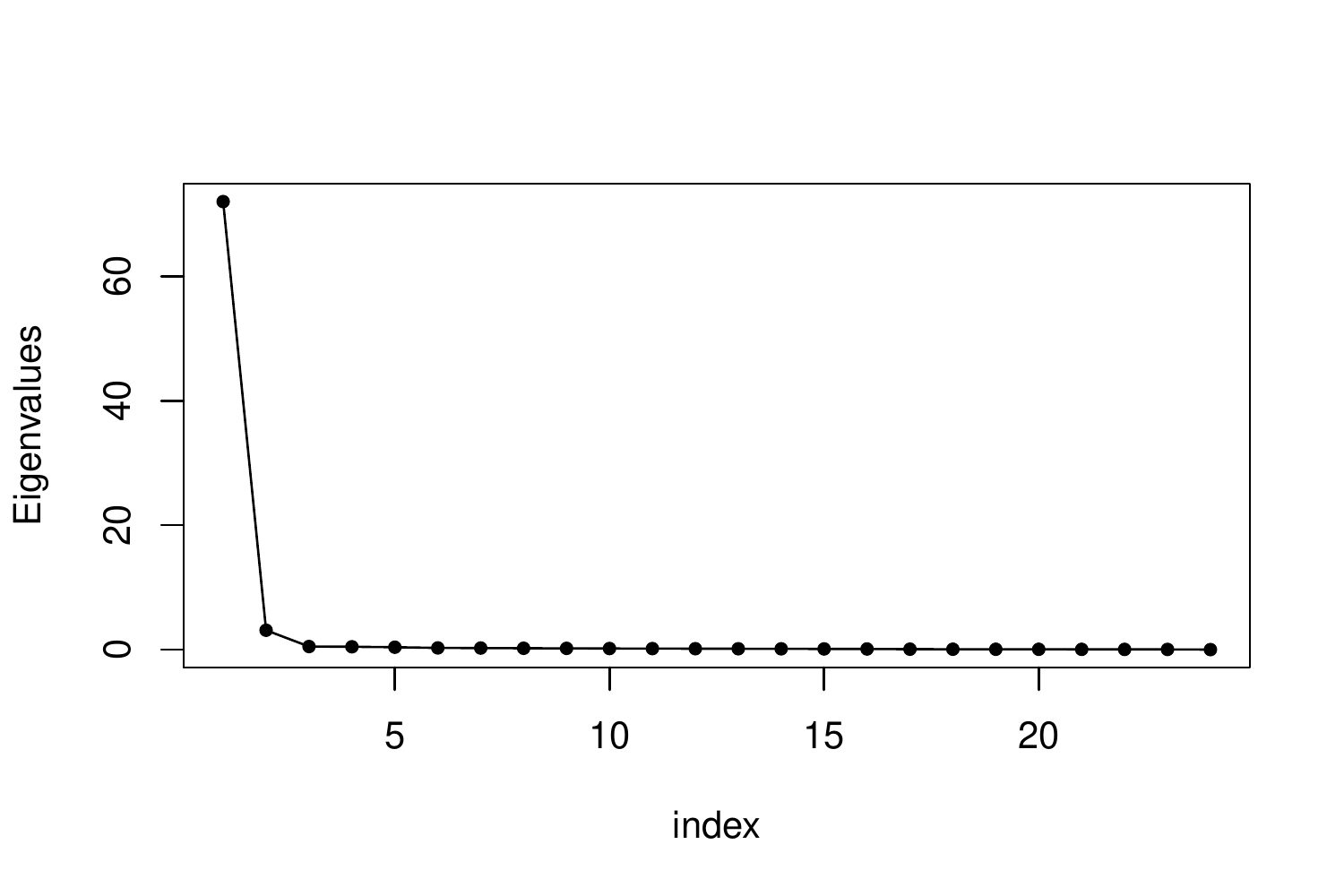}
    \caption{Scree plot of the eigenvalues of $\mathbf{D}$.}
    \label{fig:scree}
\end{figure}
as suggested by \cite{zhu06:_autom,chatterjee2015}; see Figure \ref{fig:scree}), and the $24$ data points are clustered using \texttt{Mclust} \cite{mclust2012}.
The resulting clusters are plotted in the right panel of Figure \ref{fig:omniaplysia}. 
Points 1 and 2 (corresponding to graphs 1 and 2 resp.) are plotted with larger symbols and labeled.
The cluster labels found by \texttt{Mclust} are given in Table \ref{tab:omni_aplysia}.
\begin{table}
\begin{center}
    \begin{tabular}{c|c|c|c|c|c|c|c|c|c|c|c|c}
    Graph& 1 & 2 & 3 & 4  &5&  6 & 7 & 8 & 9 &10 &11 &12\\\hline
   Cluster& 1 & 2 & 5 & 3  &5&  5 & 4 & 5 & 3 &1 &1 &1\\
    \hline\hline
    Graph& 13 &14 &15& 16& 17 &18& 19& 20 &21& 22 &23& 24\\\hline
   Cluster& 1 &5 &3& 1& 1 &4& 1& 3 &3& 1 &4& 3
    \end{tabular}
\end{center}
    \caption{\label{tab:omni_aplysia}
This table displays the cluster labels for the 24 networks when embedding the time-series using classical OMNI, then embedding the across-graph distance matrix using classical MDS scaling, and finally clustering the graphs using \texttt{Mclust}.
    In the MDS embedding, each graph is represented by a single $2-$dimensional embedded data point, and the pair (Graph, Cluster) describes which cluster these points are assigned by \texttt{Mclust}.}
\end{table}
This clustering reinforces the finding that while the stimulus is detected (graph 2 is clustered apart from the others), the transition from stimulus to gallop and crawl and the distinct nature of graph 1 (as the only spontaneous firing state measurement) are masked in this analysis. We shall see in Section \ref{sec:SPIKE} that different structures in the Omnibus matrix can be introduced to ameliorate these shortcomings.


\section{Generalized omnibus embeddings}
\label{sec:genomni}

All joint embedding procedures induce correlation across the embedded networks, and the omnibus embedding is no exception. Unfortunately, the particular structure of the correlation induced by classical OMNI (namely, the large, flat induced correlation across all graphs) may render it less effective for time-series applications or settings where the correlation varies dramatically across networks.
This motivates our next contribution, in which we show that by generalizing the structure of the omnibus matrix $M$, more nuanced (and  application0-appropriate) induced correlation structures are possible.

The consistency and asymptotic normality of the classical omnibus embedding rest on a few defining model assumptions: first,
$$\Ex[M]=\TP:=
J_m\otimes P,$$
and second, in the $k$-th block-row of $M$, the weight put on $A^{(k)}$ (which is equal to $1+(m-1)/2$) is strictly greater than the weight put on any $A^{(\ell)}$ for $\ell\neq k$ (as each of these is $1/2$).
The low-rank RDPG structure of $\Ex[M]=\TP$ allows us to use random matrix concentration results to prove that the scaled eigenvectors of $M$ are tightly concentrated about the scaled eigenvectors of $\TP$ which, up to a possible rotation, are equal to $\bZ:=\vec{1}_m\otimes\bX$.
The weights in block-row $k$ being maximized for $A^{(k)}$ ensures that the $k$-th block in ASE$(M,d)$ corresponds to the embedding of $A^{(k)}$, which is essential for subsequent inference (e.g., hypothesis testing \cite{draves2020bias}) in the omnibus setting.

It is natural to ask if, subject to the above conditions, we can generalize the omnibus structure to permit more esoteric induced correlation in the embedded space.
This motivates the following definition of the generalized omnibus matrix.
Let $F$ be a distribution on a set $\mathcal{X}\subset \R^{d}$, where $\langle x, x'\rangle\in[0,1]$ for all $x,x'\in\mathcal{X}$.
Suppose that $(A^{(1)},A^{(2)},\cdots,A^{(m)},\bX)\sim$ JRDPG$(F,n,m,R)$. The convex hull of $A^{(1)},A^{(2)},\cdots,A^{(m)}$ 
is denoted via $$\mathcal{A}_m=\Big\{\sum\limits_{q=1}^{m}c_qA^{(q)}\in\R^{n\times n}\Big| \text{ with }c_q \geq 0\text{ for all }q\in[m]\text{ and } \sum\limits_{q=1}^{m}c_q=1\Big\}.$$

\begin{definition}[Generalized Omnibus Matrix]
\label{def:genOMNI}
Consider $\mathfrak{M}\in\R^{mn\times mn}$, a generalized version of the omnibus $M$ of Definition \ref{def:omni} defined as follows.
$\mathfrak{M}$ is a block matrix satisfying the following assumptions:
\begin{enumerate}
    \item Each block entry $\mathfrak{M}^{(k,\ell)}\in\R^{n\times n}$, $k,\ell\in[m]$ is an element of  $\mathcal{A}_m$; 
    we will write $$\mathfrak{M}^{(k,\ell)}=\sum\limits_{q=1}^{m}c^{(k,\ell)}_q A^{(q)}.$$
    \item For each block row $k\in[m]$ of $\mathfrak{M}$, the cumulative weight of $A^{(k)}$ is greater than the cumulative weight of the rest of the $A^{(q)}$, $q\neq k$; i.e., for $q\neq k$
    $$\sum_{\ell} c^{(k,\ell)}_q < \sum_{\ell} c^{(k,\ell)}_k.$$
    \item $\mathfrak{M}$ is symmetric.
\end{enumerate}
Such a block matrix satisfying assumptions 1-3 above will be referred to as a \emph{Generalized Omnibus Matrix}. 
\end{definition}

\noindent If $\fM$ is a Generalized Omnibus Matrix, 
setting $\alpha(k,q)=\sum_{\ell} c_{q}^{(k,\ell)}$ to be the weight put on $A^{(q)}$ in the $k$-th block-row of $\fM$, assumption 2 above becomes, for each $k\in[m]$, 
\begin{align}
    \label{eq:rowiofM}
\alpha(k,q)<\alpha(k,k)\hspace{.5cm}\text{ for all $q\neq k$}.
\end{align}
Because $\fM$ (i.e., $\Ex(\fM)=\TP$) is unbiased and low-rank, we can appropriately modify matrix concentration and perturbation results for the scaled eigenvectors of $\fM$, and we can establish that the rows of ASE$(\fM,d)$ consistently estimate the associated rows of $\bZ$ and the associated residual errors satisfy a distributional central limit theorem. This is the content of the following theorems, whose proofs can be found in Appendix. 
\begin{theorem}[Consistency of generalized omnibus embedding estimates]
\label{theorem:genOMNI_consistency}
Let $F$ be a distribution on a set $\mathcal{X}\subset \R^{d}$, where $\langle x, x'\rangle\in[0,1]$ for all $x,x'\in\mathcal{X}$, and assume that $\Delta := \mathbb{E}[X_1 X_1^{T}]$ is rank $d$.
Let $$(A_n^{(1)},A_n^{(2)},\cdots,A_n^{(m)},\bX_n)\sim \mathrm{JRDPG}(F,n,m,R)$$ be a sequence of correlated $\mathrm{RDPG}$ random graphs.
For each $n\geq 1$, let $\fM_n$ denote the generalized omnibus matrix as in Definition \ref{def:genOMNI}.
Let the spectral decomposition of $\TP_n=\mathbb{E}(\fM_n)$ be given by
$$\widetilde{P}_n=\Ex(\fM_n)=U_{\TP_n}S_{\TP_n}U_{\TP_n}^T,$$ 
where $U_{\TP_n}\in\mathbb{R}^{mn\times d}$ and $
S_{\TP_n}\in\mathbb{R}^{d\times d}$. 
Let $\bhx_{\fM_n}=\mathrm{ASE}(\fM_n,d)=U_{\fM_n}S_{\fM_n}^{1/2}$ denote the adjacency spectral embedding of $\fM_n$. 
Then
there exists an orthogonal matrix $V_n\in\mathcal{O}_d$  and a constant $C>0$ such that, with high probability for $n$ sufficiently large,
\begin{align*}
    \|U_{\mathfrak{M}_n}S_{\mathfrak{M}_n}^{1/2}-U_{\TP_n}S_{\TP_n}^{1/2}V_n\|_{2\rightarrow \infty}\leq C\frac{m^{1/2}\log mn}{n^{1/2}}.
\end{align*} 
\end{theorem}

\noindent 
The above theorem is analogous to Lemma 1 in \cite{levin_omni_2017}. Interestingly, the consistency error rate for the generalized omnibus embedding estimates with pair-wise correlated adjacency matrices coincides with the consistency error rate for the classical omnibus embedding estimates with independent adjacency matrices. This is a consequence of two facts: first, the sum of the cumulative weights $\sum_{q=1}^{m}\alpha(k,q)$ is equal to $m$, and second, the extra correlation term in the matrix Bernstein concentration bound is of the same order as in the independent term.

Building upon this consistency, our principal result describes the limiting distributional behavior of low-dimensional embeddings of the generalized omnibus matrix for joint networks, stated below.
\begin{theorem}[Asymptotic normality of rows of the generalized omnibus embedding]
\label{thm:genOMNI} Let $F$ be a distribution on a set $\mathcal{X}\subset \R^{d}$, where $\langle x, x'\rangle\in[0,1]$ for all $x,x'\in\mathcal{X}$, and assume that $\Delta := \mathbb{E}[X_1 X_1^{T}]$ is rank $d$.
Let $(A_n^{(1)},A_n^{(2)},\cdots,A_n^{(m)},\bX_n)\sim\mathrm{JRDPG}(F,n,m,R)$ be a sequence of correlated $\mathrm{RDPG}$ random graphs, and let $\bZ_n=[\bX_n^T|\bX_n^T|\cdots|\bX_n^T]^T\in\mathbb{R}^{mn\times d}$. Further, for each $n\geq 1$, let $\fM_n$ denote the generalized omnibus matrix as in Definition \ref{def:genOMNI} and $\bhx_{\fM_n}=\mathrm{ASE}(\fM_n,d)$. For fixed indices $i\in[n]$ and $s\in[m]$, let $\Big(\hat{\bX}^{(s)}_{\fM_n}\Big)_i$ denote the $i$-th row of the $s$-th $n \times d$ block of $\bhx_{\fM_n}$ (i.e., the estimated latent position from graph $A_n^{(s)}$ of the hidden latent position $X_i$).
 Then, there exists a sequence of orthogonal matrices $\{\mathcal{Q}_n\}_n$ such that
\begin{align*}
  \lim_{n\rightarrow\infty}\PX\Big[
    n^{1/2}\Big(\bhx^{(s)}_{\fM_n}\mathcal{Q}_n-\bZ_n\Big)_i\leq z\Big]=\int_{\text{suppF}}\Phi(z,\Check\Sigma_{\rho}(x;s))dF(x),
\end{align*}
where 
\begin{align*}
        \Check\Sigma_{\rho}(x;s)&=\frac{1}{m^2}\Big(\underbrace{\sum_{q=1}^{m}\alpha^2(s,q)}_{\text{method-induced coefficient
        }}+\underbrace{2\sum_{q<l}\alpha(s,q)\alpha(s,l)\rho_{q,l}}_{\text{model-inherent coefficient
        }}\Big)\underbrace{\Delta^{-1}\EX[x^TX_j(1-x^TX_j)X_jX_j^T]\Delta^{-1}}_{:=\Sigma(x)}
\end{align*}
\end{theorem}


\noindent
The covariance matrix $ \Check\Sigma_{\rho}(x_i;s)$ corresponds to the residual of the (true) latent position of vertex $i$ and its estimate from graph $A^{(s)}$.
We note that this covariance matrix can be written as a sum of two terms, for which the first term $$\Sigma_{0}(x_i;s):=\bigg(\frac{1}{m^2}\sum_{q=1}^{m}\alpha^2(s,q)\bigg)\Sigma(x_i)$$
corresponds to the covariance matrix for the $i$-th row from the $s$-th block of the generalized omnibus embedding under the JRDPG$(F,n,m)$ model as in Definition \ref{def:JRDPG}, and the second term is accredited to the presence of the inherent correlation from the JRDPG$(F,n,m,R)$ model. 
By noting that $\rho_{k,\ell}\geq 0$ and $\alpha(k,\ell)\geq 0$ for all $k,\ell\in[m]$, we deduce that $ \Check\Sigma_{\rho}(x_i;s)\geq \Sigma_{0}(x_i;s)$ entry-wise for all $i$, implying tighter residual errors when the estimates arise from independent graphs.
This observation motivates our simulations in Section \ref{sec:effective-sample-size} where we explore the effect of the inherent correlation on effective sample size in the context of community detection in SBM's.

\begin{remark}
Note that for fixed $n$, the orthogonal matrices $V_n$, $\mathcal{Q}_n$ in Theorems \ref{theorem:genOMNI_consistency} and \ref{thm:genOMNI} accordingly, can be explicitly defined. The matrix $V_n$ is the solution of the Procrustes problem
$$\min_{V\in\mathcal{O}_d}\|U_{\fM_n}-U_{\TP_n}V\|,$$ and its solution is given by $V_n:=V_{1,n}V_{2,n}^T$ where the columns of the matrices $V_{1,n},V_{2,n}$ are the left- and right-singular vectors of $U_{\TP_n}^TU_{\fM_n}$ respectively. Moreover,  $\mathcal{Q}_n:=V_n^TW_n$, where $W_n$ is an orthogonal transformation such that $\bZ_n= \Hat{\bZ}_nW_n=U_{\TP_n}S_{\TP_n}^{1/2}W_n$.
\end{remark}

In the following example, we apply Theorem \ref{thm:genOMNI} with the (classical) omnibus matrix from \cite{levin_omni_2017}.
\begin{example}{\bf (Classical omnibus matrix)}\\
\normalfont
\label{ex:classicalomni}
Suppose first that $(A^{(1)},A^{(2)},\cdots,A^{(m)},\bX)\sim\mathrm{JRDPG}(F,n,m,R)$. 
Let $e_k\in\R^{m}$ be the vector with the $k$-th entry equal to $1$ and the rest equal to $0$, and let $e_{k\ell}=e_ke_{\ell}^T+e_{\ell}e_k^T\in\R^{m\times m}$. Then, the omnibus matrix $M\in\R^{mn\times mn}$ is defined as 
\begin{equation*}
    M = \frac{1}{2}\sum_{k\leq \ell}e_{k\ell}\otimes\Big(A^{(k)}+A^{(\ell)}\Big)
\end{equation*}
As mentioned previously, $M$ is a special case of the generalized omnibus matrix $\fM$: it satisfies all the assumptions of Definition \ref{def:genOMNI}. For $q\in [m],$ the coefficient matrix $C^{(q)}$ is given by
$$
c_q^{(k,\ell)}=\begin{cases}
1&\text{ if }k=\ell=q\\
1/2&\text{ if }k\neq \ell,\text{ and }q\in\{k,\ell\}\\
0&\text{ else},
\end{cases}
$$
and so
$$
\alpha(s,q)=\begin{cases}
1+(m-1)/2&\text{ if }s=q\\
1/2&\text{ else}.
\end{cases}
$$
Fix $s \in [m]$. Then, the resulting covariance matrix is 
\begin{align*}
   \Check\Sigma_{\rho}(x;s)=\Big(\underbrace{\frac{m+3}{4m}}_{\text{method-induced coefficient}}+\underbrace{\frac{m+1}{2m^2}\sum_{q\neq s}\rho_{q,s}+\frac{1}{2m^2}\sum_{\substack{q<l \\ q,l\neq s}}\rho_{q,l}}_{\text{model-inherent coefficient}}\Big)\Sigma(x).
\end{align*}
As expected, the first term coincides with the covariance structure in classical OMNI \cite[Theorem 1]{levin_omni_2017} and further, the second term accounts for the pair-wise correlation of the adjacency matrices.
\end{example}
\noindent

Moving beyond the classical case, the following examples of generalized omnibus structures highlight the impact the structure of $\fM$ has on the limiting covariance (and hence on the limiting correlation; see Theorem \ref{thm:pfdiffgenomni}).
\noindent
\begin{example}[Total Average Omnibus]
\label{ex:total_ave}
\emph{In the total average case, letting $$\Bar{A}=\frac{A^{(1)}+A^{(2)}+\cdots+A^{(m)}}{m},$$ we define the generalized omnibus matrix $\fM_{\Bar{A}}$ via 
\[\fM_{\Bar{A}} =\begin{bmatrix}A^{(1)}&\Bar{A}&\cdots&\Bar{A}\\\Bar{A}&A^{(2)}&\cdots&\Bar{A}\\\vdots&\vdots&\ddots&\vdots\\\Bar{A}&\Bar{A}&\cdots&A^{(m)}\end{bmatrix}.\]
In this example, 
$$
c_q^{(k,\ell)}=\begin{cases}
1&\text{ if }k=\ell=q;\\
0&\text{ if }k=\ell\neq q;\\
1/m&\text{ else},
\end{cases}
$$
and so
$$
\alpha(s,q)=\begin{cases}
1+(m-1)/m&\text{ if }s=q\\
(m-1)/m&\text{ else}.
\end{cases}
$$
The associated coefficient of $\Sigma(x)$ in $\Sigma_{0}(x;s)$ (i.e., method coefficient) is then
\begin{align*}
\frac{1}{m^2}\left(\sum_{q=1}^{m}\alpha^2(s,q)\right)&=\frac{1}{m^2}\left((m-1)^3/m^2+(1+(m-1)/m)^2 \right)\\
&=\frac{m^2+m-1}{m^3},
\end{align*}
 and the model coefficient is 
 \begin{align*}
     \frac{2}{m^2}\sum_{q<l}\alpha(s,q)\alpha(s,l)\rho_{q,l}&= \frac{2}{m^2}\bigg(\frac{(m-1)(2m-1)}{m^2}\sum_{q\neq s}\rho_{q,s}+\frac{(m-1)^2}{m^2}\sum_{\substack{q<l \\ q,l\neq s}}\rho_{q,l}\bigg). \end{align*}
When $R\equiv 0_m$ (i.e., in the uncorrelated case), in the classical OMNI setting of Example 1 for large $m$ the limiting covariance $\Sigma_{0}(x;s)$ is approximately $\Sigma(x)/4$, and is not degenerate; 
in the total average omnibus setting for large $m$, the limiting covariance is approximately $\Sigma(x)/m\approx 0$.
This is sensible heuristically, as in that setting, we are effectively embedding $J_m\otimes \bar A\approx \TP$, and the correct scaling of the residuals would ideally be $\sqrt{nm}$ rather than $\sqrt{n}$. When $R= \Theta(1)$ and $m$ large, the model coefficient dominates the method coefficient, as even when $m$ is large, the model coefficient is not (approximately) degenerate.}

\end{example}


\begin{example}[Weighted Pairwise Average Omnibus]
\label{ex:wgtave}
\emph{In the classical omnibus setting, we have that 
$$\fM^{(k,\ell)}=\frac{A^{(k)}+A^{(\ell)}}{2},$$
and all matrices are effectively weighted equally in the omnibus matrix. 
This is sensible if all $A^{(k)}$ are i.i.d., but is, perhaps, less ideal in the setting where the networks are noisily observed with the level of noise varying in $k\in[m]$.
In that setting, 
the \textit{Weighted} Pairwise Average Omnibus matrix defined via 
\[\fM_W=\begin{bmatrix}A^{(1)}&\frac{w_1A^{(1)}+w_2A^{(2)}}{w_1+w_2}&\frac{w_1A^{(1)}+w_3A^{(3)}}{w_1+w_3}&\cdots&\frac{w_1A^{(1)}+w_mA^{(m)}}{w_1+w_m}\\\frac{w_2A^{(2)}+w_1A^{(1)}}{w_2+w_1}&A^{(2)}&\frac{w_2A^{(2)}+w_3A^{(3)}}{w_2+w_3}&\cdots&\frac{w_2A^{(2)}+w_mA^{(m)}}{w_2+w_m}\\\frac{w_3A^{(3)}+w_1A^{(1)}}{w_3+w_1}&\frac{w_3A^{(3)}+w_2A^{(2)}}{w_3+w_2}&A^{(3)}&\cdots&\frac{w_3A^{(3)}+w_mA^{(m)}}{w_3+w_m}\\\vdots&\vdots&&\ddots&\vdots\\\frac{w_mA^{(m)}+w_1A^{(1)}}{w_m+w_1}&\frac{w_mA^{(m)}+w_2A^{(2)}}{w_m+w_2}&\frac{w_mA^{(m)}+w_3A^{(3)}}{w_m+w_3}&\cdots&A^{(m)}\end{bmatrix}\] may be more appropriate.
In this case, the block entries of the omnibus matrix are defined via 
$$
\fM_W^{(k,\ell)}=\frac{w_kA^{(k)}+w_{\ell}A^{(\ell)}}{w_k+w_{\ell}}
$$ with weights $w_k>0$ $k,\ell\in[m]$.
In this example, 
$$
c_q^{(k,\ell)}=\begin{cases}
1&\text{ if }k=\ell=q\\
\frac{w_q}{w_q+w_{\ell}}&\text{ if }q=k\neq \ell\\
\frac{w_q}{w_q+w_k}&\text{ if }q=\ell\neq k\\
0&\text{ else},
\end{cases}
$$
and so
$$
\alpha(s,q)=\begin{cases}
1+\sum_{r\neq s} \frac{w_s}{w_s+w_{r}}&\text{ if }s=q\\
\frac{w_q}{w_q+w_s}&\text{ else}.
\end{cases}
$$
As (if each $w_k>0$) $\alpha(s,q)<1$ for $q\neq s$, 
we immediately have that $\alpha(s,s)>\alpha(s,q)$ for $q\neq s$.
From this, it follows that in the weighted pairwise average omnibus setting, the covariance matrix is given by
\begin{align*}
    \Check{\Sigma}_{\rho}(x;s) &= \frac{1}{m^2}\Bigg(\underbrace{\sum_{q\neq s}\frac{w_q^2}{(w_q+w_s)^2}+\Big(1+\sum_{r\neq s}\frac{w_s}{w_s+w_r}\Big)^2}_{\text{method-induced coefficient
        }}\\&+\underbrace{2\sum_{q\neq s}\Big(1+\sum_{r\neq s}\frac{w_s}{w_s+w_r}\Big)\frac{w_q}{w_q+w_s}\rho_{q,s}+2\sum_{\substack{q<l \\ q,l\neq s}}\frac{w_q w_{l}}{(w_q+w_s)(w_{l}+w_s)}\rho_{q,l}}_{\text{model-inherent coefficient
        }}\Bigg)\Sigma(x)
\end{align*}
While the method coefficient of $\Sigma_{0}(x;s)$ in $ \Check\Sigma_{\rho}(x;s)$ is not easily expressed in general, specific examples can nonetheless be instructive.
Consider the setting where $w_1=w$ and $w_k=1$ for all $k\neq 1$.
Considering $s=1$ in Theorem \ref{thm:genOMNI} provides that the method coefficient of $ \Check\Sigma_{\rho}(x;s)$ is 
\begin{align*}
\frac{1}{m^2}\left(\sum_{q=1}^{m}\alpha^2(s,q)\right)&=\frac{1}{m^2}\left( (m-1)\frac{1}{(1+w)^2}+\left(1+\frac{(m-1)w}{1+w}\right)^2 \right)\\
&=\frac{m-1+(1+mw)^2}{m^2(1+w)^2}.
\end{align*}
Similarly, the model coefficient is
\begin{align*}
        \frac{2}{m^2}\sum_{q<l}\alpha(s,q)\alpha(s,l)\rho_{q,l}&= \frac{2}{m^2}\bigg(\frac{1+mw}{(w+1)^2}\sum_{q\neq s}\rho_{q,s}+\frac{1}{(1+w)^2}\sum_{\substack{q<l \\ q,l\neq s}}\rho_{q,l}\bigg) 
\end{align*}
If $w\gg m$ is large, then the method coefficient is approximately equal to $1$ and the model coefficient is approximately equal to $0$, which is the limiting covariance achieved by embedding $A^{(1)}$ separately. 
This will be further explained in the context of limiting induced correlation in the next section.
If $w\approx 1$, then both model and method correlation are approximately equal to their analogues in the classical OMNI setting. 
If $w\ll 1$ is small, then the method coefficient is approximately $1/m$, and the model correlation is effectively a function of $\rho_{k,\ell}$ for $k,\ell\neq s$ for $m$ large (assuming, for the moment, all the inherent correlations are of the same relative order).
As the number of graphs increases, the relative impact on the overall embedding of $A^{(1)}$ decreases; this is in direct contrast to the classical OMNI setting, where each graph has the same (non-trivial) relative import in the embedded space.}
\end{example}

\subsection{Limiting inherent and induced correlation}
\label{sec:indcorr_gen}
In the generalized omnibus embedding, similar to the classical setting (see Theorem \ref{theorem:ind_lat_classical_omni}), we can precisely compute the limiting correlation between estimates of the same latent position in the embedded space.
To wit, we have the following theorem; as with the other main results, its proof can be found in the Appendix, specifically Section \ref{sec:app_pfdiffgenomni}.
\begin{theorem}
\label{thm:pfdiffgenomni}
With notation and assumptions as in Theorem \ref{thm:genOMNI}, 
consider fixed indices $i\in[n]$ and $s,s_1,s_2\in[m]$. Let $\Big(\hat{\bX}^{(s)}_{\fM_n}\Big)_i$ denote the estimated latent position from graph $A_n^{(s)}$ of the hidden latent position $X_i$. 
There exists a sequence of orthogonal matrices $(\mathcal{Q}_n)_{n=1}^\infty$ such that for all $z \in \R^d$, we have that
	\begin{align}
	    \label{eq:corCLT_genomni}
	    \lim_{n \rightarrow \infty}
	\p\left[ n^{1/2} \left[ \Big(\bhx^{(s_1)}_{\fM_n} -\bhx^{(s_2)}_{\fM_n}\Big) \mathcal{Q}_n\right]_{i}
	\le z \right]
	= \int_{\text{supp } F} \Phi\left(z, \widetilde\Sigma_{\rho}(x;s_1,s_2)) \right) dF(x), 
	\end{align}
	where
	 $\widetilde\Sigma_{\rho}(x;s_1,s_2)$ is given by
\begin{align}
\label{eq:induced_covGEN}
    \widetilde\Sigma_{\rho}(x;s_1,s_2)&=\frac{1}{m^2}\bigg(\underbrace{\sum_{q=1}^{m}\Big(\alpha(s_1,q)-\alpha(s_2,q)\Big)^2}_{\text{method-induced correlation}}\nonumber\\&\hspace{1cm}+\underbrace{2\sum_{q<l}\Big(\alpha(s_1,q)-\alpha(s_2,q))(\alpha(s_1,l)-\alpha(s_2,l)\Big)\rho_{q,l}}_{\text{model-inherent correlation}}\bigg)\Sigma(x)\\
    &=\frac{1}{m^2}\bigg(
    2\sum_{q<l}(\alpha(s_1,q)-\alpha(s_2,q))(\alpha(s_1,l)-\alpha(s_2,l))(\rho_{q,l}-1)\bigg)\Sigma(x).\notag
\end{align}
\end{theorem}

The covariance matrix $\widetilde\Sigma_{\rho}(x;s_1,s_2)$ from Theorem \ref{thm:pfdiffgenomni} can be expressed in terms of the limiting correlation $\rho(s_1,s_2)$, where $\rho(s_1,s_2)$ via:
\begin{align}
\label{eq:induced_corr}
    \rho(s_1,s_2)  =&\underbrace{1-\frac{\sum_{q=1}^{m}(\alpha(s_1,q)-\alpha(s_2,q))^2}{2m^2}}_{\text{method-induced correlation}}\notag\\
    &+ \underbrace{\frac{\sum_{q<l}\Big(\alpha(s_1,q)-\alpha(s_2,q))(\alpha(s_2,l)-\alpha(s_1,l)\Big)\rho_{q,l}}{m^2}}_{\text{model-inherent correlation}}.
\end{align}
Before considering the more exotic examples described above, let's recall the classical omnibus embedding setting.

\noindent{\bf Example \ref{ex:classicalomni} continued:} Consider fixed indices $i\in[n]$ and  $s_1,s_2\in[m]$. Then, the limiting correlation between two estimated latent positions $\Big(\bhx^{(s_1)}_{M_n}\Big)_{i},\Big(\bhx^{(s_2)}_{M_n}\Big)_{i}$  from graphs $A^{(s_1)}$, $A^{(s_2)}$ respectively is given by
    \begin{equation}
    \label{eq:induced_CLASSICAL}
        \rho(s_1,s_2)=\frac{3}{4}+\frac{1}{4}\rho_{s_1,s_2}.
    \end{equation}
To see this, note that in this case
$$
\alpha(s,q)=
\begin{cases}
1+(m-1)/2&\text{ if }s=q\\
1/2&\text{ otherwise }
\end{cases}
$$
Without loss of generality, suppose that $s_1<s_2$. Then,
\begin{align*}
    \frac{1}{m^2}\sum\limits_{q=1}^{m}(\alpha(s_1,q)-\alpha(s_2,q))^2&=\frac{1}{m^2}\Big((\alpha(s_1,s_1)-\alpha(s_2,s_1))^2+(\alpha(s_1,s_2)-\alpha(s_2,s_2))^2\Big)\\&=\frac{1}{m^2}(\frac{m^2}{4}+\frac{m^2}{4})=\frac{1}{2}
\end{align*}
and
\begin{align*}
    &\frac{1}{m^2}\sum_{q<l}\Big(\alpha(s_1,q)-\alpha(s_2,q))(\alpha(s_2,l)-\alpha(s_1,l)\Big)\rho_{q,l} \\
    &\hspace{2cm}= \frac{1}{m^2}\underbrace{\sum_{q=1}^{s_1-1}(\alpha(s_1,q)-\alpha(s_2,q))(\alpha(s_2,s_1)-\alpha(s_1,s_1))\rho_{q,s_1}}_{l=s_1} \\&\hspace{2.5cm}+\frac{1}{m^2}\underbrace{\sum_{q=1}^{s_2-1}(\alpha(s_1,q)-\alpha(s_2,q))(\alpha(s_2,s_2)-\alpha(s_1,s_2))\rho_{q,s_2}}_{l=s_2}
    \\&\hspace{2cm}=\frac{1}{m^2}\sum_{q=1}^{s_1-1}\underbrace{(\alpha(s_1,q)-\alpha(s_2,q))}_{:=0}\left(\frac{-m}{2}\right)\rho_{q,s_1}\\
    &\hspace{2.5cm}+\frac{1}{m^2}\sum_{q=1}^{s_2-1}\underbrace{(\alpha(s_1,q)-\alpha(s_2,q))}_{:\neq \text{0 only when }q=s_1}(\frac{m}{2})\rho_{q,s_2}\\
    &\hspace{2cm}=\frac{1}{m^2}(\frac{m}{2})(\frac{m}{2})\rho_{s_1,s_2}\\
    &\hspace{2cm}= \frac{1}{4}\rho_{s_1,s_2}.
\end{align*}
Plugging these values into Eq. (\ref{eq:induced_corr}) yields the desired result.

\vspace{2mm}

\noindent{\bf Example \ref{ex:total_ave} continued.}
In the total average case, we can show (similar to the classical case) that 
\begin{align*}
    \rho(s_1,s_2)&=\underbrace{1-\frac{1}{m^2}}_{\text{method-induced correlation}}+\underbrace{\frac{1}{m^2}\rho_{s_1,s_2}}_{\text{model-inherent correlation}}\ .
\end{align*}
This is sensible, because for large $m$, we have that $\fM_{\Bar{A}}\approx J_m\otimes\bar A$ and the embedding of each $A^{(k)}$ is effectively equivalent to repeatedly embedding $\bar A$, yielding the large $m$ correlation approximately equal to $1$.
Note also that when $m=2$, this correlation coincides with the pairwise (unweighted) average classical OMNI setting (i.e., the matrix structure as in Definition \ref{def:omni}). 
However, the covariance matrix $\Check\Sigma$ depends on $m$ in the $\fM_{\Bar{A}}$ setting, and for $m\geq 3$ the limiting embedded correlation for $\fM_{\Bar{A}}$ is always greater than the limiting embedded correlation for the classical $M$.
\vspace{2mm}

\noindent{\bf Example \ref{ex:wgtave} continued.}
In the weighted average OMNI setting, the general form of the limiting correlation is computationally unwieldy.
However, when 
$w_{s_1}=w_{s_2}=w$ and $w_k=1$, for $k\neq s_1,s_2$, a simple computation yields that the limiting correlation for the (scaled) row difference in Theorem \ref{thm:pfdiffgenomni} is given by (where, wlog, $s_1<s_2$)
\begin{align*}
 \rho(s_1,s_2)&=\underbrace{1-\frac{\big((m-1)w+1\big)^2}{m^2(1+w)^2}}_{\text{method correlation}}+\underbrace{\frac{\big((m-1)w+1\big)^2}{m^2(1+w)^2}\rho_{s_1,s_2}}_{\text{model correlation}}\\
 &=1-\frac{\big((m-1)w+1\big)^2}{m^2(1+w)^2}(1-\rho_{s_1,s_2}).
\end{align*}
Figure \ref{fig:fig1} further highlights this relationship between the matrix structure of $\fM_W$ and the limiting covariance for the weighted pairwise average case. 
In this example, we set the number of graphs to be $m=10$. 
The plot illustrates the correlation $\rho(s_1,s_2)$ between estimates of the  same true latent position across different values of the weights $w_{s_1}=w_{s_2}=w$, where $w_k=1$ for all $k\neq s_1,s_2$. 
The different line types correspond to different values of $\rho_{s_1,s_2}$, with the vertical red line at $w=1$ corresponding to the classical OMNI setting of \cite{levin_omni_2017}.
We see that as the weight $w$ increases,
the correlation $\rho(s_1,s_2)$ decreases (i.e., in the limit, the embedded estimates derived from $A^{(s_1)}$ and $A^{(s_2)}$ are ``less'' dependent on each other). 
On the contrary, as the weight $w$ decreases, the correlation $\rho(s_1,s_2)$ increases (towards $1$) as the $A^{(k)}$ for $k\neq s_1,s_2$ has the same outsized influence on the embeddings of $A^{(s_1)}$ and $A^{(s_2)}$ which has the effect of making the embeddings effectively identical.
  \begin{figure}[t!]
\centering
\includegraphics[width=.8\textwidth]{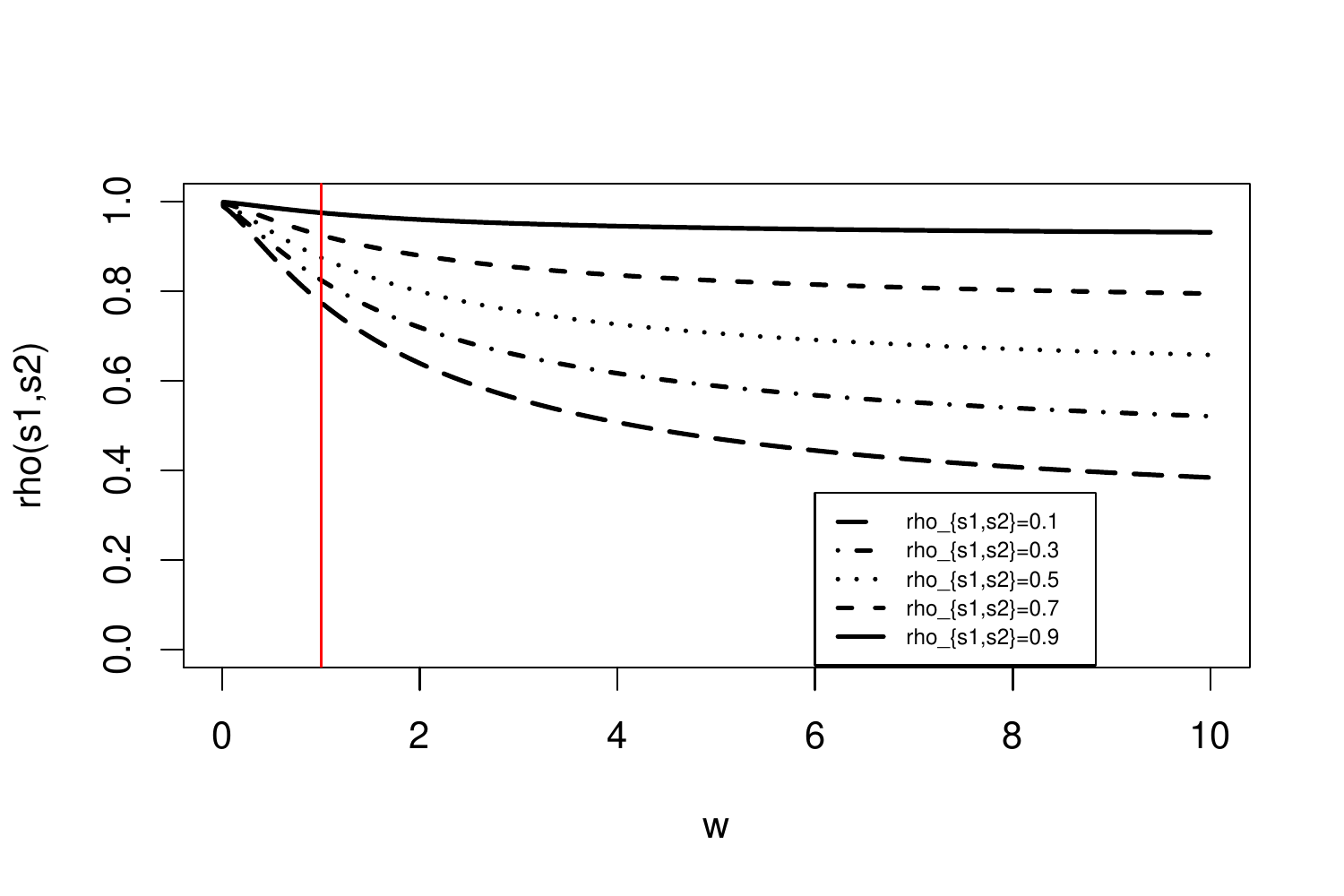}
\caption{The plot illustrates the correlation $\rho(s_1,s_2)$ between estimates of the  same true latent position across different values of the weights $w_{s_1}=w_{s_2}=w$, where $w_k=1$ for all $k\neq s_1,s_2$. The number of graphs is $m=10$. 
The different line types correspond to different values of $\rho_{s_1,s_2}$, with the vertical red line at $w=1$ corresponding to the classical OMNI setting of \cite{levin_omni_2017}.
}
  \label{fig:fig1}
\end{figure}

Indeed, the driving force behind the limiting correlation between the embedded $A^{(s_1)}$ and $A^{(s_2)}$ is the cumulative weights of the other $A^{(k)}$ for $k\neq s_1,s_2$.
If these weights are large, then the relative contribution to block-row $s_1$ (resp., $s_2$) by $A^{(s_1)}$ (resp., $A^{(s_2)}$) is low when compared to the cumulative effect of the other networks, and this masks the signal corresponding to $A^{(s_1)}$ (resp., $A^{(s_2)}$) in the $s_1$-th (resp., $s_2$-th) block of the embedding.
If, however, the weights of the $A^{(k)}$ for $k\neq s_1,s_2$ are small, then the relative contribution to block-row $s_1$ (resp., $s_2$) by $A^{(s_1)}$ (resp., $A^{(s_2)}$) is high when compared to the cumulative effect of the other networks, and this has the effect of reducing the correlation between the embeddings of $A^{(s_1)}$ and $A^{(s_2)}$.

\section{Dampened Omnibus and Aplysia spike train analysis}
\label{sec:damp}

In the example in Section \ref{sec:aplysia}, we see how the uniformity of the induced correlation in the classical omnibus embedding effectively masks much of the biologically relevant signal in the Aplysia motor program.
There is a dramatic spike in the data signal (the stimulus) followed by structured behavior and decay in spike intensity.
In this case (and in many other time series settings where we are modeling the impact of anomalous events), it is natural to assume that the dependence between two networks at different times $t_1,t_2$ decays as their difference $|t_1-t_2|$ gets bigger, and also decays over time as $t_1,t_2$ get bigger. 
This motivates the \textit{dampened} omnibus matrix defined via
$$\fM_{damp}^{(k,\ell)}=\begin{cases}
\frac{  w_k A^{(k)}+A^{(\ell)}}{w_k+1}&\text{ if }k>\ell,\\
A^{(k)}&\text{ if }k=\ell,\\
\frac{ A^{(k)}+ w_{\ell}A^{(\ell)}}{w_{\ell}+1}&\text{ if }k<\ell,
\end{cases}$$ 
where $\vec w$ is a strictly increasing vector of weights.
  \begin{figure}[t!]
\centering
\includegraphics[width=0.8\textwidth]{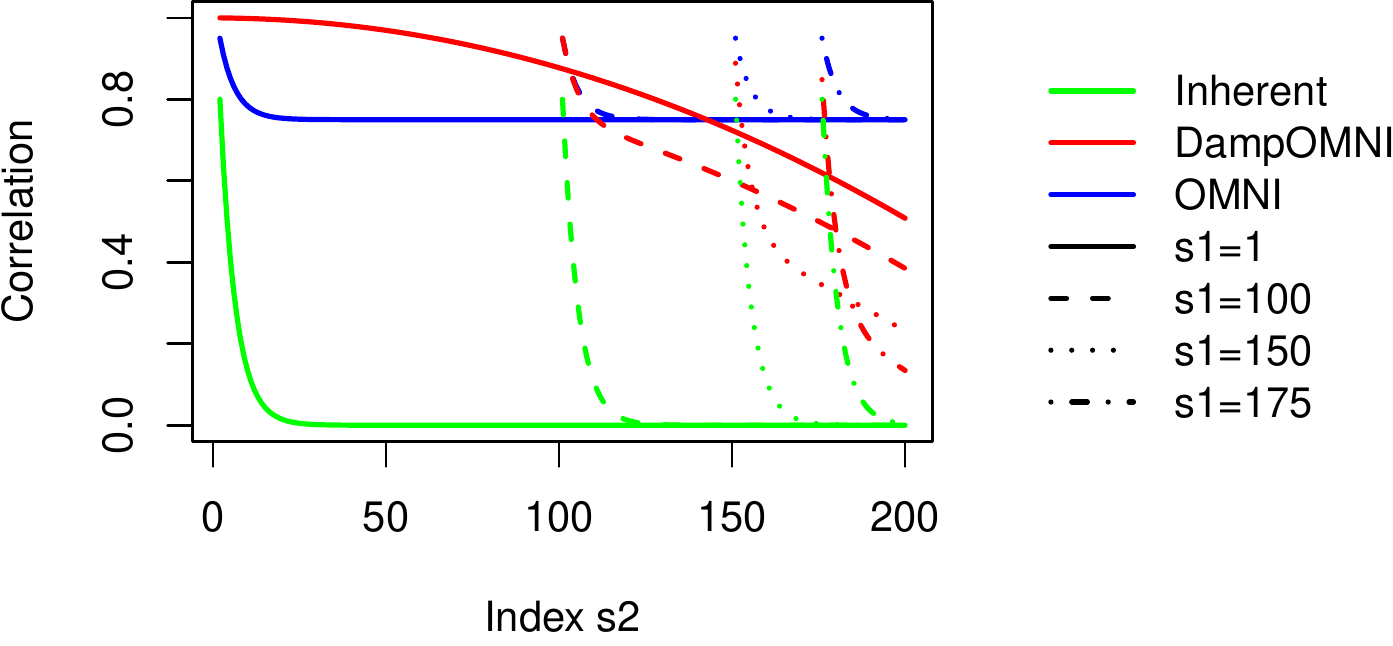}
\caption{
In a JRDPG$_{\text{for}}$ model where the entries of  $\boldsymbol{\varrho}$ are identically set to $0.8$, we consider the induced correlation for the classical omnibus embedding (blue), the dampened omnibus embedding (red).
  Considering $m=200$ networks, the 
solid (resp, dashed, dotted, dot-dash) lines correspond to the limiting induced correlation between vertices of the embedded $A^{(1)}$ (resp., $A^{(100)}$, $A^{(150)}$, $A^{(175)}$) and $A^{(s_2)}$ as $s_2$ varies from 2 (resp., 101, 151, 176) to 200. 
For each curve type, the true inherent correlation across the networks is plotted in green.}
  \label{fig:dampcorr}
\end{figure}
In the dampened omnibus matrix, $\fM_{damp}$, the relative contribution of each successive matrix (i.e., as $l>k$ increases in $A^{(l)}$) increases in block-row $k$; this (keeping the form of $\rho(t_1,t_2)$ in mind) can have the effect of decreasing the limiting induced correlation between $A^{(k)}$ and $A^{(l)}$ as $l>k$ increases.
As an example of this, consider the case where $w_{\ell}=\ell$ for all $\ell\in[m]$.
The form of (total) correlation, though cumbersome, is nonetheless instructive.
Letting 
$s_1,s_2\in[m]$, $s_1<s_2$, we have that 
the induced correlation is
  \begin{align*}
     \rho_{\text{me}}(s_1,s_2)=&1-\frac{1}{2m^2}\Bigg[\frac{(s_1-1)(s_2-s_1)^2}{(s_1+1)^2(s_2+1)^2}+\Big(\frac{s_1^2 +1}{s_1+1}+\sum_{l>s_1,\hspace{.1cm} l\neq s_2}\frac{1}{l+1}\Big)^2\\
     &\hspace{20mm}+\sum_{s_1<l<s_2}\Big(\frac{l}{l+1}-\frac{1}{s_2+1}\Big)^2 +\Big(\frac{(s_2-1)s_2+1}{s_2+1}+\sum_{l>s_2}\frac{1}{l+1}\Big)^2\Bigg]
\end{align*}
while the effect of the inherent correlation is
\begin{align*}
    \rho_{\text{mo}}(s_1,s_2)
&=\frac{1}{m^2}\Bigg[\sum_{l<s_1}\sum_{q<l}\frac{(s_2-s_1)(s_1-s_2)}{(s_1+1)^2(s_2+1)^2}\rho_{q,l}\\
     &\hspace{15mm}+\sum_{q<s_1}\frac{-(s_2-s_1)}{(s_1+1)(s_2+1)}\Big(\frac{s_1^2 +1}{s_1+1}+\sum_{h>s_1,\hspace{.1cm} h\neq s_2}\frac{1}{h+1}\Big)\rho_{q,s_1}\\
     &\hspace{15mm}+\sum_{s_1<l<s_2}
     \Bigg(
     \sum_{q<s_1}\frac{s_2-s_1}{(s_1+1)(s_2+1)}\left(\frac{1}{s_2+1}-\frac{l}{l+1} \right)\rho_{q,l}\\
     &\hspace{35mm}+ \Big(\frac{s_1^2 +1}{s_1+1}+\sum_{h>s_1,\hspace{.1cm} h\neq s_2}\frac{1}{h+1}\Big)\left(\frac{1}{s_2+1}-\frac{l}{l+1} \right)\rho_{s_1,l}\\
     &\hspace{35mm}+\sum_{s_1<q<l}\left(\frac{q}{q+1}-\frac{1}{s_2+1} \right) \left(\frac{1}{s_2+1}-\frac{l}{l+1} \right)\rho_{q,l}
     \Bigg)\\
     &\hspace{15mm}+
     \sum_{q<s_1}\frac{s_2-s_1}{(s_1+1)(s_2+1)}\left(\frac{s_2^2-s_2+1}{s_2+1}+\sum_{h>s_2}\frac{1}{h+1}\right)\rho_{q,s_2}\\
     &\hspace{15mm}+ \Big(\frac{s_1^2 +1}{s_1+1}+\sum_{l>s_1,\hspace{.1cm} l\neq s_2}\frac{1}{l+1}\Big)\left(\frac{s_2^2-s_2+1}{s_2+1}+\sum_{h>s_2}\frac{1}{h+1}\right)\rho_{s_1,s_2}\\
     &\hspace{15mm}+\sum_{s_1<q<s_2}\left(\frac{q}{q+1}-\frac{1}{s_2+1} \right) \left(\frac{s_2^2-s_2+1}{s_2+1}+\sum_{h>s_2}\frac{1}{h+1}\right)\rho_{q,s_2}\Bigg].
  \end{align*}
  With the form of $\rho_{\text{me}}$ above, consider the situation in which $s_2\gg s_1$.
  In this case, $\rho_{\text{me}}$ is growing like $\approx 1-s^2_2/2m^2$, which is decaying to $1/2$ as $s_2$ approaches $m$. 
  When both $s_2, s_1\gg 1$, $\rho_{\text{me}}$ is growing like $\approx 1-(s^2_2+s_1^2)/2m^2$, which is decaying to $0$ as $s_2, s_1$ both approach $m$. 
  This corroborates the effect we see in Figure \ref{fig:dampcorr}, where $\rho(s_1,s_2)$ is tracking the inherent (i.e., model) correlation fairly well after an initial burn in period (i.e., when $s_1$ is sufficiently large).
  For $\rho_{\text{mo}}$ (ignoring for the moment the potential large differences in the inherent correlation), when $s_2,s_1\gg (s_2-s_1)$, the dominant term is $\approx \frac{1}{m^2}s_1s_2\rho_{s_1,s_2}$ which is tracking the true inherent correlation.
  When both $s_2,s_1\gg 1$ but their difference is not necessarily relatively small, the dominant term in $\rho_{\text{mo}}$ is of the order $$\approx \frac{1}{m^2}s_1s_2\rho_{s_1,s_2}+\frac{1}{m^2}s_2\sum_{s_1<q<s_2}\rho_{q,s_2}\ .$$
  We see that this only depends on $\rho_{q,s_2}$ for $q\in\{s_1+1,\ldots,s_2-1\}$ and the effect of the inherent correlation is effectively localized.
  

In Figure \ref{fig:dampcorr}, we consider the JRDPG$_{\text{for}}$ model of Section \ref{sec:FWD} where the entries of  $\boldsymbol{\varrho}$ are identically set to $0.8$ and $m=200$. 
We plot the induced correlation for the classical omnibus embedding (blue), the dampened omnibus embedding (red), where the 
solid (resp, dashed, dotted, dot-dash) lines correspond to the limiting induced correlation between vertices of the embedded $A^{(1)}$ (resp., $A^{(100)}$, $A^{(150)}$, $A^{(175)}$) and $A^{(s_2)}$ as $s_2$ varies from 2 (resp., 101, 151, 176) to 200. 
For each curve type, the true inherent correlation across the networks is plotted in green.
While the method correlation in classical OMNI prevents the induced correlation from ever tracking the true inherent correlation, we see that the more nuanced structure presented by dampened OMNI allows (after a suitable burn in period) the induced correlation to significantly better (though far from ideally) track the inherent correlation.
Developing $\fM$ to induce a given correlation structure is a natural next step, and we are actively pursuing this question.
This decaying correlation structure appears well-suited for the Aplysia motor program, as it precisely allows for a large correlation at early time-points (corresponding to the dampening of signal after an anomalous event)
that dissipates in time.
We will explore this further in the next section.

\subsection{Dampened spike train analysis}
\label{sec:SPIKE}

Since the generalized omnibus embedding permits additional degrees of freedom for the off-diagonal entries, a richer spectrum of induced correlation is possible.  The dampened omnibus embedding, in particular, is designed to weaken cross-graph interplay over time. Given that the correlation homogeneity of the classical omnibus makes inter-network changes less apparent, it is natural to ask if an appropriately-calibrated dampened omnibus embedding can detect what the more limited classical omnibus can not. To test this, we next apply the methodology of Section \ref{sec:masking} to the Aplysia data of Section \ref{sec:aplysia} using $\fM_{damp}$ with $w_{\ell}=\ell$.
Results are summarized in Figure \ref{fig:domni_aplysia} 
where ${\bf D}_d$ is the corresponding distance matrix from Eq.\@ \ref{eq:d} in the dampened setting.

We again see that dampened OMNI is able to isolate the stimulus in the second graph.
However, unlike in the standard OMNI setting, dampened OMNI is able to tease out additional relevant structure including the uniqueness of graph 1, the transition from gallop to crawl and an unstable dynamic in the crawling motor program.
In this dampened setting 
the cluster labels found by \texttt{Mclust} are summarized in Table \ref{tab:dampened_aplysia}.
\begin{figure}[t!]
\centering
\includegraphics[width=1\textwidth]{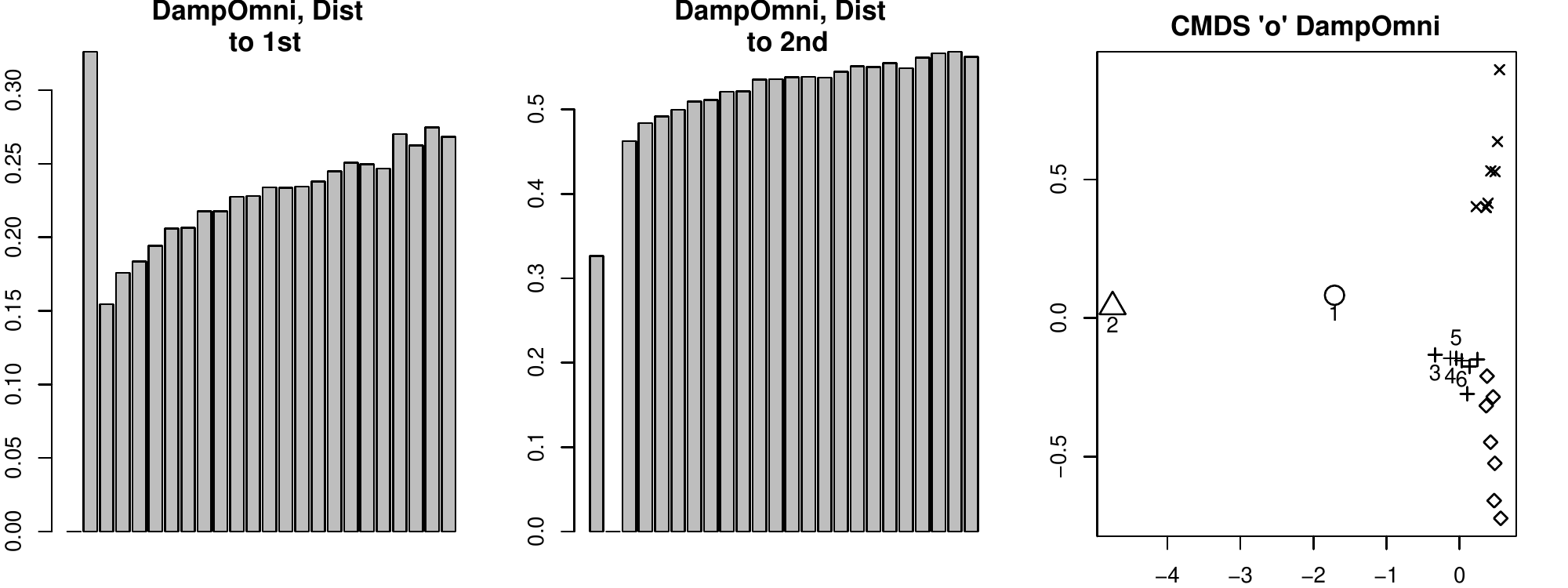}
\caption{
Using $w_{\ell}=\ell$ in the Dampened OMNI framework, in the left (resp., center) panel, we plot the average vertex distance (as the bar heights) in the embedding between graph 1 (resp., graph 2) and graph $k$ for each $k\in[m]$ in Dampened OMNI (where the bars are labeled $k=1,2,\cdots,24$).
In the right panel, we compute the $24\times 24$ distance matrices ${\bf D_d}$, and embed it into $\mathbb{R}^2$.
The resulting $24$ data points are clustered using \texttt{Mclust} (clusters denoted by shape), and are plotted, with graphs 1 through 6 further labeled with their corresponding number.
}
  \label{fig:domni_aplysia}
\end{figure}
\begin{table}[t!]
\begin{center}
    \begin{tabular}{c|c|c|c|c|c|c|c|c|c|c|c|c}
    Graph& 1 & 2 & 3 & 4  &5&  6 & 7 & 8 & 9 &10 &11 &12\\\hline
    Cluster & 1 &2& 3 &3& 3& 3 &3& 3 &3 &4& 4& 4\\
    \hline\hline
    Graph& 13 &14 &15& 16& 17 &18& 19& 20 &21& 22 &23& 24\\\hline
    Cluster& 4 &5 &5 &4 &4 &5 &4 &5& 5& 4& 5& 5\\
    \end{tabular}
\end{center}
    \caption{\label{tab:dampened_aplysia}
This table displays the cluster labels for the 24 networks when embedding the time-series using dampened OMNI with weights $w_{\ell}=\ell$, then embedding the across-graph distance matrix using classical MDS scaling, and finally clustering the graphs using \texttt{Mclust}.
    In the MDS embedding, each graph is represented by a single $2-$dimensional embedded data point, and the pair (Graph, Cluster) describes which cluster these points are assigned by \texttt{Mclust}.}
    \end{table}
From the table, we further see that the clustering is better able to isolate graph 1 (the spontaneous firing state graph) than in the classical omnibus setting.
Moreover, the clustering of graphs 3--9 and the CMDS graph embedding, shows that the dampened setting is better able to capture (imperfectly) the transition from galloping in bins 3 and 4 to crawling in bins 5--24.
Clustering yields classes that progress in orderly fashion over the first half of the motor program, after which states irregularly alternate between 4 and 5, suggesting an unstable dynamic to the program not apparent from simple visual inspection of the firing traces.

Here, the dampened omnibus structure allows us to better tease out the biologically relevant structure in the data; namely, the distinct nature of graph 1, the transition from gallop to crawl, the unstable dynamic in the crawling.
Indeed, from Figure \ref{fig:omniaplysia}, we see that the flat correlation structure induced by classical OMNI masks much of this latent structure in these biological networks, while the dampened structure induces the right correlation and both uncovers and clarifies this neuroscientifically relevant signal in the data.


	\section{Induced correlation in JRDPG{\text{gen}} model}
 \label{sec:experiments-simulations2}
In this section, we  present a series of simulations that are designed further illuminate the flexibility of the generalized omnibus framework in the JRDPG$_{\text{gen}}$ model (having explored dampened OMNI in the JRDPG$_{\text{for}}$ model in Section \ref{sec:damp}).
If we consider a collection of graphs as noisy realizations of some true, underlying, network process (i.e., the JRDPG$_{\text{gen}}$ model), then it is reasonable to expect some of the realizations to be noisier, or of lower fidelity, than others.  
As such, it is desirable to have an embedding method that can account for the difference in sampling fidelity across networks.
To explore this further, we consider the JRDPG$_{\text{gen}}$ model with $\nu_i=0.8$ for $i=1,2,\cdots,50$, and $\nu_i=0.3$ for $i=51,52,\cdots,100$ (so that $m=100$).
In this case, the first 50 $A^{(i)}$ are higher fidelity (i.e., less noisy) copies of the generator $A^{(0)}$.
It is natural then to seek to down-weight $A^{(i)}$ for $i=51,52,\ldots,100$ in the embedding.
We can achieve this via the weighted average omnibus embedding with
$$ \fM_{W}^{(i,j)}=\frac{w_iA^{(i)}+w_jA^{(j)}}{w_i+w_j}
$$
where $\vec{w}$ is a vector of positive graph weights.
In Table \ref{tab:onegen}, we consider the limiting embedding correlation across the networks within this model for a variety of different weight vectors.
We see that upweighting the higher-fidelity networks in $\fM_{W}$ has the effect of more closely modeling the true inherent correlation among the high fidelity networks in the embedded space versus classical OMNI (the $w_i=1$ for all $i$ setting). 
This comes at the expense of inducing more correlation between the lower fidelity samples (and across the lower fidelity and higher fidelity samples), as the higher fidelity samples are having more influence on the embeddings of the lower fidelity samples. 
We note also that the effect is reversed if the lower fidelity samples are upweighted.
\begin{table}[t!]
\centering
\begin{tabular}{c|c|c|c|c|c}
 & inherent corr.& $w_i\stackrel{iid}{\sim}$Unif(0,1)& $w=\vec 1_{100}$ & $w=[\vec 1_{50}, \vec{10}_{50}]$ & $w=[\vec 10_{50}, \vec{1}_{50}]$ \\\hline
  $\rho(1,2)$   &0.64 & 0.933 & 0.91 & 0.969  & 0.821\\\hline
  $\rho(1,51)$  &0.24 & 0.823 & 0.81 & 0.729  & 0.839\\\hline
  $\rho(51,52)$ &0.09 & 0.684& 0.773 &0.548  & 0.921
\end{tabular}
\caption{Considering the JRDPG$_{\text{gen}}$ model with $\nu_i=0.8$ for $i=1,2,\cdots,50$, and $\nu_i=0.3$ for $i=51,52,\cdots,100$ (so that $m=100$), we provide the correlations induced by embedding $\fM_{W}$ for a variety of weight vectors $\vec w$. Values are rounded to three digits for ease of display.}
\label{tab:onegen}
\end{table}

In the model above, suppose that we want the embedding to preserve the correlation between one particular pair of graphs (wlog, between $A^{(1)}$ and $A^{(2)}$).
This can be achieved via a special OMNI construction as follows.
Letting $m$ be even, for $i$ odd define
$$\fM_{12}^{(i,j)}=
\begin{cases}
A^{(i)} &\text{ if }i=j\text{ or }i\leq j-2 \\
\frac{A^{(i)}+A^{(j)}}{2} &\text{ if }i=j-1 \\
A^{(j)} &\text{ if }i>j,
\end{cases}
$$
for $i$ even define
$$\fM_{12}^{(i,j)}=
\begin{cases}
A^{(i)} &\text{ if }i=j\text{ or }i\leq j-1 \\
\frac{A^{(i)}+A^{(j)}}{2} &\text{ if }i=j+1 \\
A^{(j)} &\text{ if }i-1>j.
\end{cases}
$$
Then it is not difficult to compute
$$\rho(1,2)=1-(1-R_{1,2})\frac{(m-1)^2}{m^2},$$
so that $\rho(1,2)\approx R_{1,2}$ for large $m$.
While this choice of $\fM_{12}$ may not globally preserve $R$ in the embedded space, this demonstrates that local (i.e., between pairs or a few pairs) can be well-preserved.
In Figure \ref{fig:GenOmni}, we see that (considering again the JRDPG$_{\text{gen}}$ model with $\nu_i=0.8$ for $i=1,2,\cdots,50$, and $\nu_i=0.3$ for $i=51,52,\cdots,100$), the pairwise correlation is well preserved by embedding via $\fM_{12}$, not just between $A^{(1)}$ and $A^{(2)}$ but between many of the network pairs $A^{(i)}$ and $A^{(j)}$ for $i,j$ small.

\begin{figure}[t!]
\centering
  \includegraphics[width=.7\linewidth]{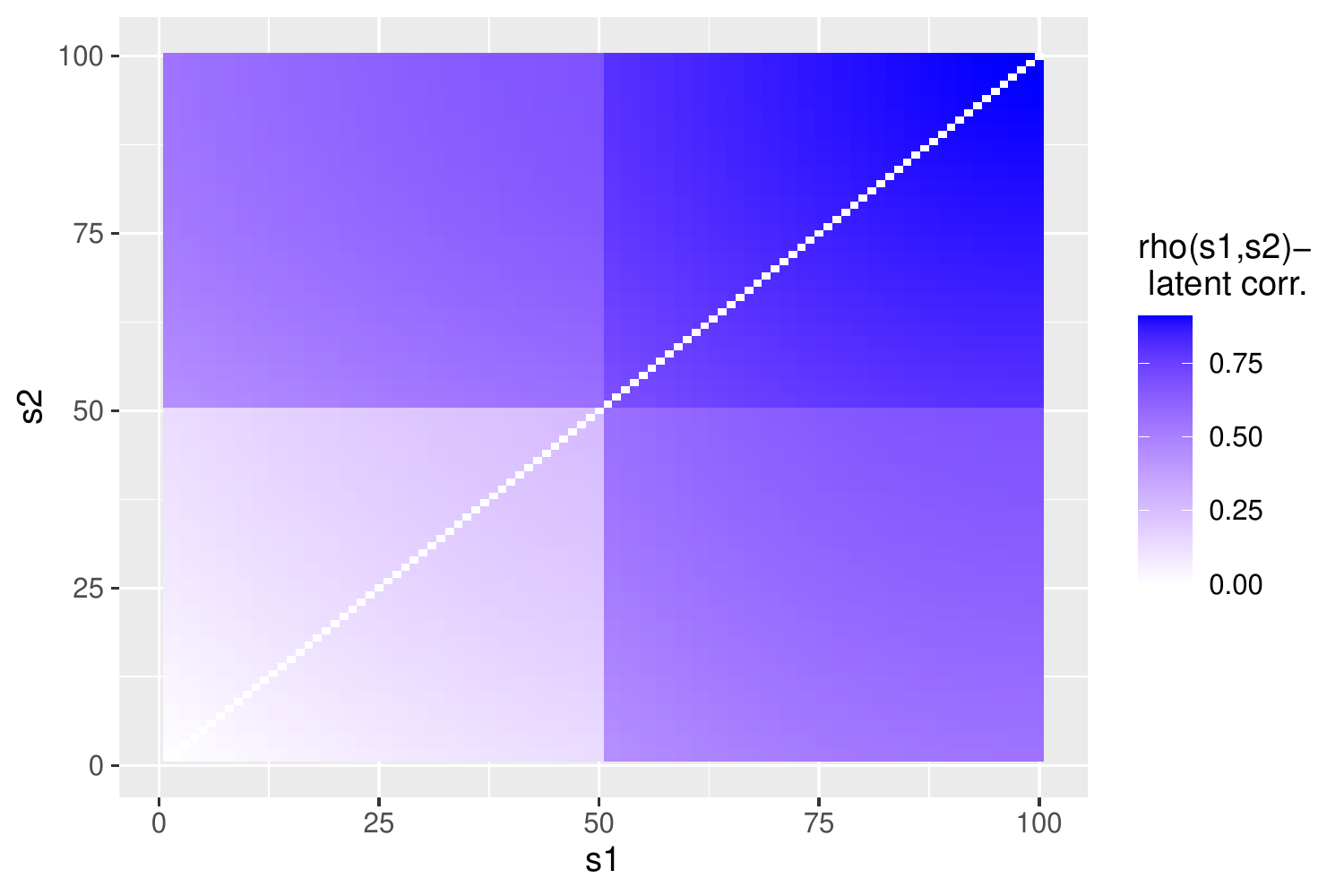}
  \caption{Considering the JRDPG$_{\text{gen}}$ model with $\nu_i=0.8$ for $i=1,2,\cdots,50$, and $\nu_i=0.3$ for $i=51,52,\cdots,100$ (so that $m=100$), we plot the difference between induced and inherent correlations when embedding via $\fM_{12}$.  Darker blue indicates less accurate preservation of the correlation structure.}
\label{fig:GenOmni}
\end{figure}

\section{Effective sample size in correlated graphs} \label{sec:effective-sample-size}

Statistical inference for multiple networks often faces questions that require to aggregate the information about the latent positions across samples of graphs. When the graphs are correlated, correlation itself adds a dampening factor that can reduce the effective sample size. In this subsection, we investigate the effect of correlation (induced and inherent) on effective sample size in subsequent inference about the latent position.

The following theorem describes the limiting behavior of the average latent positions estimated from a sample of correlated graphs. 
The proof  is given in the Appendix.
\begin{theorem}
\label{thm:CLT-effective-samplesize}

Let $\rho\in(0,1)$ be fixed.
Let $F$ be a distribution on a set $\mathcal{X}\subset \mathbb{R}^{d}$, where $\langle x, x'\rangle\in[0,1]$ for all $x,x'\in\mathcal{X}$, and assume that $\Delta := \mathbb{E}[X_1 X_1^{T}]$ is rank $d$.
Let $(A^{(1)}_n,\ldots, A^{(m)}_n,\bX_n)\sim \mathrm{JRDPG}(F,n,m,\rho),$ be a sequence of $\rho-$correlated $\mathrm{RDPG}$ random graphs and associated latent positions.
\begin{enumerate}
	\item[a)] Letting 
	$\bhx^{(s)}_{A_n}=\mathrm{ASE}(A^{(s)}_n,d)$ for each $s=1,\ldots,m$, 
	there exist sequences of orthogonal $d$-by-$d$ matrices
	$( W^{(s)}_n )_{n=1}^\infty,$ $s=1, \ldots,m$ such that for all $z \in \R^d$ and for any fixed index $i$, 
	\begin{align}
	\label{eq:corCLT-effective-samplesizeASE}
	\lim_{n \rightarrow \infty}
	\p\left[ n^{1/2} \left(\frac{1}{m}\sum_{s=1}^m \widehat \bX_{A_n}^{(s)} W^{(s)}_n -  \bX_n \right)_i
	\le z \right]
	= \int_{\text{supp } F} \Phi\left(z, \widetilde\Sigma\left(x, \frac{1-\rho}{m} + \rho\right) \right) dF(x).
	\end{align}
	\item[b)] Letting 
	$\bhx_{M_n}=\mathrm{ASE}(M_n,d)$, and denoting the $s$-th $n\times d$ block of $\bhx_{M_n}$ as $\bhx_{M_n}^{(s)}$, 
	there exist a sequence of orthogonal $d$-by-$d$ matrices
	$( \tilde W_n )_{n=1}^\infty,$  such that for all $z \in \R^d$ and for any fixed index $i$, 
	\begin{align}
	\label{eq:corCLT-effective-samplesizeOMNI}
	\lim_{n \rightarrow \infty}
	\p\left[ n^{1/2} \left(\frac{1}{m}\sum_{s=1}^m \widehat \bX_{M_n}^{(s)} \tilde W_n -  \bX_n \right)_i
	\le z \right]
	= \int_{\text{supp } F} \Phi\left(z, \widetilde\Sigma\left(x, \frac{1-\rho}{m} + \rho\right) \right) dF(x).
	\end{align}
\end{enumerate}
\end{theorem}

Stated simply, the previous theorem recovers the classical result about the covariance of the sample mean for correlated data. When the graphs are independent, the average of the estimated latent positions after a proper orthogonal alignment adds a factor of $1/m$ to the limiting covariance matrix, suggesting that this average concentrates for large sample size around the true mean. This is shown both for  a separate embedding of each graph (part a)) or a joint embedding with OMNI (part b)).
When the edges of the graphs are correlated by some positive constant $\rho$, the efficiency of this estimator is reduced, and the limiting effective sample size is $m_{\text{eff}}:=\frac{m}{1+\rho(m-1)}$.

\subsection{Recovery of community labels in stochastic blockmodels}
\label{sec:community-recovery-experiments}
We start by investigating the recovery of community labels from estimated latent positions of stochastic block model graphs. First, we consider a correlated pair of positive semidefinite $K$-block stochastic block models \cite{Holland1983} as in  \cite{athreya_survey}, Definition 8. 
In particular, letting $K=2$ we generate a pair of adjacency matrices $(A,B,\bf X)\sim$ $\mathrm{JRDPG}(F,n,2,\rho)$, where $F$ is a mixture of point mass distributions defined by
\begin{equation}
F=\frac{1}{2}\delta_{\zeta_1}+\frac{1}{2}\delta_{\zeta_2},
\end{equation}
in which $\zeta_1$, $\zeta_2\in\mathbb{R}^2$ are the latent positions and satisfy 
$$
\begin{bmatrix}
\zeta_1\\
\zeta_2
\end{bmatrix}
\begin{bmatrix}
\zeta_1\\
\zeta_2
\end{bmatrix}^T=
\begin{bmatrix}
0.5&0.5\\
0.5&0.5+\epsilon
\end{bmatrix}
=:
P_{\epsilon}
$$
where $\epsilon>0$. 
\begin{figure}[t!]
\centering
\includegraphics[width=1\textwidth]{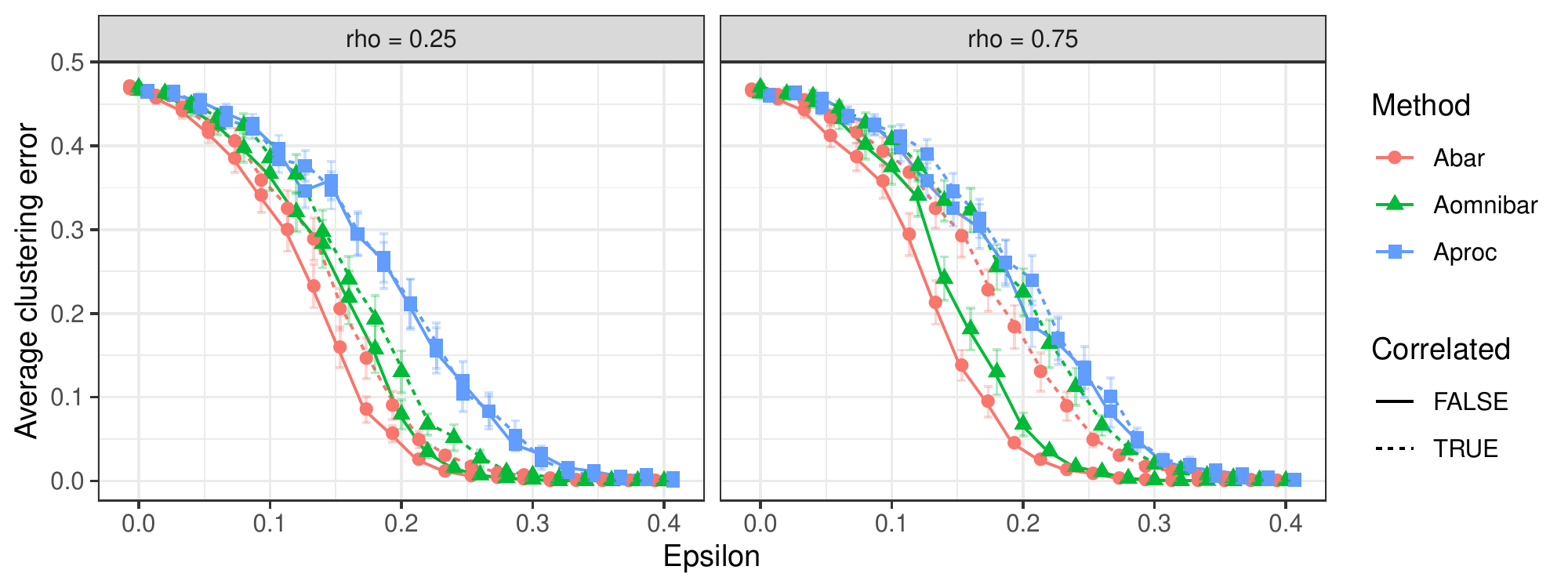}
\caption{Community detection error in a $\rho-$SBM model with $K=2$ communities, $N=100$ number of vertices and block probability matrix $P_{\epsilon}$ as a function of $\epsilon$. The latent positions are estimated by ASE embedding of the mean graph $\Bar{A}$ (red lines), OMNI embedding (green lines) and Procrustes-based pairwise embedding (blue lines). Each point is the mean of $100$ MC replicates. Observe that in both panels the ASE($\Bar{A},2$) and OMNI($M,2$) embeddings in the i.i.d case perform better than their embeddings in the correlated case, while the Procrustes-based pairwise embedding remains unaffected by the effect of correlation. Also note that the performance of the Procrustes-based pairwise embedding is inferior to the other two embedding methods.}
\label{fig:epsilon}
\end{figure}
We let $K=d=2$ and we consider the following three graph embedding techniques to help us understand the role of effective sample size on correlated random graphs: 
\begin{enumerate}
    \item \textbf{Classical OMNI embedding:} We apply the traditional omnibus embedding using $M$ as in Equation \eqref{eq:omnibus_M} to obtain estimated latent positions $\mathrm{ASE}(M;2)\in\R^{2n\times 2}$ and we average the $n$ rows across the two graphs to obtain $\widehat{\bX}^1\in\R^{n\times 2}$.
    
    \item \textbf{Mean graph embedding:} We use ASE to embed the mean of the two graphs, $\Bar{A}
= \frac{A + B}{2}$ \cite{tang2018connectome} to obtain estimated latent positions $\widehat{\bX}^2=\mathrm{ASE}(\Bar{A};2)\in\R^{n\times 2}$.

    \item \textbf{Procrustes-based embedding:} We separately embed the graphs $A,B$, obtaining two separate matrices of estimated
latent positions $\mathrm{ASE}(A,2),\mathrm{ASE}(B,2)\in\R^{n\times 2}$. We then align these two matrices via orthognal Procrustes
alignment \cite{gower_procrustes}, and average the aligned embeddings to obtain $\widehat{\bX}^3\in\R^{n\times 2}$.\end{enumerate}

%
\noindent For each embedding strategy, we further consider $\rho=0$ (i.e., conditionally independent graphs) and $\rho>0$ (i.e., correlated graphs) in our generative model resulting in a total of six estimated latent position matrices, $\widehat{\bX}_{\textrm{ind}}^{1},\widehat{\bX}_{\textrm{corr}}^{1},\widehat{\bX}_{\textrm{ind}}^{2},\widehat{\bX}_{\textrm{corr}}^{2},\widehat{\bX}_{\textrm{ind}}^{3},\widehat{\bX}_{\textrm{corr}}^{3}$. 
\begin{figure}[t!]
\centering
\includegraphics[width=1\textwidth]{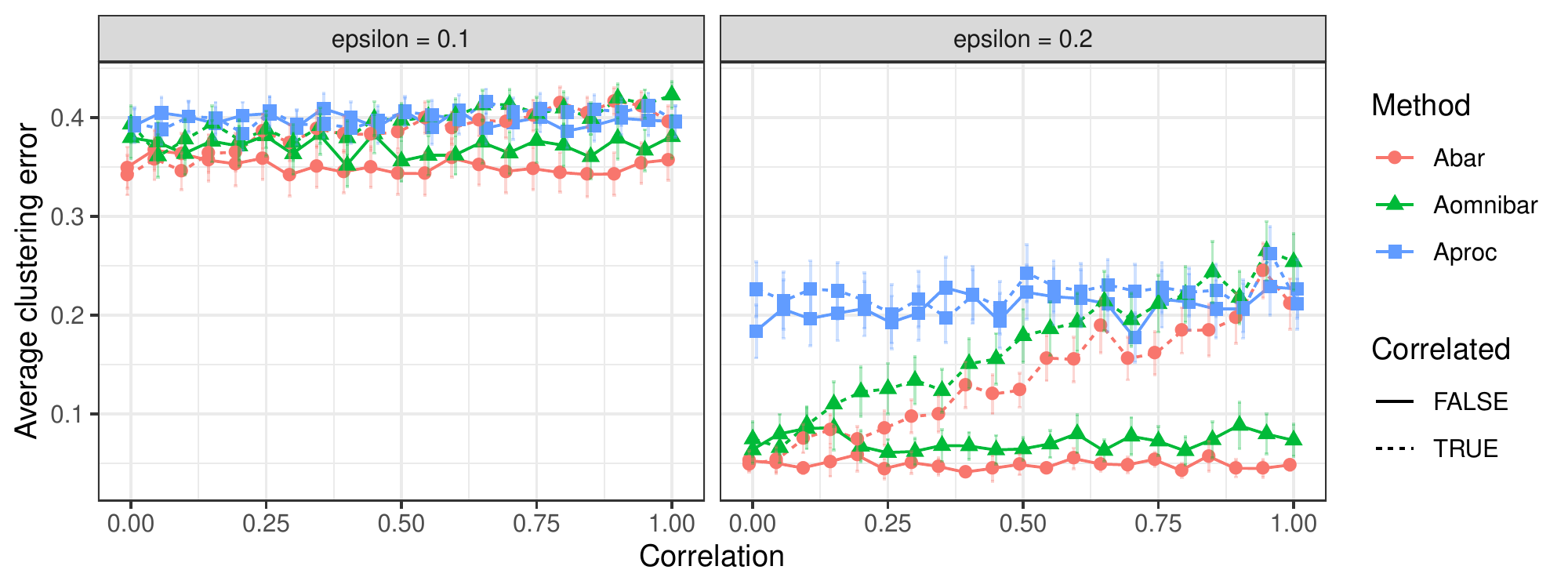}
\caption{Clustering error in recovery of community labels of estimated latent positions in a $\rho$-SBM model with $K=2$ communities, $N=100$ number of vertices and block probability matrix $P_{\epsilon}$ as a function of the correlation $\rho$. The latent positions are estimated by ASE embedding of the mean graph $\Bar{A}$ (red lines), OMNI embedding (green lines) and Procrustes-based pairwise embedding (blue lines). 
 Each point is the mean of $100$ MC replicates. In panel $(a)$, all the methods effectively achieve the same clustering error. In panel $(b)$, the performance of the ASE($\Bar{A},2$) and OMNI($M,2$) embeddings in the correlated case decreases as the correlation increases.
}
\label{fig:rho}
\end{figure}

%
%

In Figures \ref{fig:epsilon}-\ref{fig:nn}, we cluster the rows of estimated latent positions $\widehat{\bX}_{\ell}^{\kappa}$, for all $\kappa\in\{1,2,3\}$, $\ell\in\{\text{ind},\text{corr}\}$ into $K=2$ communities via the model based clustering procedure in the \texttt{Mclust} R package \cite{mclust2012}. It is shown in \cite{tang_priebe_16}, Section 4, that Gaussian mixture model-based clustering yields substantial improvement over $K$-means clustering in recovering the latent communities in ASE of stochastic blockmodels.
This is due to the associated limiting covariance  matrices often being elliptical rather than spherical. 
In each figure, we plot the clustering error in recovering the latent community labels for all six estimated latent position matrices where the solid lines correspond to the independent case and the dashed lines correspond to the correlated case. 
The latent positions estimated by ASE($\Bar{A},2$) are displayed in incarnadine/red lines, by OMNI embedding in green lines and by Procrustes-based pairwise embedding in blue lines. 
Each point is the average of 100 MC replicates.

\begin{figure}[t!]
\centering
\includegraphics[width=0.9\textwidth]{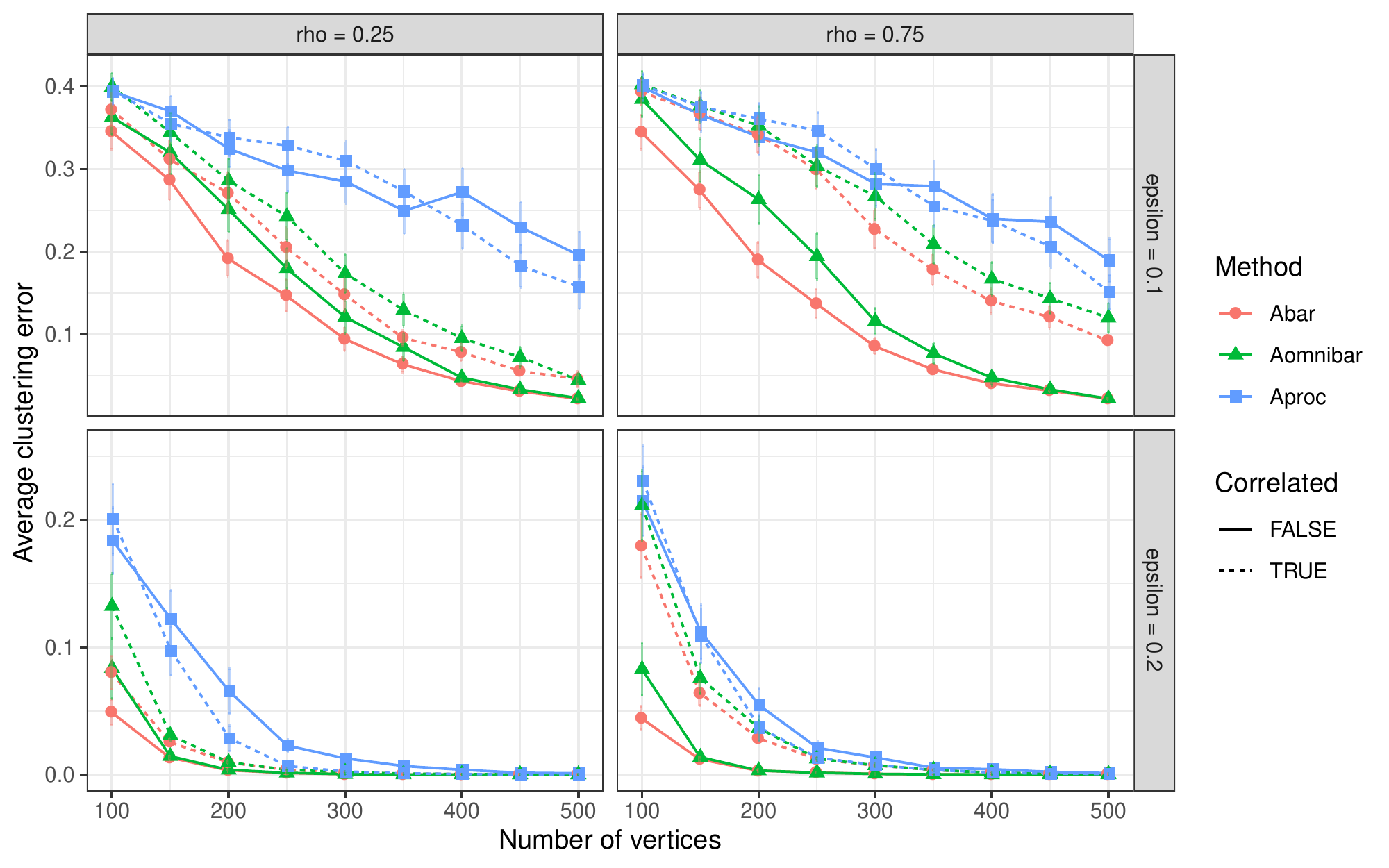}
\caption{Clustering error in recovery of community labels of estimated latent positions in a $\rho$-SBM model with $K=2$ communities and block probability matrix $P_{\epsilon}$ as a function of the number of vertices $N$. The latent positions are estimated by ASE embedding of the mean graph $\Bar{A}$ (red lines), OMNI embedding (green lines) and Procrustes-based pairwise embedding (blue lines). Each point is the mean of $100$ MC replicates.}
\label{fig:nn}
\end{figure}
In Figure \ref{fig:epsilon}, we set the number of vertices for each graph equal to $N=100$ and we plot the clustering error versus $\epsilon$. 
In panel $(a)$ (resp. $(b)$) we set the correlation for the estimated latent positions in the correlated case equal to $\rho=0.25$ (resp. $\rho=0.75$).
Observe that in both panels the ASE($\Bar{A},2$) and OMNI($M,2$) embeddings in the i.i.d case perform better than their corresponding embeddings in the correlated case, while the Procrustes-based pairwise embedding remains unaffected by the inherent correlation; we suspect this is due to the Procrustes step introducing additional signal across the networks which can mask the present inherent correlation. 
Also note that the performance of the Procrustes-based pairwise embedding is inferior to the other two embedding methods, while the omnibus and mean embeddings have comparable performance. 
In Figure \ref{fig:rho}, we again set $N=100$ and we plot the clustering error versus the correlation. 
In panel $(a)$ (resp. $(b)$) we set $\epsilon=0.1$ (resp. $\epsilon=0.2$) across all estimated latent positions in both cases. 
In panel $(a)$, the signal in distinguishing the two communities is weak and all embedding methods achieve similar performance across $\rho$. 
In panel $(b)$ however, the performance of clustering the estimated latent positions $\widehat{\bX}_{\textrm{corr}}^{1}$ and $\widehat{\bX}_{\textrm{corr}}^{2}$ decreases as the correlation increases, as the effective sample size of the correlated networks is decreasing.
%

In Figure \ref{fig:nn}, we plot the clustering error versus the number of vertices of the graphs. The rows of the figure correspond to different values of $\epsilon\in\{0.1,0.2\}$ and the columns correspond to different values of $\rho\in\{0.25,0.75\}$. Comparing panels $(a)$ and $(b)$ we observe how the effective sample size is diminished for the estimated latent positions $\widehat{\bX}_{\textrm{corr}}^{1}$, $\widehat{\bX}_{\textrm{corr}}^{2}$ as the correlation increases (i.e., more vertices are needed to achieve the same error in the highly correlated case).
For example, in panel $(a)$ ($\rho=0.25$), the clustering error for both $\widehat{\bX}_{\textrm{corr}}^{1}$ and $\widehat{\bX}_{\textrm{corr}}^{2}$ is approximately $0.2$ when the number of vertices is $N\approx250$, whereas in panel $(b)$ ($\rho=0.75$) the same error is achieved when $N\approx 350$.
This is mirrored in panels $(c)$ and $(d)$ as well.

%
%
\subsection{Vertex classification in brain networks}
\label{sec:vertex-classification}
Vertex classification is a problem that arises in applications where the goal is to predict vertex labels using a subset of the vertices for which this information is known a priori \cite{sussman12:_univer,chen_worm}. In brain networks, for instance, where the vertices correspond to brain regions or neurons, information about vertex attributes is sometimes known in more detail for a portion of the graph, and this can be used to infer these attributes in the remaining vertices  \cite{chen_worm}.  Here, we show how leveraging the information from a collection of networks  yields improvements in vertex classification accuracy, but the presence of edge correlation between the networks can reduce the effective sample size and result in smaller prediction improvements.

The HNU1 study \cite{Zuo2014} includes brain diffusion magnetic resonance images (dMRI)  from 30 healthy subjects that were scanned 10 times each over a period of one month. These scans were used to construct a collection of brain networks with the CPAC200 atlas \cite{Kiar2018}, resulting in a sample of 300 graphs (one per each subject and each scan) with 200 aligned vertices, and binary edges denoting the existence of nerve tracts between each pair of brain regions. The vertices of the networks are labeled according to the brain hemisphere, with 94 vertices corresponding to the left hemisphere, 98 to the right hemisphere, and 8 unlabeled vertices.
The post-processed brain networks were downloaded from \url{https://neurodata.io/mri/}.

To understand the effect of the correlation in subsequent inference, we start by measuring the correlation between pairs of graphs. Recall that under the $\rho$-correlated heterogeneous Erd\H os-R\'enyi (or Bernoulli) model, the correlation between two graphs $A^{(k)}\sim\operatorname{Bernoulli}(P^{(k)})$ and $A^{(\ell)}\sim\operatorname{Bernoulli}(P^{(\ell)})$
satisfies
$$\rho^{(k,\ell)} = \frac{\mathbb{E}[(A^{(k)}_{ij} - P^{(k)}_{ij})(A^{(\ell)}_{ij} - P^{(\ell)}_{ij})]}{\left[\operatorname{Var}(A^{(k)}_{ij})\operatorname{Var}(A^{(\ell)}_{ij})\right]^{1/2}}, \quad \forall i,j\in[n], i\neq j.$$
    Based on estimates of the edge probability matrices $\widehat{P}^{(k)}$ and $\widehat{P}^{(\ell)}$, we construct a plug-in estimator of the correlation given by
\begin{equation*}
    \hat{\rho}^{(k, \ell)} = \frac{2}{n(n-1)}\sum_{i<j} \frac{(A_{ij}^{(k)}-\widehat P^{(k)}_{ij})(A^{(\ell)}_{ij} - \widehat P^{(\ell)}_{ij}) }{\left[\widehat P^{(k)}_{ij} (1 - \widehat P^{(k)}_{ij})\widehat P^{(\ell)}_{ij} (1 - \widehat P^{(\ell)}_{ij})\right]^{1/2} }.
\end{equation*}
The edge probability matrix is estimated as $\widetilde{P}^{(k)} =\widehat{\mathbf{X}}^{(k)} \widehat{D}^{(k)} \widehat{\mathbf{X}}^{(k)^T} $, where $\widehat{D}^{(k)}\in\mathbb{R}^{d\times d}$ is a diagonal matrix such that $\widehat{D}^{(k)}_{uu}$ is the sign of the $u$-th eigenvalue of $A^{(k)}$ ordered by magnitude,  $\widehat{\mathbf{X}}^{(k)} = \operatorname{ASE}(A^{(k)}, d)$, and $d=15$ as suggested by an analysis to the same dataset in \cite{arroyo2019inference}. After that, $\widehat{P}^{(k)}$ is formed by trimming the entries to ensure that the values of the matrix are inside the interval $(0,1)$. Given a value $\epsilon\in(0,1/2)$, this estimator is defined as
\[\widehat{P}^{(k)}_{ij}= \left\lbrace \begin{array}{cl}
    \epsilon & \text{if }\widetilde P^{(k)}\leq\epsilon,\\
    \widetilde P^{(k)} &  \text{if }\widetilde P^{(k)}\in(\epsilon, 1-\epsilon),\\
     1-\epsilon &  \text{otherwise}.
\end{array}\right.\]
In our experiments, we choose $\epsilon=10^{-4}$, but we observe that the results remain quantitatively similar for a wide range of values. We also estimate the Pearson correlation between the graphs, which corresponds to the edge correlation under the correlated homogeneous Erd\H os-R\'enyi  model, and it is given by
\[\widehat\varrho^{(k,\ell)}=\frac{\sum_{i>j}(A_{ij}^{(k)} - \bar{A}^{(k)})(A_{ij}^{(\ell)} - \bar{A}^{(\ell)})}{[\sum_{i>j}(A_{ij}^{(k)} - \bar{A}^{(k)})^2\sum_{i>j}(A_{ij}^{(\ell)} - \bar{A}^{(\ell)})^2]^{1/2}} .\]
Here, $\bar{A}^{(k)}=\frac{1}{n(n-1)}\sum_{i>j}A^{(k)}_{ij}$ is the edge sample mean.

Figure~\ref{fig:hnu1-correlation} shows the empirical distribution of the estimated correlation values between pairs of networks, divided according to correlation estimates for pairs of networks corresponding to the same subject (blue curves) and pairs of networks belonging to different subjects (red curves).
For both edge and Pearson correlation distributions, the red curves appear to be stochastically smaller than the blue curves, suggesting that the correlation for pairs of networks within-subject is higher than the correlation between networks of  different subjects, which is expected. 
Thus, according to our theory, we would expect that the corresponding embeddings will show a stronger correlation, which can have consequences in subsequent inference.

\begin{figure}
    \centering
    \includegraphics[width=1\textwidth]{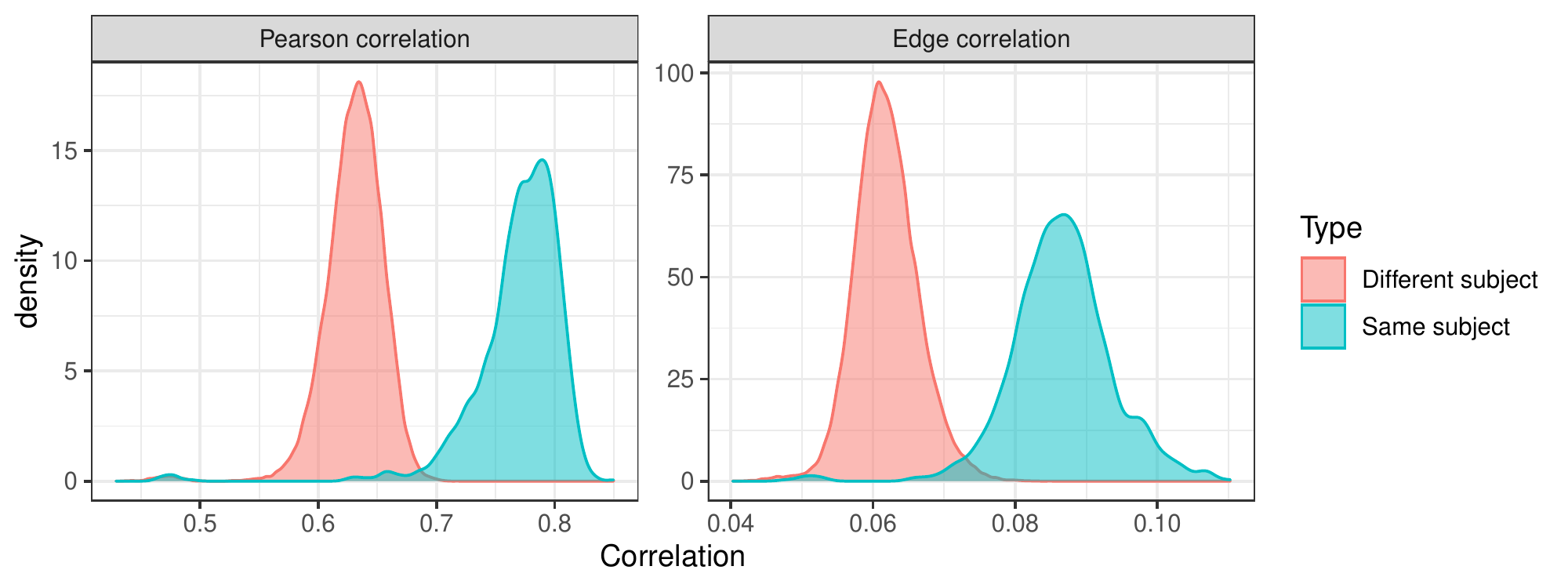}
    \caption{Empirical distribution of the correlation between pairs of graphs. The left panel shows the Pearson correlation between pairs of vectorized adjacency matrices, and the right panel shows the estimated edge correlation in the $\rho$-RDPG model. In both cases, the correlation between graphs corresponding to different subjects is larger than the same-subject correlation.}
    \label{fig:hnu1-correlation}
\end{figure}

 We now evaluate the performance of vertex classification when a sample of $m$ networks is observed.
 We focus on predicting the hemisphere labels using the information of a subset of the vertices to predict the label of the remaining ones.  The classification strategy is based on constructing an unsupervised classical omnibus embedding on the $m$ observed graphs. 
 Note that the theory states that the correlation across the embedding will be higher for the collection of graphs with the higher inherent correlation (i.e., same subject), and this would have the effect of reducing the effective sample size.
 After embedding, the average of the $m$ sets of estimated latent positions is calculated following a similar procedure to Section~\ref{sec:community-recovery-experiments}. The resulting embedding is used to perform 1-nearest neighbor classification for the unlabeled vertices.
 
We evaluate the classification performance of the omnibus embedding by randomly selecting a sample of 10\% of the vertices as the training set, for which the hemisphere labels are known, and the classification accuracy is measured on the remaining 90\% of the vertices with occluded hemisphere label.
The accuracy is computed as the average of 250  replications over randomly selected training vertices and subjects. We compare the performance of two different choices of the training set of networks: first, a random sample of $m$ networks from the same subject, and second, a random sample of $m$ networks from different subjects.
Figure~\ref{fig:hnu1-vertexclassif} shows the average classification accuracy (with bars corresponding to two standard errors) as a function of the number of networks used to estimate the latent positions. We observe that in general the accuracy improves as the sample size $m$ increases, which is expected. However, the figure also shows that the gains in accuracy are much better for the classifier that uses networks from different subjects. Our theory suggests that the presence of correlation inflates the estimation error of the corresponding latent positions of the RDPG model, and thus, the degraded performance of the method is both expected and in line with our theoretical analysis.

\begin{figure}
    \centering
    \includegraphics[width=0.8\textwidth]{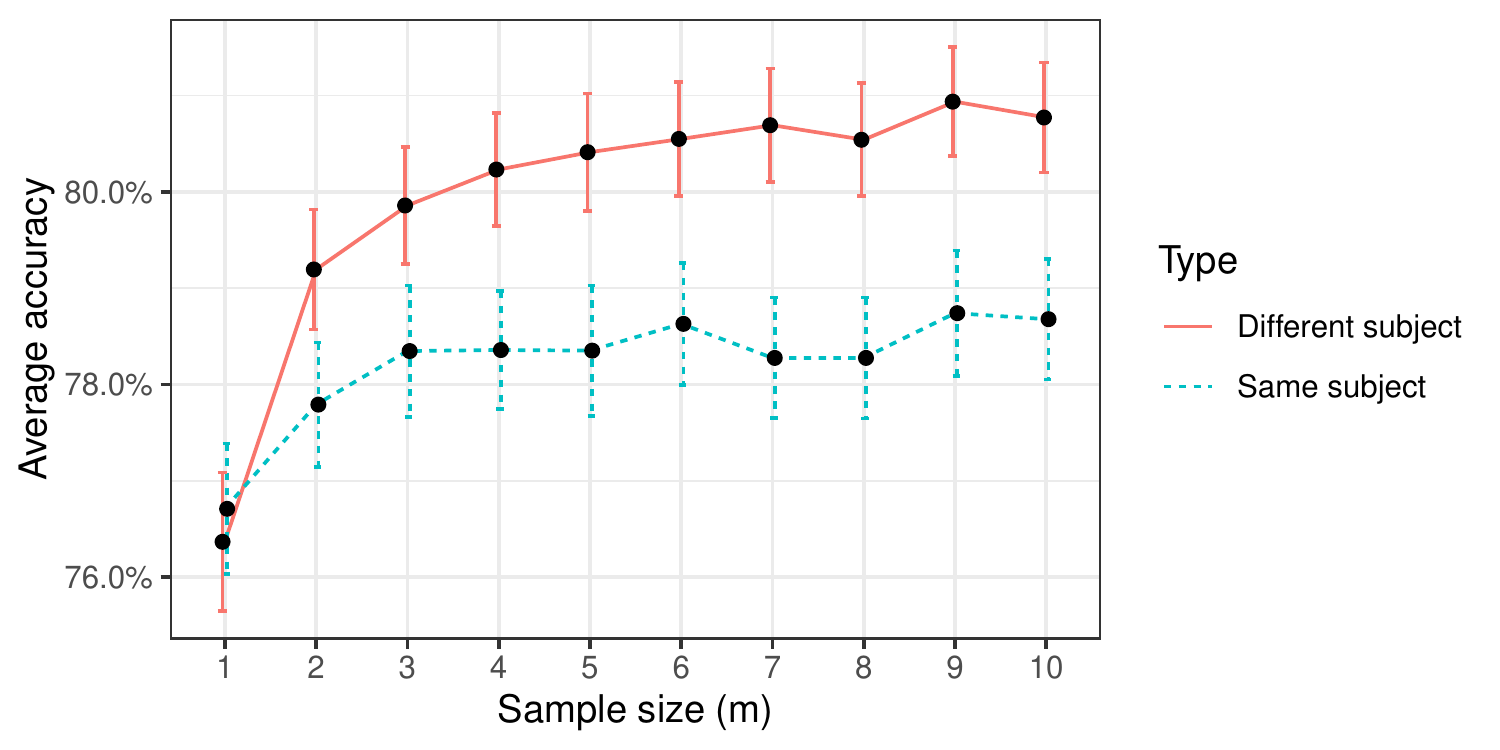}
    \caption{Average vertex classification accuracy for predicting brain hemisphere. For each run, 20 random vertices were used to train the model, and the hemisphere class was predicted on the rest based on a 1-nearest neighbor classifier using the estimated latent positions of the omnibus embedding of $m$ different graphs.  The embeddings constructed from graphs corresponding to different subjects show higher gains in accuracy as the sample size increases.}
    \label{fig:hnu1-vertexclassif}
\end{figure}

\section{Conclusion}
\label{sec:Conc}
The generalized omnibus embedding methodology is both flexible, allowing for an array of off-diagonal weightings, and tractable, in that its component pieces can be rigorously analyzed.
One such component piece, in particular, is the effect of induced and inherent correlations across the embedded networks in the omnibus framework, and in this paper we provide a first step toward understanding these correlations.
The implications of this are many, both in the analysis of real data (Section \ref{sec:aplysia}), and in the setting where we desire the omnibus method to 
imbue identically distributed
embedded data with a specific correlation structure or preserve the signal across inherently correlated networks.
Furthermore, the analysis of correlated graphs in the generalized OMNI framework, 
enables principled, high fidelity applications in time-series (and other naturally correlated) network settings.
Understanding the dual contributions of inherent and induced correlation also allows us to formulate and experiment with the network analogue of classical effective sample size analysis.

Of course, there are many significant follow-on  questions and research directions. For instance, given a feasible correlation structure for a collection of networks, can we choose weights in the generalized omnibus setting that would (from i.i.d.\@ or correlated identically distributed networks) reproduce this structure in the embedded space?
In the i.i.d.\@ setting, representing the desired correlation matrix via
$\boldsymbol{\rho}=[\rho(s_1,s_2)]\in\mathbb{R}^{m\times m}$ and the matrix of generalized OMNI weights $\boldsymbol{\alpha}=[\alpha(s_1,s_2)]\in\mathbb{R}^{m\times m}$, inducing $\boldsymbol{\rho}$ in the embedding space amounts to finding feasible $\boldsymbol{\alpha}$
solving
\begin{equation}
\label{eq:stuff}
    \boldsymbol{\rho}=1-\frac{1}{2m^2}{\bf D}_{\alpha},
\end{equation}
where 
$${\bf D}_{\alpha}=[{\bf D}_{\alpha}(k,\ell)]=[\|\alpha(k,\cdot)-\alpha(\ell,\cdot)\|_2^2].$$
If a solution exists, this might be easily solved via multidimensional scaling \cite{borg2005modern}, if not for the condition that each $\alpha(k,\cdot)$ must be nonnegative, sum to $m$, and be derived from symmetric $C^{(\ell)}$ matrices.
While this is further complicated in the presence of inherent correlation, we are nonetheless exploring possible approximation approaches for this problem, including alternately projecting onto the cone of distance matrices \cite{zhang2014distance} and the polytope defining the constraints on the $\alpha$'s.
An approximate solution for this is essential, since it would allow us to use the generalized omnibus framework to produce a given correlation structure in the embedded space.

In the absence of a general solution, there are a number of $\fM$ we could consider that would produce differing limiting levels of correlation amongst the embedded network pairs.
For example, we could consider the forward omnibus matrix $\fM_{for}$, where 
$$\fM_{for}^{(i,j)}=\begin{cases}
\frac{ (i-1)A^{(j)}+ A^{(i)}}{i}&\text{ if }i>j\\
A^{(i)}&\text{ if }i=j\\
\frac{ (j-1)A^{(i)}+ A^{(j)}}{j}&\text{ if }i<j
\end{cases}$$ 
In contrast to the dampened OMNI setting, which has decaying correlation across graph pairs, in the forward model, the correlation is increasing as the graph indices are increasing, and
can be used to model time series where the correlation is growing as the series progresses.
While our analysis of the limiting induced correlation structure given in Eq.\@ (\ref{eq:induced_corr}) permits for these specific, example-based constructions, a more automated approach is needed for broader applications.

Next, the structure of the generalized omnibus mean matrix $\TP=J_m\otimes P$ is designed to maintain the low-rank property of $P$ in $\TP$ and allows for the leading $d$ eigenvectors of $P$ to be related to those of $\TP$ (indeed $\vec{1}\otimes U_P$ provides a basis for the leading $d$-dimensional eigenspace of $\TP$).
If we consider other structures in the Kronecker product that have the form $E\otimes P$, with the restriction that $\vec{1}$ is a leading eigenvector of $E$, can we replicate the omnibus analysis?  
We are exploring this at present with examples of interest, including $E$ has ring-graph structure.
While more complex $E$ will necessarily violate the low-rank assumption on $\TP$, we are exploring if there is still sufficient concentration (about $\vec{1}\otimes \bX$) of the leading $d$-dimensional eigenspace of $\fM$ in order to  of a low-rank approximation of $\TP$.
We are also exploring extensions of the out-of-sample framework of \cite{levin2018out} to the generalized OMNI framework.    
By using the OMNI framework to embed the core vertices aligned across networks and out-of-sample embedding any remaining vertices, this would allow us to jointly embed networks in the cases when the vertices across the networks are only partially (or errorfully) aligned.
This would, in turn, enable us to use the OMNI framework towards alignment-based tasks such as seeded graph matching (as in \cite{patsolic2014seeded}).

Finally, given recent extensions for matrix concentrations in sparser regimes \cite{lei2020consistency,bandeira2016sharp}, we suspect that our consistency results can be extended to sparser graphs.  We are also currently exploring distributional results and asymptotics for omnibus embeddings in which the latent positions differ 
and in which the OMNI matrix is asymmetric (while still maintaining $\mathbb{E}(\widetilde P)=J_m\otimes P$).
These extensions can be very challenging, so preliminary results will likely be constrained to only a few special cases.  The range of possibilities for follow-on work highlights the important role of joint embeddings in graph inference; as such, the novel intricacies of embedded-space correlation that we examine here are a key component of multiple network inference.
\vspace{3mm}

\noindent{\bf Acknowledgements}
This material is based on research sponsored by the Air Force Research Laboratory and Defense Advanced Research Projects Agency (DARPA) under agreement number FA8750-20-2-1001.
This work is also supported in part by the D3M program of DARPA.
The U.S. Government is authorized
to reproduce and distribute reprints for Governmental purposes notwithstanding any
copyright notation thereon. The views and conclusions contained herein are those of
the authors and should not be interpreted as necessarily representing the official policies
or endorsements, either expressed or implied, of the Air Force Research Laboratory and
DARPA or the U.S. Government. The authors also gratefully acknowledge the support
of NIH grant BRAIN U01-NS108637.
We also gratefully acknowledge illustrative conversations with Profs. Minh Tang, Keith Levin, Daniel Sussman and Carey Priebe that helped shaped this work.


\bibliographystyle{plain}
\bibliography{biblio_summary}
\newpage
\appendix
\section{Proofs of main results}\label{app:proofs_main_results}
Herein we collect the proofs of the main theoretical results of the paper.
Before proving the theorems, we will first establish some convenient asymptotic notation for our theory moving forward.
\begin{definition}
	\label{def:whp}
	Given a sequence of events $\{ E_n \} \in \mathcal{F}$, where $n=1, 2, \cdots$, we say that $E_n$ occurs 
	\begin{itemize}
	    \item[i.] {\em asymptotically almost surely} and write $E_n \text{ a.a.s. }$ if $\p(E_n) \rightarrow 1$ as $n \rightarrow \infty$. 
	    \item[ii.] {\em with high probability},
	and write $E_n \text{ w.h.p. }$,
	if for some $a_0 \geq 2$, there exists finite positive constant $A_0$ depending only on $a_0$ such that $$\p[ E_n^c ] \le A_0n^{-a_0}$$ for all $n$.
	\end{itemize}
	We note that $E_n$ occurring w.h.p. is stronger than $E_n$ occurring a.a.s., as w.h.p.\@ implies, by the Borel-Cantelli Lemma \cite{chung1974course},
	that $\p(\limsup E_n)=1$ and with probability $1$ all but finitely many
	$E_n$ occur.
\end{definition}

\noindent For matrices $A\in\mathbb{R}^{p_1\times p_2}$ and $B\in\mathbb{R}^{p_2\times p_3}$, we will make use of the following common matrix norm identities in the proofs below (where $A_i$ denotes the $i$-th row of $A$, and $\|A\|=\|A\|_2$ denotes the spectral norm of $A$; the symbol ``$:=$'' is used below to denote a definition).
\begin{align*}
\|A\|_p&:=\sup_{x\neq 0}\frac{\|Ax\|_{p} }{\|x\|_p }\text{ for }1\leq p\leq \infty;\\
\|A\|_{2\rightarrow\infty}&:=\sup_{x\neq 0}\frac{\|Ax\|_{\infty} }{\|x\|_2 }; 
\end{align*}

\noindent
The following lemma concerns relationships between the $2 \rightarrow \infty$ bound and other classical matrix norms. For a proof, see \cite{cape2toinfty}.
\begin{lemma}
\label{lemma:2toinf}
For $A\in\R^{p_1\times p_2}$ and $ B\in\R^{p_2\times p_3}$, we have the following:
\begin{align*}
\|A\|_{2\rightarrow \infty}&=\max_{i\in[p_1]}\|A_i\|_2; \\
\|A\|_{2\rightarrow\infty}&\leq \|A\|\leq \sqrt{p_1}\|A\|_{2\rightarrow\infty}; \\ 
\|AB\|_{2\rightarrow \infty}&\leq \|A\|_{2\rightarrow \infty}\|B\|; \\
\|A\|&\leq \sqrt{\|A\|_1\|A\|_\infty }
\end{align*}
where $A_i$ denotes the $i$-th row of $A$.
\end{lemma}
\subsection{Correlation in the classical omnibus embedding and its relationship to correlated stochastic block model graphs}
\label{sec:latent}

We begin with the proof of Theorem \ref{thm:rhoCLT} from Section \ref{sec:corr}.
\begin{proof} (Induced correlation in classical OMNI and its relationship with correlated block models)

To ease notation, we will frequently suppress the explicit dependence on $n$ in the subscript of $\TP=\TP_n$, $B^{(k)}_n$, etc.,  noting that this dependence is to be implicitly understood throughout.
For each $n\geq 1$, we write the spectral decomposition of the positive semidefinite $P_n$ via 
$$P_n=\bX_n\bX_n^T=U_{P_n} S_{P_n} U_{P_n}^T.$$

Let the singular value decomposition of $U_{P}^TU_{B^{(1)}}$ be denoted $W_{11} D_1 W_{12}^T$ and define $W^{*,1}=W_{11}W_{12}^T$.
Similarly, let the singular value decomposition of $U_{P}^TU_{B^{(2)}}$ be denoted $W_{21} D_2 W_{22}^T$ and define $W^{*,2}=W_{21}W_{22}^T$.
As in the proof of Theorem 9 (p 76) in \cite{athreya_survey}, we have that for each $k=1,2$ and fixed index $i$, 
	\begin{equation} \label{eq:key} 
	\begin{aligned}
	\sqrt{n}\left(U_{B^{(k)}} S_{B^{(k)}}^{1/2} - U_P S_P^{1/2} W^{*,k}\right)_i
	=& \sqrt{n}\left((B^{(k)}-P)U_P S_P^{-1/2} W^{*,k}\right)_i+O(n^{-1/2}\log n)
	\end{aligned} \end{equation}
	with high probability.
	Letting $W_n$ be a sequence of orthogonal matrices such that 
	$$U_{P_n}S_{P_n}^{1/2}W_n=\bX_n$$ for all $n\geq 1$, we have then that, again with high probability, 
		\begin{equation} \label{eq:key11} 
	\begin{aligned}
	\sqrt{n}\left(U_{B^{(k)}} S_{B^{(k)}}^{1/2}(W^{*,k})^TW - U_P S_P^{1/2}W \right)_i
	=& \sqrt{n}\left((B^{(k)}-P)U_P S_P^{-1/2} W\right)_i+O(n^{-1/2}\log n)
	\end{aligned} \end{equation}
It follows then that for fixed $i$, we have with high probability
\begin{equation} \label{eq:key2} 
	\begin{aligned}
	&\sqrt{n}\left(U_{B^{(1)}} S_{B^{(1)}}^{1/2}(W^{*,1})^TW - U_{B^{(2)}} S_{B^{(2)}}^{1/2}(W^{*,2})^TW\right)_i
	\\
	&= 
	\sqrt{n}\left(U_{B^{(1)}} S_{B^{(1)}}^{1/2}(W^{*,1})^TW - U_P S_P^{1/2}W \right)_i+
	\sqrt{n}\left(U_P S_P^{1/2} W-
	U_{B^{(2)}} S_{B^{(2)}}^{1/2}(W^{*,2})^TW \right)_i\\
	&=\sqrt{n}\left((B^{(1)}-P)U_P S_P^{-1/2} W\right)_i-\sqrt{n}\left((B^{(2)}-P)U_P S_P^{-1/2} W\right)_i+O\left(\frac{\log n}{\sqrt{n}}\right)\\
	&=\sqrt{n}\left((B^{(1)}-B^{(2)})U_P S_P^{-1/2} W\right)_i+O\left(\frac{\log n}{\sqrt{n}}\right).
	\end{aligned} \end{equation}
Next, note that 
\begin{align*}
\sqrt{n}\left((B^{(1)}-B^{(2)})U_P S_P^{-1/2} W\right)_i&=\sqrt{n}\left((B^{(1)}-B^{(2)})\bX \right)_iW^TS_P^{-1}W\\
&=\sqrt{n}\left(\sum_j(B^{(1)}_{i,j}-B^{(2)}_{i,j})X_j\right)W^TS_P^{-1}W\\
&=\left(n^{-1/2}\sum_j(B^{(1)}_{i,j}-B^{(2)}_{i,j})X_j\right)(nW^TS_P^{-1}W).
\end{align*}
Conditioning on $X_i=x_i$,
\begin{align*}
n^{-1/2}\sum_j(B^{(1)}_{i,j}-B^{(2)}_{i,j})X_j
\end{align*}
is a scaled sum of $n-1$ i.i.d random variables (the $(B^{(1)}_{i,j}-B^{(2)}_{i,j})X_j$), each with mean
\begin{align*}
    \Ex(&(B^{(1)}_{i,j}-B^{(2)}_{i,j})X_j)=\Ex(\Ex((B^{(1)}_{i,j}-B^{(2)}_{i,j})X_j|X_j))\\
    &=\Ex(X_j \Ex(B^{(1)}_{i,j}-B^{(2)}_{i,j}|X_j))=0
\end{align*} 
and covariance matrix $$\widetilde{\Sigma}(x_i)=2(1-\rho)\EX[X_j X_j^T(x_i^T X_j-(x_i^T X_j)^2)],$$ 
as
\begin{align*}
\text{Cov}((B^{(1)}_{i,j}-B^{(2)}_{i,j})X_j)  &= \EX[(B^{(1)}_{i,j}-B^{(2)}_{i,j})^2X_jX_j^T] \\ 
&=\Ex(\EX[(B^{(1)}_{i,j}-B^{(2)}_{i,j})^2X_jX_j^T|X_j]) \\
&=\Ex(X_jX_j^T\EX[(B^{(1)}_{i,j}-B^{(2)}_{i,j})^2|X_j] )\\
&=\Ex(X_jX_j^T\EX[((B^{(1)}_{i,j})^2+(B^{(2)}_{i,j})^2-2B^{(1)}_{i,j}B^{(2)}_{i,j})|X_j] )\\
&=\Ex(X_jX_j^T
[2x_i^TX_j-2x_i^TX_j(x_i^TX_j+\rho(1-x_i^TX_j))] 
)\\
&=\Ex(X_jX_j^T
[2x_i^TX_j-2(x_i^TX_j)^2-\rho2(1-x_i^TX_j)x_i^TX_j])\\ 
&=2(1-\rho)\EX[X_j X_j^T(x_i^T X_j-(x_i^T X_j)^2)].
\end{align*}
The classical multivariate central limit theorem then yields
\begin{equation}
    \label{eq:normal}
n^{-1/2}\sum_j(B^{(1)}_{i,j}-B^{(2)}_{i,j})X_j \xrightarrow{D} \mathcal{N}(0,\widetilde{\Sigma}(x_i)).
\end{equation}
The strong law of large numbers (SLLN) ensures that $\frac{1}{n}\bX^T\bX=\frac{1}{n}W^TS_PW\xrightarrow{a.s} \Delta$ and therefore, 
$nW^TS_P^{-1}W\xrightarrow{a.s} \Delta^{-1}$.
Therefore, by the multivariate Slutsky's Theorem, conditional on $X_i=x_i$, we have that 
\begin{align}
\label{eq:conv}
\sqrt{n}\left((B^{(1)}-B^{(2)})U_P S_P^{-1/2} W\right)_i \xrightarrow{D} \mathcal{N}(0,\widetilde\Sigma(x_i,\rho)).
\end{align}
By setting $W^{(1)}=(W^{*,1})^T W$ and $ W^{(2)}=(W^{*,2})^TW$, the requisite result then follows from multivariate Slutsky's applied to Eq. (\ref{eq:key2}) and (\ref{eq:conv}),
and integrating the above display over the possible values of $x_i$ with respect to distribution $F$.
\end{proof}

\subsection{Consistency of generalized omnibus embeddings}
\label{sec:genomniconsis}
Here, we prove our consistency result, Theorem \ref{theorem:genOMNI_consistency} from Section \ref{sec:genomni}, which is, in part, what makes the generalized omnibus embedding useful for estimation.
To prove Theorem \ref{theorem:genOMNI_consistency}, we will need a Bernstein concentration inequality bound given immediately below, in \ref{ssec:bern}; and number of intermediate supporting lemmas (Lemmas \ref{lemma:eig_order}--\ref{lem:Q1}) all of which are suitably adapted from \cite{levin_omni_2017}.

Again, to ease notation, we will on occasion suppress the explicit dependence on $n$ in the subscript of $\bX=\bX_n$, $\fM=\fM_n$, $\TP=\TP_n$, etc., noting that this dependence is implicitly understood throughout.
Recall that the spectral decomposition of $\TP$ is given by
$$\widetilde{P}=\Ex(\fM)=U_{\TP}S_{\TP}U_{\TP}^T,$$ 
where $U_{\TP}\in\mathbb{R}^{mn\times d}$ and $
S_{\TP}\in\mathbb{R}^{d\times d}$.
Also recall that the adjacency spectral embedding of $\fM$ is given by ASE$(\fM,d)=U_{\fM}S_{\fM}^{1/2}$.
\subsubsection{Concentration inequality via matrix Bernstein}
\label{ssec:bern}
\begin{lemma}
\label{lemma:conc}
Let $F$ be a distribution on a set $\mathcal{X}\in \R^d $ satisfying $\langle x,x' \rangle \in [0,1]$ for all $x,x'\in \mathcal{X}$. 
Let $X_1,X_2,\cdots$ $,X_n,Y\stackrel{i.i.d}{\sim}F$, and let $P=\bX\bX^T$ where $\bX=[X_1^T|X_2^T|\cdots|X_n^T]^T$. Let $$(A_n^{(1)},A_n^{(2)},\cdots,A_n^{(m)},\bX_n)\sim \mathrm{JRDPG}(F,n,m,R)$$ be a sequence of correlated $\mathrm{RDPG}$ random graphs, and let $\fM_n$ denote the generalized omnibus matrix as in Definition \ref{def:genOMNI}. Then for $n$ sufficiently large, with high probability, 
$$\|\mathfrak{M}-\EX \mathfrak{M}\|\leq 4m\sqrt{(n-1)\log mn}.
$$
\end{lemma}
\begin{proof}
Condition on $\bX$ and let $P=\bX\bX^T$, so that 
\[\EX \mathfrak{M}=\widetilde{P}=\begin{bmatrix}P&P&\cdots&&P\\P&P&\cdots&&P\\\vdots&\vdots&\ddots&&\vdots\\P&P&\cdots&&P\end{bmatrix}.\]
For all $l\in[m]$ and $i,j\in[n],$ we define an auxiliary block matrix $E^{(l)}_{i,j}\in\R^{mn\times mn}$ which will help us express the difference $\mathfrak{M}-\EX \mathfrak{M}$ as a sum of independent Hermitian matrices, which will allow us to apply Bernstein's matrix bound (see, for example, Theorem 5.4.1 in \cite{vershynin2018high}).
For $i\neq j$, let 
$e^{ij}=e_ie_j^T+e_je_i^T$ 
where $e_i$ is a vector in $\R^n$ with all its entries equal to 0 except the i-th entry which is equal to 1,
so that $e^{ij}$ is a matrix whose entries are all equal to 0 except the $(i,j)$th and $(j,i)$th entries which are equal to 1.
For $\ell\in[m]$, we define the $m\times m $  matrix $C^{(\ell)}$ (symmetric because $\mathfrak{M}$ is symmetric)  as follows, 
 \[ C^{(\ell)}=\begin{bmatrix}c_\ell^{(1,1)}&c_\ell^{(1,2)}&\cdots&&c_\ell^{(1,m)}\\
 c_\ell^{(2,1)}&c_\ell^{(2,2)}&\cdots&&c_\ell^{(2,m)}\\\vdots&\vdots&\ddots&&\vdots\\
 c_\ell^{(m,1)}&c_\ell^{(m,2)}&\cdots&&c_\ell^{(m,m)}\end{bmatrix}.\]
We then define the $mn\times mn$ block (symmetric) matrix $E^{(\ell)}_{i,j}$ as the Kronecker product of $C^{(\ell)}$ and $e^{ij}$, i.e., $E^{(\ell)}_{i,j}=C^{(\ell)}\otimes e^{ij}$. Now, we can write $\mathfrak{M}-\EX \mathfrak{M}$ as 
\begin{align*} 
    \mathfrak{M}-\EX \mathfrak{M}= \sum_{1\leq \ell\leq m}\sum_{1\leq i<j\leq n}(A^{(\ell)}_{ij}-P_{ij})E^{(\ell)}_{i,j} = \sum_{1\leq i<j\leq n}\Big(\underbrace{\sum_{1\leq \ell\leq m}(A^{(\ell)}_{ij}-P_{ij})E^{(\ell)}_{i,j}}_{:=B_{i,j,m}}\Big).
\end{align*}
These $B_{i,j,m}$ are then symmetric, mean-zero, independent matrices, as we have grouped edges indexed by the same $(i,j)$ pair across graphs to account for the across graph correlation and achieve the desired independence.

As we have written $\mathfrak{M}-\EX \mathfrak{M}$ as the sum of $\binom{n}{2}$ symmetric, mean-zero, independent matrices
$B_{i,j,m}$ where for all $i,j\in[n]$, we have that
$$\|B_{i,j,m}\|\leq m\max\limits_{\ell\in[m]}\|C^{(\ell)}\|\leq m\Big(\max\limits_{\ell\in[m]}\max\limits_{q\in[m]}\sum\limits_{k=1}^{m}c_\ell^{(q,k)}\Big).$$We can now apply the matrix Bernstein inequality to derive the desired concentration of $\mathfrak{M}-\EX \mathfrak{M}$.
To this end, let 
\begin{align*}
L&:=\max\limits_{\ell\in[m]}\max\limits_{q\in[m]}\sum\limits_{k=1}^{m}c_\ell^{(q,k)}\leq m;
\end{align*}
\noindent 
To apply the matrix Bernstein inequality, it remains to compute the variance term
\begin{align}
\label{eq:varbndMB}
\nu(\mathfrak{M}-\EX \mathfrak{M})&=\left\|
\sum_{1\leq i<j\leq n}\Ex[B_{i,j,m}^2]\right\|\notag\\
&\leq\left\|\sum_{1\leq i<j\leq n}\sum_{1\leq \ell\leq m}\Ex\left[(A^{(\ell)}_{ij}-P_{ij})^2(E^{(\ell)}_{i,j})^2\right]\right\| \notag\\ &\hspace{1cm}  +2\left\|\sum_{1\leq i<j\leq n}\sum_{ \ell_1<\ell_2}\Ex\left[(A^{(\ell_1)}_{ij}-P_{ij})(A^{(\ell_2)}_{ij}-P_{ij})E^{(\ell_1)}_{i,j}E^{(\ell_2)}_{i,j}\right]\right\|.
\end{align}

\noindent
To bound the variance term $\nu(\mathfrak{M}-\EX \mathfrak{M})$ we will bound the two terms in Eq.\@ (\ref{eq:varbndMB}) independently. Let $D_{ij}=e^{ij}e^{ij}=e_ie_i^T+e_je_j^T\in\R^{n\times n}$. 
The mixed-product property of Kronecker products implies that 
$$
E^{(\ell_1)}_{i,j}E^{(\ell_2)}_{i,j}=(C^{(\ell_1)}\otimes e^{ij})(C^{(\ell_2)}\otimes e^{ij})=C^{(\ell_1)}C^{(\ell_2)}\otimes D_{ij}.
$$
Note that for any symmetric matrix $A$, we have that $\|A\|\leq\sqrt{\|A\|_1\|A\|_{\infty}}=\|A\|_\infty$
(see, for example, \cite{horn85:_matrix_analy} Ex. 5.6.21), 
where $\|\cdot\|_1$ is the maximum column sum matrix norm and $\|\cdot\|_\infty$ the maximum row sum matrix norm.
The first term is then bounded above via
\begin{align*}
\left\|\sum_{1\leq i<j\leq n}\sum_{1\leq \ell\leq m}\Ex\left[(A^{(\ell)}_{ij}-P_{ij})^2(E^{(\ell)}_{i,j})^2\right]\right\| \notag&= \Big\|\underbrace{\sum_{1\leq \ell\leq m}\sum_{1\leq i<j\leq n}P_{ij}(1-P_{ij})[(C^{(\ell)})^2\otimes D_{ij}]}_{:=\Delta_1}\Big\|\\&\leq\|\Delta_1\|_{\infty}\\
&\leq \frac{1}{4}(n-1)\max_{q\in[m]} \sum_{\ell=1}^{m}\sum_{r=1}^{m}\sum_{k=1}^{m}c_\ell^{(q,k)}c_\ell^{(k,r)}\\
&=\frac{1}{4}(n-1)\max_{q\in[m]} \sum_{\ell=1}^{m}\sum_{k=1}^{m}c_\ell^{(q,k)}\alpha(k,\ell)\\
&\leq \frac{1}{4}(n-1)mL.
\end{align*}
The second term is then bounded above via
\begin{align*}
    &\left\|\sum_{1\leq i<j\leq n}\sum_{ \ell_1<\ell_2}\Ex\left[(A^{(\ell_1)}_{ij}-P_{ij})(A^{(\ell_2)}_{ij}-P_{ij})E^{(\ell_1)}_{i,j}E^{(\ell_2)}_{i,j}\right]\right\|
    \\&\hspace{5mm}= \Big\|\underbrace{\sum_{1\leq i<j\leq n}P_{ij}(1-P_{ij})\sum_{\ell_1< \ell_2}\rho_{\ell_1,\ell_2}C^{(\ell_1)}C^{(\ell_2)}\otimes D_{ij}}_{:=\Delta_2}\Big\|\\
    &\hspace{5mm}\leq\sqrt{\|\Delta_2\|_{\infty}\|\Delta_2\|_1}\\
    &\hspace{5mm}\leq \frac{1}{4}(n-1)
    \sqrt{
    \left(\max_{q\in[m]} \sum_{\ell_1<\ell_2}\sum_{r=1}^m\sum_{k=1}^{m}c_{\ell_1}^{(q,k)}c_{\ell_2}^{(k,r)}\right)\left(\max_{r\in[m]} \sum_{\ell_1<\ell_2}\sum_{q=1}^m\sum_{k=1}^{m}c_{\ell_1}^{(q,k)}c_{\ell_2}^{(k,r)} \right)}.
\end{align*}
Note that 
\begin{align*}
\max_{q\in[m]} \sum_{\ell_1<\ell_2}\sum_{r=1}^m\sum_{k=1}^{m}c_{\ell_1}^{(q,k)}c_{\ell_2}^{(k,r)}&= \max_{q\in[m]} \sum_{\ell_1<\ell_2}\sum_{k=1}^{m}c_{\ell_1}^{(q,k)}\alpha(k,\ell_2)\\
&\leq m\cdot \max_{q\in[m]} \sum_{\ell_1=1}^m\sum_{k=1}^{m}c_{\ell_1}^{(q,k)}\\
&= m\cdot \max_{q\in[m]} \sum_{\ell_1=1}^m\alpha(q,\ell_1)\\
&= m^2,
\end{align*}
and that a similar bound holds for $\max_{r\in[m]} \sum_{\ell_1<\ell_2}\sum_{q=1}^m\sum_{k=1}^{m}c_{\ell_1}^{(q,k)}c_{\ell_2}^{(k,r)}$.
Therefore, we have that the second term in the variance is bounded by $(n-1)m^2/4$ yielding a total variance bound of 
$$\nu(\mathfrak{M}-\EX \mathfrak{M})\leq (n-1)m^2.$$
To apply the matrix Bernstein bound, let $t=4m\sqrt{(n-1)\log mn}$, and then
\begin{align*}
    \PX\left[\|\mathfrak{M}-\EX \mathfrak{M}\|\geq 4m\sqrt{(n-1)\log mn}\right]&\leq 2nm\cdot \text{exp}\left\{\frac{-8m^2(n-1)\log mn}{\nu(\mathfrak{M}-\EX \mathfrak{M})+\frac{4m^2L\sqrt{(n-1)\log mn}}{3}}
    \right\}\\
    &\leq 
    2nm\cdot \text{exp}\left\{\frac{-8m^2(n-1)\log mn}{(n-1)m^2+\frac{4m^3\sqrt{(n-1)\log mn}}{3}}
    \right\}\\
    &\leq 
    2nm\cdot \text{exp}\left\{\frac{-8m^2(n-1)\log mn}{2(n-1)m^2}
    \right\}\\
    &= 2m^{-3}n^{-3},
\end{align*}
where the third inequality holds for $n$ sufficiently large ($\sqrt{n}\geq \frac{4}{3}m\log mn$) as we assume $m$ is not growing in $n$.
Integrating over the $\bX$ then yields the desired result.
\end{proof}
\subsubsection{Supporting Lemmas}
\begin{lemma}[Observation 2 in \cite{levin_omni_2017}]
\label{lemma:eig_order}
Let $F$ be a distribution on a set $\mathcal{X}\in \R^d $ satisfying $\langle x,x' \rangle \in [0,1]$ for all $x,x'\in \mathcal{X}$. 
Let $X_1,X_2,\cdots$ $,X_n,Y\stackrel{i.i.d}{\sim}F$, and let $P=\bX\bX^T$ where $\bX=[X_1^T|X_2^T|\cdots|X_n^T]^T$. 
Then with probability at least $1-\frac{d^2}{n^2}$, we have that
$$|\lambda_i(P)-n\lambda_i \Ex(YY^T)|\leq 2d\sqrt{n\log n}.$$
Furthermore, with high probability there exists a constant $C>0$ such that for all $i\in[d]$, $\lambda_i(P)\geq Cn\delta$ and 
$\lambda_i(\TP)\geq Cnm\delta$, where $\delta=\lambda_d(\Ex(YY^T))$.
\end{lemma}
\noindent Assuming $\delta>0$, combining Lemma \ref{lemma:eig_order} with the fact that the rows of $U_{\TP}S_{\TP}^{1/2}$ are bounded in Euclidean norm by $1$,  we obtain 
\begin{equation}
\label{eq:U_P}
    \|U_{\TP}\|_{2\rightarrow\infty}\leq C(mn)^{-1/2} \hspace{.1cm}\text{ w.h.p.}
\end{equation}
\begin{lemma}
\label{lemma:evector_subspace}
With notation as in Lemma \ref{lemma:eig_order}, assume that $\delta>0$ so that $\Ex(YY^T)$ is full rank.
Let the singular value decomposition of $U_{\TP}^TU_{\mathfrak{M}}\in\mathbb{R}^{d\times d}$ be given by $V_1 \Sigma V_2^T$.
Then, there exists a constant $C>0$ such that
\begin{align*}
    \|U_{\TP}^TU_{\mathfrak{M}}-V_1V_2^T\|_F\leq C\frac{\log mn}{n} \hspace{0.5cm} \text{w.h.p.}
\end{align*}
\end{lemma}
\begin{proof}[Adapted from the proof of Proposition 16 in \cite{lyzinski15_HSBM}]
Working on the intersection of the sets where both Lemma \ref{lemma:conc} and \ref{lemma:eig_order} hold (noting this set has high probability), note that Weyl's theorem gives that $\lambda_d(\fM)\geq Cnm$ for some constant $C$.
Let $\sigma_1\geq \sigma_2\geq\cdots\geq\sigma_d$ be the singular values of $U_{\TP}^TU_{\fM}$, so that $\sigma_i=cos(\theta_i)$ where the $\theta_i$'s are the principal angles between the subspaces spanned by $U_{\TP}$ and $U_{\fM}$.
The Davis-Kahan theorem (see, for example, Theorem 3.6 in \cite{Bhatia1997} Theorem VII.3.1) then implies that with high probability (where $C$ is a constant that can change line-to-line)
\begin{align*}
    \|U_{\TP}^TU_{\mathfrak{M}}-V_1V_2^T\|_F&=\sqrt{\sum_i(1-\sigma_i)^2}\\
    &\leq \sum_i(1-\sigma_i^2)\\
    &\leq d\max_i |\sin(\theta_i)|^2\\
    &\leq C\frac{d^2\|\fM-\TP\|^2}{\lambda_d(\fM)^2}\\
    &\leq C\frac{m^2(n-1)\log mn}{n^2m^2}\\
    &\leq C\frac{\log mn}{n},
\end{align*}
where the bounds in the fourth line follows from Davis-Kahan, and those in the second-to-last line follow from Lemma \ref{lemma:conc} (numerator) and Lemma \ref{lemma:eig_order} (denominator).
\end{proof}
\vspace{0.2cm}
\normalfont
\noindent

\begin{lemma}
\label{lemma:V_bounds}
With the assumptions and notation of Lemma \ref{lemma:evector_subspace}, let $V=V_1V_2^T$. 
Then we have that
\begin{align}
\label{eq:31}
    \|VS_{\fM}-S_{\TP}V\|_F\leq Cm^{2}\log mn \hspace{0.5cm} w.h.p.,
\end{align}
\begin{align}
\label{eq:32}
    \|VS_{\fM}^{1/2}-S_{\TP}^{1/2}V\|_F\leq C\frac{m^{3/2}\log mn}{n^{1/2}} \hspace{0.5cm} w.h.p.,
\end{align}
\begin{align}
\label{eq:33}
    \|VS_{\fM}^{-1/2}-S_{\TP}^{-1/2}V\|_F\leq C\frac{m^{1/2}\log mn}{n^{3/2}} \hspace{0.5cm} \text{w.h.p.},
\end{align}
where $S_{\fM}^{-1},$ and $S_{\TP}^{-1}$ are understood to be the pseudoinverses of $S_{\fM},$ and $S_{\TP}$ in the unlikely event these matrices are singular.
\end{lemma}
\begin{proof}[Adapted from the proof of Proposition 17 in \cite{lyzinski15_HSBM}]
Let 
$$R:=\um-\up\up^T\um,
$$
so that 
\begin{align*}
VS_{\fM}&=(V-U_{\TP}^TU_{\fM})\sm+\up^T\um\sm\\
&=(V-U_{\TP}^TU_{\fM})\sm+\up^T\fM\um I_{\pm}\\
&=(V-U_{\TP}^TU_{\fM})\sm+\up^T(\fM-\TP)\um I_{\pm}+\up^T\TP\um I_{\pm}\\
&= (V - \up^T\um)\sm+ \up^T(\fM-\TP)R I_{\pm}+\up^T(\fM-\TP)\up\up^T\um I_{\pm}+
\up^T\TP\um I_{\pm}\\
&=(V - \up^T\um)\sm+ \up^T(\fM-\TP)RI_{\pm}+\up^T(\fM-\TP)\up\up^T\um I_{\pm}+
\spp\up^T\um I_{\pm},
\end{align*}
where $I_{\pm}$ is a random diagonal matrix with $\pm1$ entries on its diagonal indicating whether the signs of the eigenvalues associated with $\um$ agree for $\fM$ and for $|\fM|$.
We will first bound the Frobenius norm of $\up^T(\fM-\TP)\up$.
Specifically, we will prove that there exists a constant $C$ such that w.h.p.
\begin{equation}
\label{eq:u(M-P)u}
    \|U_{\TP}^T({\mathfrak{M}}-\TP)U_{\TP}\|_F\leq  d\|U_{\TP}^T({\mathfrak{M}}-\TP)U_{\TP}\|_{\text{max}}\leq Cdm^2\sqrt{\log mn}.
\end{equation}
The proof proceeds as follows.
Note that $U_{\TP}^T({\mathfrak{M}}-\TP)U_{\TP}\in\R^{d\times d}$ and $\|U_{\TP}^T({\mathfrak{M}}-\TP)U_{\TP}\|_{\text{max}}=\max\limits_{i,j\in[d]}\Big|\langle{(\mathfrak{M}-\TP)U_{\cdot j},U_{\cdot i}}\rangle\Big|$. For any $1\leq i,j \leq d$,
\begin{align*}
    \langle{(\mathfrak{M}-\TP)U_{\cdot j},U_{\cdot i}}\rangle &= U_{i \cdot}^T(\mathfrak{M}-\TP)U_{\cdot j} \\&=2\sum_{1\leq\tilde{k}<\tilde{l}\leq mn}(\mathfrak{M}_{\tilde{k},\tilde{l}}-\TP_{\tilde{k},\tilde{l}})U_{\tilde{k} i}U_{\tilde{l} j}-\sum_{\tilde{k}}
    \TP_{\tilde{k},\tilde{k}}U_{\tilde{k} i}U_{\tilde{k} j}\\&= 2\sum_{1\leq k<l\leq n}\Big(\sum_{s=1}^m\sum_{t=1}^m\mathfrak{M}^{(s,t)}_{k,l}-m^2P_{k,l}\Big)U_{k i}U_{l j}-m\sum_{k=1}^n
   P_{k,k}U_{k i}U_{k j}\\&=2\sum_{1\leq k<l\leq n}\Big(\sum_{s=1}^m\sum_{t=1}^m\sum_{\ell=1}^mc_{\ell}^{(s,t)}A^{(\ell)}_{k,l}-m^2P_{k,l}\Big)U_{k i}U_{l j}-m\sum_{k=1}^n
   P_{k,k}U_{k i}U_{k j}\\&=2\sum_{1\leq k<l\leq n}\underbrace{\Big(\sum_{\ell=1}^m \Big(\sum_{s=1}^m\alpha(s,\ell)\Big)A^{(\ell)}_{k,l}-m^2P_{k,l}\Big)}_{\text{mean }0 \text{ r.v.}}U_{k i}U_{l j}-m\sum_{k=1}^n
   P_{k,k}U_{k i}U_{k j}.
\end{align*}
Conditioned on $\bX$ (and hence on $P$), $ U_{i \cdot}^T(\mathfrak{M}-\TP)U_{\cdot j}$ consists of two terms. The first term is a sum of $\binom{n}{2}$ independent, mean zero, bounded random variables taking values in $[-2m^2U_{ki}U_{lj},2m^2U_{ki}U_{lj}]$ and the second term is of order $O(m)$, and thus, dominated by the first term. By Hoeffding's inequality,
\begin{align*}
    \PX\bigg[\Big|2\sum_{1\leq\tilde{k}<\tilde{l}\leq mn}(\mathfrak{M}_{\tilde{k},\tilde{l}}-\TP_{\tilde{k},\tilde{l}})U_{\tilde{k} i}U_{\tilde{l} j}\Big|\geq 8m^2\sqrt{\log mn}\bigg]&\leq 2\exp\bigg(\frac{-32m^4\log mn}{\sum\limits_{1\leq k<l \leq n}(4m^2U_{ki}U_{lj})^2}\bigg)\\&\leq 2(mn)^{-2}.
\end{align*}
Integrating over $\bX$ yields the desired result that there exists a constant $C>0$ such that
$$\|U_{\TP}^T({\mathfrak{M}}-\TP)U_{\TP}\|_F\leq Cm^2\sqrt{\log mn}\ \  \text{ w.h.p.}$$
Consider the events (where $C$ is an appropriately chosen constant)
\begin{align*}
    \mathcal{E}_1&:=\left\{\|\up^T(\fM-\TP)\up\|_F\leq Cm^2\log  mn \right\}\\
    \mathcal{E}_2&:=\left\{\text{ the statement of Lemma \ref{lemma:eig_order} holds} \right\}\\
        \mathcal{E}_3&:=\left\{\text{ the statement of Lemma \ref{lemma:evector_subspace} holds} \right\}\\
    \mathcal{E}_4&:=\left\{ \text{ the statement of Lemma \ref{lemma:conc} holds} \right\}
\end{align*}
As each $\mathcal{E}_i$ is a high probability event, $\mathcal{E}=\cap_{i=1}^4\mathcal{E}_i$ is also a high probability event.
In what follows, we will condition on the events in $\mathcal{E}$ occurring.

Conditioning on $\mathcal{E}$,
from the Davis-Kahan theorem (shown in detail in Lemma \ref{lem:Q1}), we have that
\begin{equation}
\label{eq:boundingeeee}
\|\um-\up\up^T\um\|
\leq \sqrt{ \frac{C\log mn}{n} }
\end{equation}
holds for some constant $C$.
Also, Weyl's theorem \cite[Section 6.3]{horn85:_matrix_analy} with Lemmas \ref{lemma:eig_order} and \ref{lemma:conc} imply that 
$$|\lambda_{j}(\fM)|\leq Cm\sqrt{n\log  mn}$$
for $j>d$,
and that
$$|\lambda_{j}(\fM)-\lambda_j(\TP)|\leq Cm\sqrt{n\log  mn}$$
for $j\leq d$.
As
for $j\leq d$, we have that there is a constant $C>0$ such that $\lambda_j(\TP)\geq Cnm\delta$, and so $\lambda_i(|\fM|)=\lambda_i(\fM)\geq Cnm\delta$ for all $i\in[d]$ and an appropriately chosen constant $C$ (abusing notation, the two $C$'s need not be equal). Therefore, we have
that $I_{\pm}=I_d$ so that 
$$
\spp\up^T\um I_\pm=\spp\up^T\um=\spp(\up^T\um-V)+\spp V.
$$
We have then that
\begin{align*}
\|V\sm-\spp V\|_F&\leq 
\|V-\up^T\um\|_F(\|\sm\|+\|\spp\|)\\
&\hspace{2mm}+
\|\up^T(\fM-\TP)R\|_F+
\|\up^T(\fM-\TP)\up\up^T\um\|_F.
\end{align*}
Now, we have that 
$
\|\fM\|,\|\TP\|\leq mn.
$
It follows then that there exists a constant $C$ such that (where the first bound follows from Lemma \ref{lemma:evector_subspace}, and the second by combining Lemma \ref{lemma:conc} and Eq. (\ref{eq:boundingeeee}))
\begin{align*}
\|V-\up^T\um\|_F(\|\sm\|+\|\spp\|)&\leq Cm\log mn\\
\|\up^T(\fM-\TP)R\|_F&\leq\sqrt{d}\|\up^T(\fM-\TP)R\|\leq Cm\log mn,
\end{align*}
so that
\begin{align*}
\|V\sm-\spp V\|_F\leq Cm\log mn + \|\up^T(\fM-\TP)\up\|_F\leq Cm^{2}\log mn.
\end{align*}
as desired, thus proving Eq. (\ref{eq:31}).

To prove Eq.\@ (\ref{eq:32}) and (\ref{eq:33}), note that 
we have shown that with high probability, both of the following events hold:
\begin{itemize}
\item[i.] $\|V\sm-\spp V\|_F\leq Cm^{2}\log mn$ (by the first part of the Lemma).
\item[ii.] For all $i$, $j\in[d]$, $\lambda_i(\fM),\lambda_j(\TP)$ are of order $\Omega(nm\delta)$ (by Lemmas \ref{lemma:eig_order} and \ref{lemma:conc}).
\end{itemize}
Given these events, as we have that 
$$(\sm^{1/2}-\spp^{1/2})(\sm^{1/2}+\spp^{1/2})=
(\sm-\spp)\Rightarrow 
(\sm^{1/2}-\spp^{1/2})=
(\sm-\spp)(\sm^{1/2}+\spp^{1/2})^{-1},
$$
the $i,j$-th entry of 
$V\sm^{1/2}-\spp^{1/2} V$ is equal to
$$V_{i,j}(\lambda_i^{1/2}(\fM)-\lambda_j^{1/2}(\TP))=
\frac{V_{i,j}(\lambda_i(\fM)-\lambda_j(\TP))}{\lambda_i^{1/2}(\fM)+\lambda_j^{1/2}(\TP)}
$$
and Eq. (\ref{eq:32}) follows. 
To prove Eq. (\ref{eq:33}), note that 
the $i,j$-th entry of 
$V\sm^{-1/2}-\spp^{-1/2} V$ is equal to
$$V_{i,j}(\lambda_i^{-1/2}(\fM)-\lambda_j^{-1/2}(\TP))=
\frac{V_{i,j}(\lambda_j(\TP)^{1/2}-\lambda_i(\fM)^{1/2})}{\sqrt{\lambda_i(\fM)\lambda_j(\TP)}}.
$$
The proof then follows immediately.
\end{proof}
\begin{lemma}
\label{lem:2toinfU}
With the assumptions in Lemma \ref{lemma:evector_subspace}, denote the SVD of $U_{\TP}^TU_{\mathfrak{M}}$ as $V_1\Sigma V_2^T$ and set $V:=V_1V_2^T$. 
There exists a constant $C$ such that w.h.p.
\begin{equation}
\label{eq:claim2}
    \|({\mathfrak{M}}-\TP)U_{\TP}\|_{2\rightarrow \infty}\leq C\sqrt{d}m\sqrt{\log mn}.
\end{equation}
\end{lemma}
\begin{proof} We note first that
$$\frac{1}{\sqrt{d}}\|(\mathfrak{M}-\TP)U_{\TP}\|_{2\rightarrow \infty}\leq \|(\mathfrak{M}-\TP)U_{\TP}\|_{\max}=\max_{i\in[mn],j\in[d]}|\langle (\mathfrak{M}-\TP)U_{\cdot j},e_i\rangle|,$$
where $e_i\in\R^{mn}$ is the unit vector with all of its entries equal to $0$, except for the $i$-th entry.

Let $s\in[m]$ arbitrary. 
There exists a matrix $W$ such that 
$$\up \spp^{1/2} W=\bZ\ \ \Rightarrow\ \  \up=\bZ W^T (\spp)^{-1/2}=\begin{bmatrix}
\bX W^T (\spp)^{-1/2}\\
\bX W^T (\spp)^{-1/2}\\
\vdots\\
\bX W^T (\spp)^{-1/2}
\end{bmatrix}.$$
Therefore, for all $1\leq k\leq n$, $U_{k,j}=U_{k+n,j}=\cdots=U_{k+(m-1)n,j}$.
For each $(s-1)n+1\leq i\leq sn$ and $1\leq j\leq d$,
\begin{align}
\label{eq:proc_max_norm}
    \langle(\mathfrak{M}-\TP)U_{\cdot j},e_i\rangle &= e_i^T(\mathfrak{M}-\TP)U_{\cdot j}\\&=\sum_{\tilde{k}=1}^{mn}(\mathfrak{M}_{i,\tilde{k}}-\TP_{i,\tilde{k}})U_{\tilde{k},j}\nonumber\\&=\Large\sum_{k=1}^{n}\Big[\sum_{t=1}^{m}\mathfrak{M}^{(s,t)}_{i,k}-mP_{i,k}\Big]U_{k,j}\nonumber\\&=\sum_{k=1}^{n}\Big[\sum_{t=1}^{m}\sum_{\ell=1}^{m}c_{\ell}^{(s,t)}A^{(\ell)}_{i,k}-mP_{i,k}\Big]U_{k,j}\nonumber\\&=\sum_{k=1}^{n}\Big[\sum_{\ell=1}^{m}\alpha(s,\ell)A^{(\ell)}_{i,k}-mP_{i,k}\Big]U_{k,j}\nonumber,
\end{align} where $\alpha(s,\ell):=\sum\limits_{t=1}^{m}c_{\ell}^{(s,t)}\geq 0$ and $\sum\limits_{\ell=1}^{m}a(s,\ell)=m$ for all $s\in [m]$.
Conditioning on $\bX$ (and hence on $\mathbf{P}$) for all $s\in[m]$, for any $(s-1)n+1\leq i\leq sn$ and $1\leq j\leq d$, the above expansion is a sum of $n$ independent (in $k$), bounded, mean zero random variables taking values in $[-mU_{k,j},mU_{k,j}]$. Hence, by Hoeffding's inequality, 
\begin{align*}
    \PX\Big(\Big|\sum_{\tilde{k}=1}^{mn}(\mathfrak{M}_{i,\tilde{k}}-\TP_{i,\tilde{k}})U_{\tilde{k},j}\Big|\geq 2m\sqrt{\log mn}\Big)\leq 2\exp{\Bigg(-\frac{8m^2\log mn}{\sum_{k=1}^{n}(2mU_{k,j})^2}\Bigg)}\leq 2(mn)^{-2},
\end{align*}
where we used the fact that $\sum_{k=1}^{n}(2mU_{k,j})^2\leq 4m^2$ as the columns of $\up$ are norm $1$.
Therefore, by integrating over $\bX$ we have that w.h.p. there exists a constant $C>0$ such that $\|(\mathfrak{M}-\TP)U_{\TP}\|_{2\rightarrow \infty}\leq C\sqrt{d}m\sqrt{\log mn}$ as desired.
\end{proof}

\begin{lemma}
\label{lem:Q1}
With the assumptions in Lemma \ref{lemma:evector_subspace}, denote the SVD of $U_{\TP}^TU_{\mathfrak{M}}$ as $V_1\Sigma V_2^T$ and set $V:=V_1V_2^T$. Set also $Q^{(1)}=U_{\mathfrak{M}}-U_{\TP}V$, there exists a constant $C$ such that w.h.p.,
\begin{equation*}
    \|Q^{(1)}\|\leq C\frac{\log^{1/2}mn}{n^{1/2}}.
\end{equation*}
\end{lemma}
\begin{proof}
Following the reasoning from Lemma 6.8 in \cite{cape2toinfty}, we add and substract $U_{\TP}U_{\TP}^TU_{\mathfrak{M}}$ and by triangle inequality,
\begin{align*}
    \|Q^{(1)}\|=\|U_{\mathfrak{M}}-U_{\TP}V\|\leq\|U_{\mathfrak{M}}-U_{\TP}U_{\TP}^TU_{\mathfrak{M}}\|+\|U_{\TP}(U_{\TP}^TU_{\mathfrak{M}}-V)\|.
\end{align*}
The first term can be rewritten as follows,
\begin{align*}
    \|U_{\mathfrak{M}}-U_{\TP}U_{\TP}^TU_{\mathfrak{M}}\|=\|U_{\mathfrak{M}}U_{\mathfrak{M}}^T-U_{\TP}U_{\TP}^TU_{\mathfrak{M}}U_{\mathfrak{M}}^T\|=\|(I-U_{\TP}U_{\TP}^T)U_{\mathfrak{M}}U_{\mathfrak{M}}^T\|=\|\sin{\Theta(U_{\mathfrak{M}},U_{\TP})}\|.
\end{align*}

\noindent 
Given the intersection of the events in the statements of Lemmas \ref{lemma:eig_order} and \ref{lemma:conc}, the Davis-Kahan theorem gives that, 
\begin{align*}
    \|U_{\mathfrak{M}}-U_{\TP}U_{\TP}^TU_{\mathfrak{M}}\|\leq \frac{\|{\mathfrak{M}}-\TP\|}{\lambda_d(\fM)}\leq C\frac{\log^{1/2}mn}{n^{1/2}}.
\end{align*}
Further, we can bound $\|U_{\TP}(U_{\TP}^TU_{\mathfrak{M}}-V)\|$ using Lemma \ref{lemma:evector_subspace}; as the intersection of the events in Lemma \ref{lemma:eig_order}, \ref{lemma:evector_subspace} and \ref{lemma:conc} has high probability, this leads us to the desired result.
\end{proof}

\vspace{.5cm}
\noindent
We now address the proof of Theorem \ref{theorem:genOMNI_consistency}.
\begin{proof} (Consistency of generalized omnibus embeddings)
Let $W_n$ be a sequence of matrices such that $\bZ=\bZ^*W_n$, and let $V$ be as in Lemma \ref{lemma:V_bounds}.
Define the matrices $Q^{(1)},Q^{(2)}$ as follows
\begin{align*}
    Q^{(1)}&= U_{\mathfrak{M}}-U_{\TP}V; \\
    Q^{(2)}&=U_{\TP}U_{\TP}^TU_{\mathfrak{M}}-U_{\TP}V.
\end{align*}
Note that, as defined, 
$\fM U_{\fM}=U_{\fM}S_{\fM}I_{\pm}$ where $I_{\pm}$ is a random sign matrix indicating whether the signs of the eigenvalues associated with $U_{\fM}$ agree for $\fM$ and $|\fM|$.
Given the events of Lemmas \ref{lemma:eig_order} and \ref{lemma:conc}, we have that (by Weyl's Theorem) $I_{\pm}=I_d$; as this term appears only in the residual terms ($H_2$-$H_5$) of the below decomposition, and as, when bounding the residual terms we assume the events of Lemmas \ref{lemma:eig_order} and \ref{lemma:conc}, we write, with a slight abuse of notation, $\fM U_{\fM}=U_{\fM}S_{\fM}$ in the below decomposition.
We decompose the term $U_{\mathfrak{M}}S_{\mathfrak{M}}^{1/2}-U_{\TP}S_{\TP}^{1/2}V$ as follows,
 \begin{align}
\label{eq:cons}
   U_{\mathfrak{M}}S_{\mathfrak{M}}^{1/2}-U_{\TP}S_{\TP}^{1/2}V &= \underbrace{(\mathfrak{M}-\TP)U_{\TP}S_{\TP}^{-1/2}V}_{:=H_1}+
   \underbrace{(\mathfrak{M}-\TP)U_{\TP}(VS_{\mathfrak{M}}^{-1/2}-S_{\TP}^{-1/2}V)}_{:=H_2}\nonumber\\
   &-\underbrace{U_{\TP}U_{\TP}^T(\mathfrak{M}-\TP)U_{\TP}VS_{\mathfrak{M}}^{-1/2}}_{:=H_3}+
   \underbrace{(I-U_{\TP}U_{\TP}^T)(\mathfrak{M}-\TP)Q^{(1)}S_{\mathfrak{M}}^{-1/2}}_{:=H_4}\nonumber\\&+
   \underbrace{Q^{(2)}S_{\mathfrak{M}}^{1/2}+U_{\TP}(VS_{\mathfrak{M}}^{1/2}-S_{\TP}^{1/2}V)}_{:=H_5}.
\end{align}
For the $H_i$'s, $i=1,2,3,4$, there exists a constant $C>0$ such that the following bounds hold w.h.p.:
\begin{align*}
    \|H_1\|_{2\rightarrow \infty}&\leq \| S_{\TP}^{-1/2}\| \cdot\|(\mathfrak{M}-\TP)U_{\TP}\|_{2\rightarrow \infty}\\
    &\leq  C\frac{m^{1/2}\log^{1/2}mn}{n^{1/2}},   \hspace{10mm}\text{ [Lemmas \ref{lemma:eig_order}, \ref{lem:2toinfU}]}\\ 
    \|H_2\|_{2\rightarrow \infty}&\leq \|VS_{\mathfrak{M}}^{-1/2}-S_{\TP}^{-1/2}V \|\cdot \|(\mathfrak{M}-\TP)U_{\TP} \|_{2\rightarrow \infty}\\
    &\leq 
    C\frac{m^{3/2}\log^{3/2}mn}{n^{3/2}}, \hspace{10mm} \text{ [Lemmas \ref{lemma:V_bounds}, \ref{lem:2toinfU}]}\\
    \|H_3\|_{2\rightarrow \infty}&\leq 
    \|U_{\TP}\|_{2\rightarrow\infty} 
    \|\up^T(\mathfrak{M}-\TP)\up\|\cdot 
    \|S_{\mathfrak{M}}^{-1/2}\|
    \\
    &\leq 
    C\frac{m\log mn}{n},  \hspace{10mm} \text{ [Lemma \ref{lemma:eig_order}; Eq. (\ref{eq:U_P}),(\ref{eq:u(M-P)u})]}\\
    \|H_4\|_{2\rightarrow \infty}&
    \leq \|(I-U_{\TP}U_{\TP}^T)\|_{2\rightarrow\infty}
    \|(\mathfrak{M}-\TP)\| 
    \|Q^{(1)}\| 
    \|S_{\mathfrak{M}}^{-1/2}\| \\
    &\leq C\frac{m^{1/2}\log mn}{n^{1/2}}. \hspace{10mm} \text{ [Lemmas \ref{lemma:eig_order}, \ref{lemma:conc}, \ref{lem:Q1}
    ]}.
\end{align*}

Next, by Lemmas \ref{lemma:eig_order}, \ref{lemma:evector_subspace} and Eq. (\ref{eq:U_P}) we get 
\begin{align*}
\|Q^{(2)}S_{\fM}^{1/2}\|_{2\rightarrow \infty}\leq 
\|U_{\TP}\|_{2\rightarrow\infty} 
\|U_{\TP}^TU_{\mathfrak{M}}-V \|  
\cdot \|S_{\fM}^{1/2} \| 
\leq \frac{C\log mn}{n} ,
\end{align*}
and by Lemma \ref{lemma:V_bounds} and Eq (\ref{eq:U_P}) we get \begin{align*}
    \|U_{\TP}(VS_{\fM}^{1/2}-S_{\TP}^{1/2}V)\|_{2\rightarrow \infty}\leq 
    \|U_{\TP}\|_{2\rightarrow\infty}\|(VS_{\fM}^{1/2}-S_{\TP}^{1/2}V)\|\leq 
    \frac{Cm\log mn}{n}.
\end{align*} 
Hence, by triangle inequality,
\begin{align*}
    \|H_5\|_{2\rightarrow\infty}\leq C\frac{m\log mn}{n}.
\end{align*}
Plugging into Eq. (\ref{eq:cons}) the bounds above, for sufficiently large $n$, applying triangle inequality yields the desired result.
\end{proof}


\subsection{Central Limit Theorem for the rows of the generalized omnibus embedding}
\label{sec:genomnipf}

To prove Theorem \ref{thm:genOMNI} from Section \ref{sec:genomni}, we adapt the proof of Theorem 1 in \cite{levin_omni_2017} and further we extend it to correlated networks. The main difficulty in this adaptation is the more complex structure of the general $\fM$, which requires a number of modifications that we have already completed. 
\begin{itemize}
\item Account for the model correlation in Lemmas \ref{lemma:conc}, \ref{lemma:V_bounds} and \ref{lem:2toinfU}.
\item The proofs of the Bernstein matrix concentration result and Lemma \ref{lemma:V_bounds} necessitate a more delicate decomposition of $\fM$ in order to leverage classical concentration inequality results.
Moreover, the proof of Lemma \ref{lemma:conc} highlights how the coefficient matrices $C^{(l)}$ of the adjacency matrices $A^{(l)}$ in $\mathfrak{M}$ fully characterize the omnibus  matrix $\mathfrak{M}$. 
    \item 
    We adapt the general exchangeability result in the proof of Lemma 5 of \cite{levin_omni_2017} (used there to bound the term `$B_1$') to our current setting, and the general block-form of $\fM$ still allows for a weaker (within block) exchangeability argument to be employed, which is sufficient for our purposes.
    \item The model correlation and the weights in the $\fM$ matrix necessitate novel decompositions to compute the relevant covariance structures.
    \item Considering row-wise differences of $\fM$ is a key contribution to the literature.
\end{itemize}
\noindent
The overall layout of the proof is as follows. 
Let $\bhx_{\fM_n}=\mathrm{ASE}(\fM_n,d)=U_{\fM_n}S_{\fM_n}^{1/2}\in\R^{mn\times d}$ and let $\bZ_n:=[\bX_n^T|\bX_n^T|\cdots|\bX_n^T]^T\in\mathbb{R}^{mn\times d}$. 
Fix indices $i\in[n]$ and $s\in[m]$, and let $h=n(s-1)+i$. The quantity of interest is the $h$-th row (i.e., $i$-th row from $s$-th block) of
$$n^{1/2}\Big(\bhx_{\fM_n}V_n^TW_n-\bZ_n
\Big),$$
where $V_n,W_n\in\R^{d\times d}$ are suitable orthogonal transformations and $W_n$ is such that $\bZ_n= U_{\TP_n}S_{\TP_n}^{1/2}W_n$. This quantity can be decomposed into a sum of matrices (Eq. (\ref{eq:cons})) as (dropping the subscripted dependence on $n$)
\begin{align*}
n^{1/2}\Big(U_{\fM}S_{\fM}^{1/2}-U_{\TP}S_{\TP}^{1/2}V\Big)_h V^T W &=n^{1/2}\Big((\mathfrak{M}-\TP)U_{\TP}S_{\TP}^{-1/2}V\Big)_hV^T W +(n^{1/2}R_hV)V^T W\\
&=n^{1/2}\Big((\mathfrak{M}-\TP)U_{\TP}S_{\TP}^{-1/2}W\Big)_h+n^{1/2}R_h W,
\end{align*}
where $R_h\in\R^{mn\times d}$ is the residual matrix. 
Next, Theorem \ref{lemma:normality} establishes that $$n^{1/2}\Big((\mathfrak{M}-\TP)U_{\TP}S_{\TP}^{-1/2}W\Big)_h $$
converges in distribution to a mixture of normals.

We begin with a more limited conditional central limit theorem.
Recall that, by the definition of the JRDPG, the latent positions of the expected omnibus matrix $\EX \fM=\TP=U_{\TP}S_{\TP}U_{\TP}^T$ are given by 
\[Z^{\star}=\begin{bmatrix}X^{\star}\\X^{\star}\\\vdots\\X^{\star} \end{bmatrix}=U_{\TP}S_{\TP}^{1/2}\in\R^{mn\times d}.\]
Recall that the matrix of the true latent positions is denoted by  $Z=[X^T X^T \cdots X^T]^T\in\R^{mn\times d}$, so that $Z=Z^{\star}W$ for some suitable-chosen orthogonal matrix $W$.

\begin{theorem}
\label{lemma:normality}
With notation and assumptions as in Theorem \ref{thm:genOMNI}, fix some $i\in[n]$ and some $s\in[m]$ and let $h=n(s-1)+i$. Conditional on $X_i=x_i\in\R^{d}$, there exists a sequence of $d$-by-$d$ orthogonal matrices $\{W_n\}_n$ such that
\begin{align*}
    n^{1/2}\Big[(\fM_n-\TP_n)U_{\TP_n}S_{\TP_n}^{-1/2}\Big]_hW_n\xrightarrow{\mathcal{L}}\mathcal{N}(0,\Check\Sigma(x_i)),
\end{align*}
where 
\begin{align*}
        \Check\Sigma_{\rho}(x_i;s)&=\frac{1}{m^2}\Big(\underbrace{\sum_{q=1}^{m}\alpha^2(s,q)}_{\text{method coefficient}}+\underbrace{2\sum_{q<l}\alpha(s,q)\alpha(s,l)\rho_{q,l}}_{\text{model coefficient}}\Big)\underbrace{\Delta^{-1}\EX[x_i^TX_j(1-x_i^TX_j)X_jX_j^T]\Delta^{-1}}_{:=\Sigma(x_i)} 
\end{align*}
is a covariance matrix that depends on $x_i$.
\end{theorem}

\begin{proof} Following the proof of Lemma 6 in \cite{levin_omni_2017}, for each $n=1,2,\cdots$, choose orthogonal $W_n\in\R^{d\times d}$ so that $\bZ=\bZ^{\star}W_n=U_{\TP}\spp^{1/2} W_n$. 
Dropping the explicit dependence on $n$, we rewrite the term 
$n^{1/2}\Big[(\fM-\TP)U_{\TP}S_{\TP}^{-1/2}\Big]_hW$  as follows
\begin{align}
\label{eq:mainclt}
    n^{1/2}\Big[(\fM-\TP)U_{\TP}S_{\TP}^{-1/2}\Big]_h W&=
    n^{1/2}\Big[(\fM-\TP)U_{\TP}S_{\TP}^{-1/2}W\Big]_h \notag \\ 
    &= n^{1/2}\Big[(\fM-\TP)U_{\TP}S_P^{1/2}W W^TS_{\TP}^{-1}W\Big]_h\notag \\ 
     &= n^{1/2}\Big[(\fM-\TP)\bZ \Big]_h W^TS_{\TP}^{-1}W\notag \\ 
     &=\frac{n^{-1/2}}{m}\Big[\fM \bZ-\TP \bZ\Big]_h\left[nW^TS_{P}^{-1}W\right].
\end{align}
Consider next the scaled rows of $\fM \bZ-\TP \bZ$.  
We write (where for a matrix $A$, either $(A)_k$ or $[A]_k$ are used to denote the $k$-th row of $A$ and $(A)_{hk}$ or $[A]_{hk}$ denoting $A_{hk}$)
\begin{align*}
    \frac{n^{-1/2}}{m}[\fM \bZ-\TP \bZ]_h&=
    \frac{n^{-1/2}}{m}\sum_{k=1}^{mn}(\fM-\TP)_{hk}(\bZ)_k
    \\&=\frac{n^{-1/2}}{m}\sum_{\ell=1}^m\sum_{j=1}^{n}\Bigg(\fM_{i,j}^{(s,\ell)}-P_{ij}\Bigg)X_j\\&=\frac{n^{-1/2}}{m}\sum_{j\neq i}\Bigg(\sum_{\ell=1}^m \fM_{i,j}^{(s,\ell)}-mP_{ij}\Bigg)X_j-n^{-1/2}P_{ii}X_i\ .
\end{align*}

\vspace{0.2cm}
Conditioning on $X_i=x_i\in\R^d$, we first observe that 
\begin{align*}
    \frac{P_{ii}}{n^{1/2}}X_i=\frac{x_i^Tx_i}{n^{1/2}}x_i\xrightarrow{a.s.}0.
\end{align*}
Moreover, the remaining portion of the scaled sum becomes 
\begin{align*}
    &n^{-1/2}\sum_{j\neq i}\frac{1}{m}\Bigg(\sum_{\ell=1}^m \fM_{i,j}^{(s,\ell)}-m(x_i^TX_j)\Bigg)X_j\\
    &=n^{-1/2}\sum_{j\neq i}\frac{1}{m}\bigg(\sum_{\ell=1}^{m}\left(\sum_{q=1}^m c_q^{(s,\ell)}A^{(q)}_{i,j}\right)-m(x_i^TX_j)\bigg)X_j\\
    &=n^{-1/2}\sum_{j\neq i}\frac{1}{m}\bigg(\sum_{q=1}^m\underbrace{\bigg(\sum_{\ell=1}^{m}c_q^{(s,\ell)}\bigg)}_{:=\alpha(s,q)} A^{(q)}_{i,j}-m(x_i^TX_j)\bigg)X_j,
\end{align*} where $\sum\limits_{q=1}^{m}\alpha(s,q)=m$ for all $s\in [m]$.
The above expression is a sum of $n-1$ independent $0$-mean random variables, the $$\frac{1}{m}\Big(\sum\limits_{q=1}^m\alpha(s,q) A^{(q)}_{i,j}-m(x_i^TX_j)\Big)X_j\Big),$$ each with covariance matrix denoted by $\Grave\Sigma(x_i)$, computed as follows (suppressing the conditioning on $X_i=x_i$) 
\begin{align*}
\Grave\Sigma(x_i)&=
    \frac{1}{m^2}\Ex\Big(\bigg(\sum_{q=1}^{m}\alpha(s,q)A^{(q)}_{i,j}-m(x_i^TX_j)\bigg)^2X_jX_j^T\Big)\\
    &=\frac{1}{m^2}\EX\Big[\EX\Big[\Big(\sum_{q=1}^{m}\alpha(s,q)A^{(q)}_{i,j}-m(x_i^TX_j)\Big)^2X_jX_j^T\,\Big|\,X_j\Big]\Big]\\
    &=\frac{1}{m^2}\EX\Big[\bigg(\EX\Big[\Big(\sum_{q=1}^m \alpha(s,q)A_{i,j}^{(q)}\Big)^2|X_j\Big]-m^2(x_i^TX_j)^2\bigg)X_jX_j^T\Big]\\
    &=\frac{1}{m^2}\EX\Bigg[\bigg(\Big(\sum_{q=1}^{m}\alpha^2(s,q)\Big)x_i^TX_j+\Big(2\sum_{q<l}\alpha(s,q)\alpha(s,l)-m^2\Big)(x_i^TX_j)^2\\&\hspace{2.5cm}+\Big(2\sum_{q<l}\alpha(s,q)\alpha(s,l)\rho_{q,l}\Big)x_i^TX_j(1-x_i^TX_j)\bigg)X_jX_j^T\Bigg]
    \\&=\frac{1}{m^2}\left(\sum_{q=1}^{m}\alpha^2(s,q)+2\sum_{q<l}\alpha(s,q)\alpha(s,l)\rho_{q,l}\right)\EX\Big[\left(x_i^TX_j-(x_i^TX_j)^2\right)X_jX_j^T\Big],
\end{align*} where the third equality holds because conditioning on $P=XX^T$,
$$\EX\Big[\sum_{q=1}^{m}\alpha(s,q)A^{(q)}\,\Big| P\Big]=mXX^T.$$
\noindent
The fourth equality holds from the fact that conditioning on $X_i$ and $X_j$, $$\EX[(A^{(q)}_{i,j})^2|X_i,X_j]=X_i^TX_j$$ for all $q\in [m]$, and from the fact that conditioning on $X_i$ and $X_j$, $(A^{(q)}_{i,j},A^{(l)}_{i,j})$ are $\rho_{q,l}$-correlated Bernoulli($X_i^TX_j)$ random variables, $$\EX[A^{(q)}_{i,j}A^{(l)}_{i,j}|X_i,X_j]=(X_i^TX_j)^2+\rho_{q,l}X_i^TX_j(1-X_i^TX_j)$$ for all $q\neq l$. Lastly, the final equality holds by observing $$m^2=\left(\sum_{q=1}^m\alpha(s,q)\right)^2=\sum_{q=1}^m\alpha(s,q)^2+2\sum_{q<l}\alpha(s,q)\alpha(s,l).$$ 

\noindent
Thus, by the multivariate central limit theorem we have that
\begin{align}
\label{eq:clt}
   n^{-1/2}\sum_{j\neq i}\frac{1}{m}\Bigg(\sum_{q=1}^m \fM_{i,j}^{(s,q)}-m(x_i^TX_j)\Bigg)X_j  \xrightarrow{\mathcal{L}}\mathcal{N}(0,\Grave\Sigma(x_i)).
\end{align}
Next, recall (as in the proof of Theorem \ref{thm:rhoCLT}) $nW_n^TS_P^{-1}W_n\xrightarrow{a.s.}\Delta^{-1}$, so that 
by the multivariate version of Slutsky's theorem in Eq. (\ref{eq:mainclt}) we get
\begin{align*}
    n^{1/2}\Big[(\fM-\TP)U_{\TP}S_{\TP}^{-1/2}\Big]_hW_n\xrightarrow{\mathcal{L}}\mathcal{N}(0, \Check\Sigma_{\rho}(x_i;s)),
\end{align*}
where $ \Check\Sigma_{\rho}(x_i;s))=\Delta^{-1}\Grave\Sigma(x_i)\Delta^{-1}$.
\end{proof}

\vspace{.3cm}
\noindent
This equips us for the proof of Theorem \ref{thm:genOMNI}.
\begin{proof} (Central Limit Theorem for the rows of the generalized omnibus matrix)
Fix the index $h\in[mn]$ as $h=n(s-1)+i$, where $s\in[m]$, $i\in[n]$. Recall that for each $n=1,2,\cdots$, we choose orthogonal $W_n\in\R^{d\times d}$ so that $\bZ=\bZ^{\star}W_n=U_{\TP}\spp^{1/2} W_n$. 
With $V_n$ defined as in Lemma \ref{lemma:evector_subspace}, the matrix difference $$n^{1/2}\Big(U_{\fM_n}S_{\fM_n}^{1/2}-U_{\TP_n}S_{\TP_n}^{1/2}V_n\Big)_h V_n^T W_n$$
can be decomposed into a sum of matrices (as shown in Eq.(\ref{eq:cons}) of the proof of Theorem \ref{theorem:genOMNI_consistency}) as  follows
$$n^{1/2}\Big(U_{\fM_n}S_{\fM_n}^{1/2}-U_{\TP_n}S_{\TP_n}^{1/2}V_n\Big)_hV_n^T W_n=n^{1/2}\Big((\mathfrak{M}_n-\TP_n)U_{\TP_n}S_{\TP_n}^{-1/2}W_n\Big)_h+n^{1/2}R_{h,n}V_nW_n,$$
where $R_{h,n}\in\R^{mn\times d}$ is the $h-$th row of the residual matrix defined as $R=H_2-H_3+H_4+H_5$. 
\noindent
Further, in the proof of Theorem \ref{theorem:genOMNI_consistency} it is shown that for any $i=2,3,5$ we have
\begin{align*}
    n^{1/2}\|H_i\|_{2\rightarrow\infty}\leq \frac{Cm\log mn}{n^{1/2}}\hspace{1cm}\text{w.h.p.}
\end{align*}
It remains to provide an appropriate bound for $H_4$.

Rather than showing $\sqrt{n}\|H_4\|_{2\rightarrow\infty}$ converges to $0$ in probability,
we will prove directly that
$\sqrt{n}\|(H_4)_h\|_{2}$ converges to $0$ in probability, which is sufficient.
Adapting the exchangeability bound on the analogous term from \cite{levin_omni_2017}, we recall that (recalling that $I_{\pm}$ is the random diagonal sign matrix designed to give $\fM U_{\fM}=U_{\fM}S_{\fM}I_{\pm}$)
\begin{align*}
H_4:=&(I-U_{\TP}U_{\TP}^T)(\mathfrak{M}-\TP)Q^{(1)}S_{\mathfrak{M}}^{-1/2}I_{\pm}\\
=&(I-U_{\TP}U_{\TP}^T)(\mathfrak{M}-\TP)(\um-\up\up^T\um)S_{\mathfrak{M}}^{-1/2}I_{\pm}\\
&+(I-U_{\TP}U_{\TP}^T)(\mathfrak{M}-\TP)(\up\up^T\um-\up V)S_{\mathfrak{M}}^{-1/2}I_{\pm}\ .
\end{align*}
Denote the first term in this expression via $H_{41}$ (the second via $H_{42}$), and note that
\begin{align*}
H_{41}:=&(I-U_{\TP}U_{\TP}^T)(\mathfrak{M}-\TP)(\um-\up\up^T\um)S_{\mathfrak{M}}^{-1/2}I_{\pm}\\
=&\underbrace{(I-U_{\TP}U_{\TP}^T)(\mathfrak{M}-\TP)(I-\up\up^T)\um\um^T}_{:=E_1}(\um S_{\mathfrak{M}}^{-1/2}I_{\pm})\\
H_{42}:=&(I-U_{\TP}U_{\TP}^T)(\mathfrak{M}-\TP)(\up\up^T\um-\up V)S_{\mathfrak{M}}^{-1/2}I_{\pm}\\
=&(I-U_{\TP}U_{\TP}^T)(\mathfrak{M}-\TP)\up(\up^T\um-\up) VS_{\mathfrak{M}}^{-1/2}I_{\pm}\ .
\end{align*}
Considering first $H_{41}$, we have that 
\begin{align*}
\|(H_{41})_h\|_2\leq 
\|(E_1)_h\|_2\|\um \sm^{-1/2}\|.
\end{align*}
Consider now the term $\|(E_1)_h\|_{2}$.
For any symmetric matrix $B\in\mathbb{R}^{k\times k}$, we define $\Pi_{d}({B})$ to be the orthogonal projection onto the eigenspace corresponding to the eigenvectors of $B$ with the $d$ largest (in magnitude) eigenvalues.
Similarly, let $\Pi_{d}^{\perp}({\bf B})$ denote the orthogonal projection onto $(\Pi_{d}({B}))^\perp$.

Note that for any permutation matrix $Q$, we have that
\begin{align*}
\Pi_{d}(QB Q^T)&=Q\Pi_{d}(B)Q^T\\
\Pi_{d}^\perp(QB Q^T)&=Q\Pi_{d}^\perp(B)Q^T.
\end{align*}
As in \cite{levin_omni_2017}, for any $(B, H)\in\mathbb{R}^{k\times k}\times \mathbb{R}^{k\times k}$, define the operator $\mathcal{L}(B,H)$ via:
	$$\mathcal{L}(B, H)=\Pi_d^{\perp}(H)(B-H)\Pi_d^{\perp}(H)\Pi_d(B).$$
Let us consider permutations $\TQ\in\mathbb{R}^{mn\times mn}$ of the form
$$\TQ=I_m \otimes Q
$$
for permutations $Q\in\mathbb{R}^{n\times n}$.
Note that, for such $\TQ$, we have that
$$
\TQ \fM \TQ^T=\begin{pmatrix}
Q\fM^{(1,1)}Q^T & Q\fM^{(1,2)}Q^T &\cdots &Q\fM^{(1,m)}Q^T\\
Q\fM^{(1,2)}Q^T & Q\fM^{(2,2)}Q^T &\cdots &Q\fM^{(2,m)}Q^T\\
\vdots& \vdots &\ddots&\vdots\\
Q\fM^{(1,m)}Q^T & Q\fM^{(2,m)}Q^T &\cdots &Q\fM^{(m,m)}Q^T
\end{pmatrix},
$$
and that for each $(i,j)$ pair,
$$Q\fM^{(i,j)}Q^T=\sum_\ell c^{(i,j)}_\ell QA^{(\ell)}Q^T.$$
Similarly, 
$$\TQ \TP \TQ^T=J_m\otimes (Q P Q^T).
$$
Note that $\mathcal{L}(\fM, \TP)=E_1$, and as these orthogonal projections are unique, we have
	\begin{equation} \label{eq:E1_underpermutation}
	\begin{aligned}
	    \mathcal{L}&(\TQ \fM  \TQ^T, \TQ \TP \TQ^T)\\
	    &=\TQ (I-\up\up^T)\TQ^T\TQ(\fM-\TP)\TQ^T \TQ(I-\up \up^T) \TQ^T \TQ(\um \um^T)\TQ^T\\
	    &=\TQ E_1 \TQ^T.
	\end{aligned} \end{equation}
Since the rows of $\bX$ are i.i.d.\@ and the correlation across a given pair of graphs is the same for all edge pairs (i.e., $\text{corr}(A^{(k)}_{i,j},A^{(\ell)}_{i,j})=R_{k,\ell}$ independent of $i$ and $j$), together this implies that the matrix-pair entries of $(\fM,\TP)$ are equal in distribution to those of 
$(\TQ\fM \TQ^T,\TQ\TP \TQ^T)$.
Therefore, the entries of 
$\mathcal{L}(\TQ \fM  \TQ^T, \TQ \TP \TQ^T)=\TQ E_1\TQ^T$ are equal in law to those of $\mathcal{L}( \fM  ,  \TP )= E_1$.
Therefore, for each row $i$ we have that
$$
\|(\TQ E_1)_i\|^2=\|(\TQ E_1\TQ^T)_i\|^2\stackrel{\mathcal{L}}{=}\| (E_1)_i\|^2.
$$
This implies that 
if $Q_{i,j}=1$ then
for any $k\in[m]$, 
\begin{equation}
\label{eq:nodepend}
\Ex(\| (E_1)_{(k-1)n+i}\|^2)=\EX(\|(\TQ E_1 )_{(k-1)n+i}\|^2) =\Ex(\| (E_1)_{(k-1)n+j}\|^2).
\end{equation}
This guarantees that 
$$\EX(\|(E_1 )_{(k-1)n+i}\|^2)$$ depends only on $k$ and not on $i$,  and note that the analogous result follows immediately for $$\EX(\|(\widetilde QE_1 )_{(k-1)n+i}\|^2).$$
We can then define for $i,j\in[n]$,
$$r_k:=\EX(\|(\TQ E_1 )_{(k-1)n+i}\|^2) =\Ex(\| (E_1)_{(k-1)n+j}\|^2).$$
Observe that
$$\Ex(\|E_1\|^2_F)=\sum_{k=1}^m n r_k\geq n\max_k r_k.$$
Because $h=(s-1)n+i$, for $i \in [n]$, Eq.\eqref{eq:nodepend} and an application of Markov's inequality yield
\begin{align*}
    \p(\sqrt{n}\|(E_1)_h\|>t)&\leq \frac{n\Ex(\|(E_1)_h\|^2)}{t^2}=\frac{nr_s}{t^2}\\
    &\leq\frac{\Ex(\|E_1\|^2_F)}{t^2}.
\end{align*}
The entries of $\fM-\TP$ are bounded between $[-1,1]$, and $(I-\up \up^T)\um\um^T$ is rank $d$ (with spectral norm bounded by $1$). Hence, globally 
$$
\|E_1\|^2_F\leq \|I-\up \up^T\|^2 \|\fM-\TP\|_F^2 \|(I-\up \up^T)\um\um^T\|^2_F\leq m^2n^2 d^2.
$$
Now, we also have that, with high probability (so that the bad set has probability bounded above by $Cn^{-2}$), there exists a constant $C>0$ such that
\begin{align*}
\|E_1\|^2_F&\leq  d^2
\|I-\up \up^T\|^2 
\|\fM-\TP\|^2 
\|(I-\up \up^T)\um\|^2\|\um^T\|^2_F\\
&\leq C\frac{m^2 n\log^2 mn}{n }=Cm^2\log^2 mn \hspace{10mm} \text{by Lemma \ref{lemma:conc}, and  Eq.(\ref{eq:boundingeeee})}.
\end{align*}
We then have that there exists a constant $C>0$ such that
$$
\Ex(\|E_1\|^2_F)\leq C(m^2\log^2 mn+ n^{-2} m^2n^2 d^2)=C(m^2\log^2 mn+ m^2 d^2).
$$
Letting $t=n^{1/4}$ in our Markov bound, we see that 
$\p(\sqrt{n}\|(E_1)_h\|>n^{1/4})\rightarrow 0$.
As, w.h.p., we have that $\|\um \sm^{-1/2}\|\leq C/\sqrt{mn}$, we have that
\begin{align*}
n^{1/2}\|(H_{41})_h\|_2&\leq n^{1/3}\|(E_1)_h\|\, n^{1/6}\|\um \sm^{-1/2}\|\stackrel{P}{\rightarrow} 0.\end{align*}

Turning our attention to $H_{42}$, we have that (w.h.p.)
\begin{align*}
\|H_{42}\|_F
&\leq \|(I-U_{\TP}U_{\TP}^T)(\mathfrak{M}-\TP)(\up\up^T\um-\up V)S_{\mathfrak{M}}^{-1/2}I_{\pm}\|_F\\
&\leq \|(I-U_{\TP}U_{\TP}^T)\| \cdot\|\mathfrak{M}-\TP\|\cdot\|\up\|\cdot
\|\up^T\um- V\|_F
\|S_{\mathfrak{M}}^{-1/2}\|\\
&\leq C (m \sqrt{ n \log mn}) \left(\frac{\log mn}{n}\right)\left(\frac{1}{(mn)^{1/2}}\right)
\leq 
C \frac{m^{1/2}\log^{3/2}mn}{n},
\end{align*}
where the last line follows from Lemmas \ref{lemma:eig_order}, 
\ref{lemma:conc}, and
\ref{lemma:evector_subspace}.
Therefore, $n^{1/2}H_{42}$ converges to $0$ in probability, and combined this yields that the $h$-th row of $n^{1/2}H_4$ converges to $0$ as desired.

Since the matrix $W_n$ is unitary, the bounds above imply that the term $n^{1/2}R_{h,n}W_n$ converges to $0$ in probability.
Moreover, by Lemma \ref{lemma:normality} we have
\begin{align*}
   \lim_{n\rightarrow\infty}\PX\Big[ n^{1/2}\Big((\fM-\TP)U_{\TP}S_{\TP}^{-1/2}\Big)_hW_n\leq x\Big]\xrightarrow{\mathcal{D}}\int_{\text{suppF}}\Phi(x,\Check\Sigma_{\rho}(y;s))dF(y),
\end{align*}by integrating over the latent positions $X_i$. Finally, an application of Slutsky's theorem completes the proof.
\end{proof}

\subsection{Limiting correlation for the general omnibus embedding}\label{sec:app_pfdiffgenomni}

Here, we supply details for the computation of the limiting correlation across rows of the general omnibus embedding; this is the content of Theorem \ref{thm:pfdiffgenomni} from Section \ref{sec:indcorr_gen}.
\begin{proof}
We mimic the proof of Theorem \ref{thm:rhoCLT} here, and so omit some detail. Fix some $i\in[n]$ and some $s_1,s_2\in[m]$ and for $j=1,2$, let $h_j=n(s_j-1)+i\in [mn]$.
Conditioning on $X_i=x_i$, 
analogous to Eq.(\ref{eq:key2}), we write (where $\mathcal{Q}_n=V_n^TW_n$ as defined in the proof of Theorem \ref{thm:genOMNI})
\begin{align*}
    n^{1/2} &\left( \Big(\bhx_{\fM}V_n^TW_n\Big)_{h_1} -\Big(\bhx_{\fM}V_n^TW_n\Big)_{h_2}\right)\\
    &= n^{1/2}\Big((\mathfrak{M}_n-\TP_n)\bZ\Big)_{h_1}- n^{1/2}\Big((\mathfrak{M}_n-\TP_n)\bZ\Big)_{h_2}+o_P(1)\\
&=\frac{n^{-1/2}}{m}\left(\sum_{j\neq i}\Big(\sum_{q=1}^m\fM_{i,j}^{(s_1,q)}-\fM_{i,j}^{(s_2,q)}\Big)X_j\right)[nW_n^TS_P^{-1}W_n]+o_P(1)\\
&=n^{-1/2}\left(\sum_{j\neq i}\frac{1}{m}\Big(\sum_{q=1}^m\alpha(s_1,q)A_{i,j}^{(q)}-\alpha(s_2,q)A_{i,j}^{(q)}\Big)X_j\right)[nW_n^TS_P^{-1}W_n]+o_P(1).
\end{align*}
Each of the $n-1$ terms, $\frac{1}{m}\Big(\sum_{q=1}^m\alpha(s_1,q)A_{i,j}^{(q)}-\alpha(s_2,q)A_{i,j}^{(q)}\Big)X_j$,
is an independent, mean zero, random variable, with common covariance matrix $\Phi_{\rho}(x_i,s_1,s_2)$.
The desired result will then follow from an application of the multivariate central limit theorem and multivariate Slutsky theorems (as in the proof of Theorem \ref{thm:rhoCLT}), provided we can show the right form for $\Phi_{\rho}(x_i,s_1,s_2)$.
To this end, we consider (suppressing the conditioning on $X_i=x_i)$
\begin{align}
\label{eq:gencorrind}
\Phi_{\rho}(x_i,s_1,s_2)&:=
\frac{1}{m^2}\Ex\Big(\bigg(\sum_{q=1}^{m}\alpha(s_1,q)A^{(q)}_{i,j}-\sum_{q=1}^{m}\alpha(s_2,q)A^{(q)}_{i,j})\bigg)^2X_jX_j^T\Big)\notag\\
    &=\frac{1}{m^2}\bigg[\EX[(x_i^TX_j)X_jX_j^T]\sum_{q=1}^{m}(\alpha(s_1,q)-\alpha(s_2,q))^2\notag\\
    &\hspace{3mm}+2\EX[(x_i^TX_j)^2X_jX_j^T]\sum_{q<l}(\alpha(s_1,q)-\alpha(s_2,q))(\alpha(s_1,l)-\alpha(s_2,l))\Big)\notag\\
    &\hspace{3mm}+2\EX[x_i^TX_j(1-x_i^TX_j)X_jX_j^T]\sum_{q<l}(\alpha(s_1,q)-\alpha(s_2,q))(\alpha(s_1,l)-\alpha(s_2,l))\rho_{q,l}\bigg]\notag\\
        &=\frac{1}{m^2}\bigg[\sum_{q=1}^{m}(\alpha(s_1,q)-\alpha(s_2,q))^2\notag\\
    &\hspace{3mm}+2\sum_{q<l}(\alpha(s_1,q)-\alpha(s_2,q))(\alpha(s_1,l)-\alpha(s_2,l))\rho_{q,l}\bigg]\Sigma(x_i)\notag\\
    &=\frac{1}{m^2}\bigg(
    2\sum_{q<l}(\alpha(s_1,q)-\alpha(s_2,q))(\alpha(s_1,l)-\alpha(s_2,l))(\rho_{q,l}-1)\bigg)\Sigma(x_i),
\end{align}
where $\Sigma(x_i)=\EX[(x_i^TX_j-(x_i^TX_j)^2)X_jX_j^T]$, and the third equality follows from 
\begin{align*}
0&=\left(\sum_q \alpha(s_1,q)-\alpha(s_2,q) \right)^2\\
&=\sum_q \left(\alpha(s_1,q)-\alpha(s_2,q) \right)^2+2\sum_{q<l}(\alpha(s_1,q)-\alpha(s_2,q))(\alpha(s_1,l)-\alpha(s_2,l)).
\end{align*}
The proof then follows mutatis mutandis as that of Theorem \ref{thm:rhoCLT}.
\end{proof}

\subsection{Central Limit Theorem for the rows of the average embeddings}\label{sec:app_pfdiffgenomni2}
We supply details for the computation of the limiting covariance matrix of the rows of the sums of the aligned embeddings across the graphs, as described in Theorem \ref{thm:CLT-effective-samplesize} from Section \ref{sec:experiments-simulations}. As in the previous theorems, the proof mimics the one of Theorem \ref{thm:rhoCLT}, and so we omit some details. 

\begin{proof} To prove part a), fix some $i\in[n]$ . By the same arguments than Equations~\eqref{eq:key} and \eqref{eq:key11}, observe that
\begin{align*}
\sqrt{n} \left(\frac{1}{m}\sum_{s=1}^m \widehat \bX_{A_n}^{(k)} W^{(s)}_n -  \bX_n \right)_i=  \sqrt{n} \left(\frac{1}{m}\sum_{s=1}^m(A_n^{(s)} - P_n) U_PS_P^{-1/2}W_n \right)_i + O(n^{-1/2}\log n),
\end{align*}
where $W^{(s)} = W^{(s,1)}(W^{(s,2)})^T$, with $W^{(s,1)}\Lambda^{(s)}(W^{(s,2)})^T = U_P^TU_{A^{(s)}}$  the singular value decomposition of the matrix on the right hand side, and $W_n$ is a sequence of orthogonal matrices such that $U_{P_n}S_{P_n}^{1/2}W_n = \bX_n$, 

Now, observe that the first term can be written as
\begin{align*}
\sqrt{n} \left(\frac{1}{m}\sum_{s=1}^m(A_n^{(s)} - P_n) U_PS_P^{-1/2}W_n \right)_i    = & \sqrt{n} \left(\frac{1}{m}\sum_{s=1}^m(A_n^{(s)} - P_n) \bX_n W_n^TS_P^{-1}W_n \right)_i\\
= & \frac{1}{\sqrt{n}}\sum_{j\neq i} \left(\frac{1}{m}\sum_{s=1}^m(A_{ij}^{(s)} - P_{ij}) X_{j}\right) (nW_n^TS_P^{-1}W_n).
\end{align*}
Conditioning on $X_i=x_i$, the expression is a sum of $n-1$ independent terms with mean zero. The covariance matrix of each term can be calculated as
\begin{align*}
\mathbb{E}\left ( \left(\frac{1}{m}\sum_{s=1}^m(A_{ij}^{(s)} - P_{ij})\right)^2 X_{j}X_j^T\right) 
= & \mathbb{E}\left (\frac{1}{m^2} \mathbb{E}\left( \left. \left(\sum_{s=1}^m(A_{ij}^{(s)} - P_{ij})\right)^2\right| \bX \right) X_{j}X_j^T\right)\\
= & \mathbb{E}\left ( \frac{1}{m^2}\left(mP_{ij} - m^2P^2 + m(m-1)P_{ij}(P_{ij} + \rho(1- P_{ij})) \right) X_{j}X_j^T\right)\\
= & \left(\frac{1-\rho}{m} + \rho\right)\mathbb{E}\left ( P_{ij}(1 - P_{ij}) X_{j}X_j^T\right)\\
= & \widetilde{\Sigma}\left(x_i, \frac{1-\rho}{m} + \rho\right).
\end{align*}
As in Theorem~\ref{thm:rhoCLT}, the classical multivariate central limit theorem gives the desired result.

To prove part b), we follow similar arguments to the proofs of Theorems~\ref{thm:rhoCLT} and \ref{thm:pfdiffgenomni}. Fix some $i\in[n]$, and observe that
\begin{align*}
\sqrt{n} \left(\frac{1}{m}\sum_{s=1}^m \widehat \bX_{M_n}^{(s)} \tilde W_n -  \bX_n \right)_i=   \frac{\sqrt{n}}{m}\sum_{s=1}^m\left((M_n - \widetilde{P}_n) \bZ_n(\tilde W_n^TS_P^{-1}\tilde W_n) \right)_{n(s-1)+i} + o_P(1),
\end{align*}
where $\tilde W_n = V_n^TW_n$ as defined in the proof of Theorem~\ref{thm:omniCLT}. The first term can be written as
\begin{align*}
\frac{{n}^{1/2}}{m} & \sum_{s=1}^m\left((M_n - \widetilde{P}_n) \bZ_n \tilde W_n^TS_P^{-1}\tilde W_n\right)_{n(s-1)+i}\\
 & = \frac{\sqrt{n}}{m}\sum_{s=1}^m\sum_{t=1}^m\frac{1}{m}\left((M_n^{(s, t)} - {P}_n) \bX_n(\tilde W_n^TS_P^{-1}\tilde W_n) \right)_{i} \\
 &=  \frac{\sqrt{n}}{m}\sum_{j\neq i} \sum_{s=1}^m\sum_{t=1}^m\frac{1}{m}\left((M_{ij}^{(s, t)} - {P}_{ij}) X_j \right)(\tilde W_n^TS_P^{-1}\tilde W_n) \\
 &=  \frac{\sqrt{n}}{m}\sum_{j\neq i} \sum_{s, t}\frac{1}{m}\left(\left(\frac{1}{2}(A^{(s)}_{ij} + A_{ij}^{(t)}) - {P}_{ij}\right) X_j \right)(\tilde W_n^TS_P^{-1}\tilde W_n) \\
 &=  \frac{\sqrt{n}}{m}\sum_{j\neq i} \sum_{s, t}\left(\frac{1}{2m}(A^{(s)}_{ij} - P_{ij}) + \frac{1}{2m}(A_{ij}^{(t)} - {P}_{ij})  \right)X_j(\tilde W_n^TS_P^{-1}\tilde W_n) \\
 & =  \frac{\sqrt{n}}{m}\sum_{j\neq i} \left(\frac{1}{2m}\sum_t\sum_{s}(A^{(s)}_{ij} - P_{ij}) + \frac{1}{2m}\sum_s\sum_{t}(A_{ij}^{(t)} - {P}_{ij})  \right)X_j(\tilde W_n^TS_P^{-1}\tilde W_n) \\
 & =   \frac{\sqrt{n}}{m}\sum_{j\neq i} \left(\frac{1}{2}\sum_{s}(A^{(s)}_{ij} - P_{ij}) + \frac{1}{2}\sum_{t}(A_{ij}^{(t)} - {P}_{ij})  \right)X_j(\tilde W_n^TS_P^{-1}\tilde W_n) \\
 & =  \frac{1}{\sqrt{n}}\sum_{j\neq i} \left(\frac{1}{m}\sum_{s=1}^m(A_{ij}^{(s)} - P_{ij}) \right)X_{j}(n\tilde{W}_n^TS_P^{-1} \tilde{W}_n).
\end{align*}
The covariance of each of the terms is calculated in the same way as in part a), and the result follows by similar arguments.
	
\end{proof}

\end{document}